\documentclass[10pt]{article}
\usepackage[margin=1in]{geometry}
\usepackage{amsmath,amssymb,amsthm,mathtools,enumitem,booktabs}
\usepackage{xcolor}
\usepackage{mathrsfs}
\usepackage[colorlinks,linkcolor=blue!55!black,citecolor=blue!55!black]{hyperref}

\newtheorem{theorem}{Theorem}[section]
\newtheorem{proposition}[theorem]{Proposition}
\newtheorem{lemma}[theorem]{Lemma}
\newtheorem{corollary}[theorem]{Corollary}
\theoremstyle{definition}
\newtheorem{definition}[theorem]{Definition}
\theoremstyle{remark}
\newtheorem{remark}[theorem]{Remark}

\DeclareMathOperator{\spec}{Sp}\DeclareMathOperator{\dist}{dist}
\DeclareMathOperator{\Tr}{Tr}\DeclareMathOperator{\tr}{tr}
\DeclareMathOperator{\Ind}{Index}\DeclareMathOperator{\sgn}{sgn}
\DeclareMathOperator{\sig}{sig}\DeclareMathOperator{\supp}{supp}
\DeclareMathOperator{\sfl}{sf}\DeclareMathOperator{\Lag}{Lag}
\DeclareMathOperator{\Herm}{Herm}\DeclareMathOperator{\wind}{wind}
\DeclareMathOperator{\Ran}{Ran}\DeclareMathOperator{\Gr}{Gr}\DeclareMathOperator{\ord}{ord}
\newcommand{\R}{\mathbb{R}}\newcommand{\C}{\mathbb{C}}\newcommand{\Z}{\mathbb{Z}}
\newcommand{\N}{\mathbb{N}}
\newcommand{\Om}{\Omega}\newcommand{\Gam}{\Gamma}\newcommand{\Lm}{\Lambda}
\newcommand{\U}{\mathrm{U}}\newcommand{\Hh}{\widehat H}
\newcommand{\cA}{\mathcal{A}}\newcommand{\cB}{\mathcal{B}}\newcommand{\Vp}{\mathrm{V}}
\newcommand{\cL}{\mathcal{L}}\newcommand{\cR}{\mathcal{R}}
\newcommand{\cD}{\mathcal{D}}
\newcommand{\cH}{\mathcal{H}}
\newcommand{\cU}{\mathcal{U}}\newcommand{\cG}{\mathcal{G}}\newcommand{\cS}{\mathcal{S}}
\newcommand{\cW}{\mathcal{W}}\newcommand{\cT}{\mathcal{T}}
\newcommand{\psid}{\Psi\mathrm{DO}}
\newcommand{\dbar}{\bar\partial}

\newcommand{\ip}[2]{\langle #1,#2\rangle}\newcommand{\norm}[1]{\lVert#1\rVert}
\newcommand{\kp}{\kappa}\newcommand{\Th}{\Theta}
\newcommand{\BDI}{\mathrm{BDI}}

\newcommand{\cN}{\mathcal{N}}
\newcommand{\muRS}{\mathrm{Mas}}
\title{\bf A bulk--edge correspondence
for two-dimensional self-adjoint boundary and interface problems}
\author{Guillaume Bal\thanks{Departments of Statistics and Mathematics and CCAM, University of
Chicago, Chicago, IL 60637; guillaumebal@uchicago.edu}}
\date{\today}

\begin{document}\maketitle

\begin{abstract}\noindent
This paper analyzes the quantized asymmetric transport along the boundary of a two-dimensional topological
insulator occupying a half plane, or along the interface separating two insulators. The Hamiltonian is
a formally self-adjoint elliptic differential operator $H$ of order $p$ acting on
$\C^N$-valued functions, with variable coefficients and a non-characteristic boundary. In contrast with problems posed on the Euclidean plane, terminating a sample requires a choice of self-adjoint
realization, $H_L$, and we show that a quantized transport depends on that choice. The Cauchy-data space $\cB=\C^{pN}$ carries a Hermitian form $\cA$, and  a termination is a family of $\cA$-Lagrangian subspaces indexed by the tangential momentum. A graded polynomial congruence transports them into one compact moduli space $\Lag(\cA)\cong\U(\kp)$, $\kp=pN/2$.

The physical observable modeling the asymmetric transport is the edge, or interface, current $\sigma_I[H_L]=\Tr\,i[H_L,P]\varphi'(H_L)$.
It is shown to be well defined along two routes.  For Shapiro--Lopatinski elliptic realizations, we use estimates from Boutet de Monvel's calculus to show that the current observable is a well-defined trace, that $\mathcal I=2\pi\sigma_I[H_L]$ is a Fredholm index, and that it is stable under deformations and compactly supported changes of the operator and of its termination. For translation invariant operators along the edge, a spectral analysis defines $\mathcal I$ when the edge spectrum escapes a bulk gap at large momenta. $\mathcal I$ is then also a spectral flow and a Maslov index.

Our main theorem is a bulk--edge correspondence for differential terminations
\[
  \mathcal I=\cR-\cR_L+\eta\bigl(L_\infty,\Lm^\Phi(+)\bigr)-\eta\bigl(L_\infty,\Lm^\Phi(-)\bigr),
\]
with $\cR$, a bulk quantity, the rotation of the Calder\'on path over the momentum line, $\cR_L\in\Z$ the winding of
the termination, and the two spectral asymmetries measuring its common limit $L_\infty$
at infinite momentum against the transported endpoints $\Lm^\Phi(\pm)$ of the Calder\'on path.
The right-hand side is computed from symbols by finite-dimensional linear algebra.  We determine the range of $\mathcal I$ over transported pointwise terminations, a window of consecutive values, one per chamber, which differential terminations shift by their winding $\cR_L$; we extend the formula to inadmissible and to pseudodifferential terminations, and, by reflection, to interfaces with arbitrary transmission conditions. At the transparent condition, the invariant equals the bulk-difference invariant of the Euclidean plane. Eight examples illustrate the theoretical findings.\end{abstract}



\section{Introduction}\label{sec:intro}

Since the first experiments displaying the quantum Hall effect \cite{vonKlitzing80} and their
topological interpretation \cite{TKNN82}, the analysis of topological phases of matter has received
continuous interest in the engineering, physics and mathematics literatures; we refer to
\cite{HasanKane10,QiZhang11,BernevigHughes13,ProdanSchulzBaldes16} for general introductions.
Robust asymmetric transport along interfaces separating two-dimensional insulating materials has been
observed or predicted in many areas of applied science, from the integer quantum Hall effect to
photonics and geophysical fluid flows. The robustness of this transport to arbitrarily large
perturbations affords a topological interpretation. An insulator alone displays no large-scale transport.
At an interface between two insulators in different topological phases, the difference of phases is compensated by an asymmetric transport that is guaranteed to exist.
Non-trivial topological phases of matter are by and large wave phenomena,
and therefore appear wherever Hamiltonian dynamics is a reasonable model. In electronics they include, beyond the quantum Hall effect itself, moir\'e structures in multilayer two-dimensional
materials and gated rhombohedral graphene \cite{bistritzer2011moire,FrazierBal}; beyond electronics, photonic crystals and waveguides \cite{HaldaneRaghu08,Ozawa19}, geophysical shallow-water models
\cite{DelplaceMarstonVenaille17,graf2021topology}, and cold plasma models
\cite{fu2021topological}.

\medskip\noindent{\bf Setting.}
We consider differential operators acting on functions of a two-dimensional Euclidean domain, the setting adapted to the description of
large-scale, macroscopic transport features. The half-plane geometry of this paper allows us to model
\emph{termination} of physical samples.  The central object in the analysis of this transport is the interface current observable
\begin{equation}\label{eq:sigmaIintro}
  \sigma_I[H_L]:=\Tr\,i[H_L,P]\,\varphi'(H_L),
\end{equation}
with $H_L$ a self-adjoint realization, $P=P(x)$ a switch function in the coordinate $x$ along the interface, $i[H_L,P]$ the operator of the current through a line $\{x=x_0\}$ transverse to it, and $\varphi'(H_L)$ a density of states supported in a bulk spectral gap. Here $i[H_L,P]$ is the form $i(\ip{H_Lu}{Pv}-\ip{Pu}{H_Lv})$, whose diagonal $v=u$ is the rate of change of $\ip{u}{Pu}$ along $u(t)=e^{-itH_L}u$. For pointwise terminations and smooth $P$ it acts as $i[H,P]$, and \eqref{eq:sigmaIintro} is a trace as written. With tangential derivatives in the termination, the boundary carries part of the current, and \eqref{eq:sigmaIintro} means $\Tr\,\rho(H_L)\,i[H_L,P]\,\rho(H_L)$ with $\varphi'=\rho^2$, the expectation of the current in the state $\varphi'(H_L)$ (Definition~\ref{def:sigmaI}, Remark~\ref{rem:E2current}). The expectation (trace) admits a natural physical interpretation as an asymmetric transport \cite{ProdanSchulzBaldes16}.
For interface Hamiltonians on the Euclidean plane $\R^2$, in which the transition between two bulk phases is implemented by spatially varying coefficients, the observable is quantized, $2\pi\sigma_I[H]\in\Z$, and stable against large classes of perturbations.
A bulk-edge correspondence identifies it with a bulk-difference invariant $\BDI(H_+,H_-)$ depending
only on the two bulk phases, computable by a Fedosov--H\"ormander formula \cite[Ch.~19]{Hormander3}. We refer to
\cite{Bal2019,Bal2022,QuinnBal,BalReview}  as well as \cite{faure2023manifestation} for this program, to \cite{BalYu,FrazierBal25,tauber2020anomalous} for (non-elliptic)
cases where the correspondence fails, and to \cite{JudTauber} for its dependence on boundary conditions. Higher-dimensional extensions \cite{bal2023topological,BernevigHughes13,ProdanSchulzBaldes16} are not considered.

A terminated sample is modeled by posing the operator on a half plane $\Om=\R_x\times\R_{+,y}$.   Throughout, $H$ is a formally self-adjoint elliptic differential operator of order $p\geq1$ on $\Omega$, acting on $\C^N$-valued functions, with smooth bounded coefficients that are $x$-independent outside a bounded set, a non-characteristic boundary and a bulk spectral gap. It must be equipped with a boundary condition, called a {\em termination} $L$ so that $(H,L)$ defines an essentially self-adjoint realization $H_L$. Pointwise terminations impose linear relations among the boundary values $\gamma u(x)$ at each $x$, without tangential derivatives; transported when the Green form depends on the tangential momentum $\xi$, they form a copy of the unitary group $\U(\kp)$ with $\kp=\tfrac12pN$, of real dimension $\kp^2$. More general differential terminations, with tangential derivatives, add a winding in $\xi$. While general nonlocal self-adjoint extensions are not treated, many pseudo-differential terminations are also covered.  For these terminations the current observable exists, is quantized, and is explicitly computable.

\medskip\noindent{\bf Invariant.}
We consider two routes to its definition. One is based on pseudo-differential calculus while the other one applies to translation-invariant operators by direct spectral analysis. In the presence of a boundary with general (pseudo-)differential terminations, the calculus on $\R^2$ used in \cite{Bal2019,Bal2022,QuinnBal} is replaced by the Boutet de Monvel calculus of boundary problems; see  \cite{BdM71} as well as \cite{Grubb96,Schrohe01} and a brief appendix for notation and main results we use in this paper. Along with the Helffer--Sj\"ostrand formula and Combes--Thomas estimates, the a priori elliptic estimates of the calculus allow us to show that $i[H_L,P]\,\varphi'(H_L)$ is trace-class and to identify $\mathcal I:=2\pi\sigma_I[H_L]$ with a Fredholm index, stable under continuous deformations and under changes of the operator and of the termination on a bounded set.

The second route defines the invariant by spectral calculus for  $x$-independent operators. All it requires is that the edge spectrum {\em escape} a bulk gap at large momenta. Escape is implied by ellipticity and is considerably more general, but it yields no stability under $x$-dependent perturbations. It is known in the Euclidean setting that non-elliptic problems may be less stable than elliptic ones \cite{BalYu,FrazierBal25,tauber2020anomalous}. On this route $\mathcal I$ is also a spectral flow of edge branches, a Maslov index, and a winding number of eigenphases. It is in that setting that the bulk--edge correspondence is proved. It applies on the first route as well whenever the data are $x$-independent outside a bounded set, the index then being that of an $x$-independent operator.

\medskip\noindent{\bf Bulk edge correspondence.} The Cauchy data $\gamma u\in\cB=\C^{pN}$ at the boundary carry the Hermitian Green form $\cA(\xi)$ of $H$. A termination is at each $\xi$ an $\cA(\xi)$-Lagrangian subspace $L(\xi)$, and the decaying solutions of
the fiber equation define the Calder\'on Lagrangian $\Lm^+(\xi,E_0)$. For $p\ge2$ the form $\cA(\xi)$ depends on $\xi$. A graded polynomial congruence (Theorem~\ref{thm:graded}) removes that dependence and places every termination and the Calder\'on path in one fixed moduli space
$\Lag(\cA)\cong\U(\kp)$, where both are shown to have limits at infinite momentum. For  every admissible termination differential along $\Gam$ whose edge spectrum escapes the gap at large momentum, our main result, Theorem~\ref{thm:main}, is the
bulk--edge correspondence
\begin{equation}\label{eq:APSintro}
  \mathcal I=\cR-\cR_L+\eta\bigl(L_\infty,\Lm^\Phi(+)\bigr)-\eta\bigl(L_\infty,\Lm^\Phi(-)\bigr).
\end{equation}
Here $\cR$ is the rotation of the Calder\'on path over the momentum line, $L_\infty$ the common limit of the transported termination at $\xi=\pm\infty$ with  $\cR_L\in\Z$ its winding, and the two spectral asymmetries compare $L_\infty$ with the endpoints $\Lm^\Phi(\pm)=\lim_{\xi\to\pm\infty}\Phi(\xi)^{-1}\Lm^+(\xi,E_0)$ of the transported Calder\'on path (Corollary~\ref{cor:graded}); under \textup{[S]}, $\Phi=\mathbf 1$ and $\Lm^\Phi(\pm)=\Lm^+(\pm)$. For transported pointwise terminations,  we have $\cR_L=0$. The right-hand side is evaluated from symbols by finite-dimensional linear algebra and one quadrature, with no operator diagonalized (section~\ref{sec:APS}). The formula extends, with regularized spectral asymmetries, to inadmissible terminations, which lie on the wall bounding the admissible set (Proposition~\ref{prop:wall}), and to some pseudodifferential terminations (Remark~\ref{rem:PDO}). When $p=1$,  \textup{[S]} holds, $\Phi=\mathbf 1$, and admissibility of a pointwise termination is Shapiro--Lopatinski ellipticity, \eqref{eq:APSintro} reads $\mathcal I=\cR+\eta(L,\Lm^+(+))-\eta(L,\Lm^+(-))$ with $\Lm^+(\pm)$ the decaying data of the principal symbol at $\xi=\pm1$, and $\mathcal I$ takes exactly $\kp+1$ consecutive values over admissible pointwise terminations, one on each chamber. 

\medskip\noindent{\bf Range of invariant and interfaces.}
The termination enters \eqref{eq:APSintro} through one point $L_\infty$ of $\U(\kp)$ and one integer $\cR_L$. As $L_\infty$ ranges over the admissible set, $\mathcal I$ takes exactly $\kp-\kp_\cap+1$ consecutive values, one on each chamber, where $\kp_\cap$ measures the relative position of the two endpoints of the Calder\'on path (Theorem~\ref{thm:range}). When the Calder\'on path closes, $\kp_\cap=\kp$ and the termination enters only through $\cR_L$. The latter may shift the invariant arbitrarily. It is bounded by $\kp d$ for differential terminations of order $d$ when the Green form does not depend on $\xi$, and unbounded over pseudodifferential ones (Remarks~\ref{rem:polyrange} and \ref{rem:PDO}).

By reflection, an interface between two media is a half-plane problem for a doubled operator, whose pointwise transmission conditions form $\U(2\kp)$. All the results above apply to the doubled problem. See \cite{KS} for related results. At the transparent transmission condition the invariant equals the bulk-difference invariant of \cite{Bal2022,QuinnBal}, with extensions to sharp interfaces (Theorem~\ref{thm:R2}).

\medskip\noindent{\bf Related work and outline.}
Many results on bulk--edge correspondences exist in the mathematical literature. The closest ones are those for continuum Hamiltonians on the Euclidean plane \cite{Bal2019,Bal2022,QuinnBal,BalReview}, which the current paper generalizes to boundary and interface contributions. The recent work \cite{KS}, which shaped the present approach, treats half-space and interface problems with varying self-adjoint boundary conditions by $K$-theoretic means; see Remark~\ref{rem:KS}. Most bulk--edge correspondences concern terminated samples, for discrete and continuous Hamiltonians, possibly with magnetic contributions or adapted Bloch structures; see, e.g., \cite{Hatsugai93,schulz2000simultaneous,EGS05,Graf2013,ProdanSchulzBaldes16,KellendonkSchulzBaldes04,BourneRennie18,cornean2024orbital,drouot2026bulk}, and \cite{drouot2021microlocal} for microlocal methods. Spectral flow over the self-adjoint boundary conditions of half-line Dirac operators appears in \cite{Thiang21}. See \cite{Gontier23} for second-order operators in a channel. The geometric constructions we use are the Calder\'on
projector \cite{Calderon63} and the Grassmannian of boundary conditions \cite{BBW}; the indices take the form of a Maslov index and its relation to spectral flow \cite{BBF98,Furutani04,Gontier23,Duistermaat,Hormander71}, as for the Bott--Maslov indices counting edge states of periodic tight-binding models \cite{AvilaSBVB13}; and the identification of self-adjoint realizations with a unitary group is von Neumann's parametrization \cite{BHS20,KS}.

The rest of the paper is structured as follows. Our main hypotheses and results are summarized in section~\ref{sec:setting}, and the Green form, the Calder\'on data, their graded transport and their limits at infinite momentum are developed in section~\ref{sec:green}. The proof that the edge observable is well defined may be found in section~\ref{sec:obs}. The bulk-edge correspondence for half spaces is presented in detail in section~\ref{sec:APS} and constraints on the range of the invariant in section~\ref{sec:range}, with a generalization to two glued half-spaces along an interface with arbitrary (elliptic self-adjoint) transmission in section~\ref{sec:interface}. Illustrative examples are worked out in section~\ref{sec:examples}, and Appendix~\ref{app:import} collects background from the Boutet de Monvel calculus and from symplectic geometry.

\paragraph{Notation.} Spatial variables are denoted by  $(x,y)$ with $\Om=\R_x\times\R_{+,y}$ and $\Gam=\partial\Om=\{y=0\}$; dual variables are $(\xi,\zeta)$; $D=\tfrac1i\nabla=(D_x,D_y)$. We say $f$ is a \emph{switch function} (or simply a {\em switch}), $f\in\cS[a,b;c,d]$, when $f:\R\to\R$ is bounded with $f=a$ on
$(-\infty,c]$ and $f=b$ on $(d,\infty)$; $\cS[a,b]$ is the union over $c\le d$ and
$\cS:=\cS[0,1]$, and $C^\infty\cS[\cdot]$ denotes the smooth switch functions. Traces of Cauchy data
are $\gamma_ju:=(D_y^ju)|_\Gam$ and $\gamma u:=(\gamma_0u,\dots,\gamma_{p-1}u)$. Inner products $\ip{v}{u}=\int\overline v\cdot u$ are conjugate linear in the first argument. $\U(\kp)$ is the unitary
group, $\Herm(\kp)$ the Hermitian matrices, $\Gr(j,\C^\kp)$ the complex Grassmannian, and
$\langle\,\cdot\,\rangle=(1+|\cdot|^2)^{1/2}$. The boundary part \eqref{eq:res} of a resolvent is written $\cG(z)$, $G^m$ denotes the filtration \eqref{eq:filtration}, and $\Hh_\sharp$ of section~\ref{sec:interface} is the fiber operator of the glued interface
Hamiltonian while $\Hh^\natural$ appears in Definition~\ref{def:calderon}. We use $\spec$ to denote spectrum.  

\section{Operator hypotheses and main results}\label{sec:setting}

We work on the half plane $\Om=\R_x\times\R_{+,y}$ with boundary $\Gam=\{y=0\}$, and write
\begin{equation}\label{eq:Hsplit}
  H=\sum_{k=0}^{p}A_k(x,y,D_x)\,D_y^k = \mathrm{Op}^w(a^w),
\end{equation}
with $A_k$ a matrix-valued differential operator in $x$ of order at most $p-k$ and Weyl symbol $a_k(x,y,\xi)$ in $x$, and $a^w(x,y;\xi,\zeta)$ the Weyl symbol of $H$. The principal symbol is called $\sigma_p$. The \emph{bulk operator} $H^{\R^2}$ is the operator on $\R^2$ obtained by freezing $a^w$ at $y=y_0$ of
\textup{[H4]} below and extending it constantly in $y$.

\medskip
\noindent\textbf{(i) The volume operator.}
$H$ is a differential operator on $\Om$ of order $p\ge1$ acting on $\C^N$-valued functions, with smooth bounded coefficients with bounded derivatives and equal to $x$-independent coefficients outside a compact set. We assume throughout the paper:
\begin{enumerate}[label=\textup{[H\arabic*]},leftmargin=*,nosep]
\item \emph{formal self-adjointness}: $a^w(x,y;\xi,\zeta)\in\Herm(\C^N)$ for all real arguments, or equivalently $\ip{v}{Hu}=\ip{Hv}{u}$ for $u,v\in C_c^\infty(\Om;\C^N)$;
\item \emph{ellipticity}: the principal symbol $\sigma_p(x,y;\xi,\zeta)$ is invertible for $(\xi,\zeta)\ne0$;
\item \emph{bulk gap}: there is a nonempty open interval $\Delta$, possibly unbounded, with
  $\Delta\cap\spec H^{\R^2}=\emptyset$, equivalently $\Delta\cap\spec a^w(y_0;\xi,\zeta)=\emptyset$ for all real $(\xi,\zeta)$;
\item \emph{asymptotic flatness}: $a^w$ is independent of $(x,y)$ for $y\ge y_0$, so that $H=H^{\R^2}$ there with constant coefficients.
\end{enumerate}
We use the notation $E_0\in\Delta$ for a fixed reference energy and the letter $E$ as a running spectral parameter. We set
\[
  \cB:=\C^{pN},\qquad \kp:=\tfrac12pN.
\]
That $pN$ is even and $\kp$ an integer follows from \textup{[H1]}--\textup{[H2]}, since $\det\sigma_p(x,y;1,\zeta)$ has its $pN$ roots in conjugate pairs with none real (Lemma~\ref{lem:hermODE}). The boundary is non-characteristic: $\sigma_p(x,y;0,1)=A_p(x,y)$, so the coefficient $A_p$ of $D_y^p$ is invertible by \textup{[H2]}. As a restriction, coefficients may vary in $x$ only within the collar $y<y_0$. Although we consider realizations $H_L$ that are {\em not} elliptic, the volume operator always is.

\begin{definition}[Calder\'on Lagrangian]\label{def:calderon}
Fix $x\in\Gam$ and $\xi\in\R$, and let $\Hh(x,\xi):=\sum_ka_k(x,y,\xi)D_y^k$ be the fiber operator on $\R_+$. It retains its $y$-dependence and, by [H4], has constant coefficients for $y\ge y_0$. The operator on $\R$ with these constant coefficients, the fiber at $\xi$ of the bulk operator $H^{\R^2}$, is denoted $\Hh^\natural(\xi)$.
For real $E$, define $\cL:=\Hh(x,\xi)-E$ and
\begin{equation}\label{eq:Lam}
  \Lm^+(x,\xi,E):=\bigl\{\gamma\psi:\cL\psi=0,\ \psi\in L^2(\R_+)\bigr\}.
\end{equation}
\end{definition}

By Theorem~\ref{thm:calderonB} the space $\Lm^+(x,\xi,E)$ has dimension $\kp=\tfrac12pN$ for every $\xi\in\R$ and $E\in\Delta$. By Lemma~\ref{lem:endpoints} it has limits
\[
  \Lm^+(x,\pm,E):=\lim_{\xi\to\pm\infty}\Lm^+(x,\xi,E),
\]
the \emph{endpoints} of the Calder\'on Lagrangian, against which the boundary condition will be tested. They may depend on $E$ in general, although they do not under the condition \textup{[P]} of section~\ref{sec:green}. We write $\Lm^+(x,\pm):=\Lm^+(x,\pm,E_0)$ at the reference energy when the spectral parameter is not displayed.

\medskip\noindent
\textbf{Split type \textup{[S]}.} $H$ is of \emph{split type}, denoted by \textup{[S]}, if $\ord A_k=0$ for $1\le k\le p-1$: the
coefficients of the intermediate normal derivatives are matrices. This holds automatically at $p=1$. \textup{[S]} is not
assumed, and is used only where stated as a convenient condition.

\medskip
\noindent\textbf{(ii) The termination.}

\begin{definition}[boundary conditions and realizations]\label{def:local}
A \emph{boundary condition} for $H$ is a trace operator
\begin{equation}\label{eq:traceop}
  T=\sum_{j<p}B_j\gamma_j,\qquad
  B_j \text{ classical (polyhomogeneous) pseudodifferential along }\Gam,
\end{equation}
of the Boutet de Monvel calculus, without the integral terms of class zero that calculus also allows (Appendix~\ref{app:import}),
with $\kp\times N$ blocks $B_j$, i.e., $\kp$ scalar conditions, whose Weyl symbols $b_j(x,\xi)$ form a $\kp\times pN$ matrix of rank $\kp$ at every point. Its \emph{realization} $H_L$, in the terminology of \cite{Grubb96}, is $H$ acting on
\[
  \cD_L:=\ker T\cap H^p(\Om;\C^N)
\]
and its \emph{field} is
\begin{equation}\label{eq:Lfield}
  L(x,\xi):=\Bigl\{w=(w_0,\dots,w_{p-1})\in\cB=\C^{pN}:\ \sum_{j<p}b_j(x,\xi)\,w_j=0\Bigr\},
\end{equation}
a $\kp$-dimensional subspace of the Cauchy-data space at each $(x,\xi)\in T^*\Gam$.

A boundary condition is {differential of \emph{tangential} order $d$} when the $B_j$ are differential operators of
order $\le d$ along $\Gam$, so that $L(x,\xi)$ is polynomial in $\xi$ of degree $\le d$; such $T$ always meet the requirements of \eqref{eq:traceop}, differential operators being classical with polynomial symbols. This {tangential} order is the one used for boundary conditions throughout (it differs from the Boutet de Monvel order of $T$, which is a graded version; see Appendix~\ref{app:import}). A boundary condition is called \emph{pointwise} when $d=0$.  Then $L(x,\xi)=L(x)$ and
$\cD_L=\bigl\{u\in H^p(\Om;\C^N):\ \gamma u(x)\in L(x)\ \text{for all }x\in\Gam\bigr\}$.
Differential boundary conditions are local operators. A nonlocal boundary condition, with $b_j$
not polynomial in $\xi$, is allowed by \eqref{eq:traceop} and considered in Theorems~\ref{thm:APS}\textup{(ii)} and \ref{thm:fiberexist} and in Remark~\ref{rem:PDO}.
A boundary condition whose realization is essentially self-adjoint on $\cD_L$, $H_L$ then denoting its
closure, is a \emph{termination}; these form the \emph{moduli space of terminations}. Throughout, $L$
denotes a termination, and also its field. At an interface we use the term \emph{transmission condition} (section~\ref{sec:interface}).
The $i$-th row of $T$ has (Boutet de Monvel) \emph{order} $m_i:=\max_j(j+\operatorname{ord}B_{ij})$, and $T$ is \emph{normal}, in the sense of Boutet de Monvel and Grubb \cite[Def.~1.4.3]{Grubb96}, when each row reaches $m_i$ at $j=m_i$, with a coefficient of order zero, and for each $m$ these coefficients in the rows of order $m$ form a matrix of full rank.
\end{definition}

Section~\ref{sec:obs} uses the boundary calculus for $B_j$ differential along $\Gam$. Which fields $L$ qualify is analyzed in section~\ref{sec:green}. As $\kp=\tfrac12pN$, a self-adjoint realization imposes half as many scalar conditions as there are Cauchy data, and at fixed $(x,\xi)$ the possible values of $L$ form a copy of $\U(\kp)$, of real dimension
$\kp^2$ (Theorem~\ref{thm:signature}, Proposition~\ref{prop:U}).

\medskip
\noindent\textbf{The hypotheses on the termination.}
The principal frozen fiber equation $\sigma_p(x,0;\pm1,D_y)\psi=0$ is equivalent to the first-order
system $\partial_y\Psi=M_\infty(x,0;\pm)\Psi$ for $\Psi=\gamma\psi\in\cB$, with $M_\infty$ its
\emph{companion matrix}, whose eigenvalues are the $i\zeta$ over the characteristic roots $\zeta$ of
$\det\sigma_p(x,0;\pm1,\cdot)$. Let $V_\pm(x)\subset\cB$, called the \emph{frozen data}, be the decaying Cauchy data of that problem,
that is the stable subspace of $M_\infty(x,0;\pm)$. They reappear in the proof of
Lemma~\ref{lem:endpoints} as the value at $\epsilon=0$ of rescaled data. With $\Th_t:=\mathrm{diag}(1,t,\dots,t^{p-1})\otimes\mathbf 1_N$ the grading of $\cB$
by order of trace, let the \emph{principal field} of $T$ be
$L^{\downarrow}(x,\pm):=\lim_{t\to\infty}\Th_t^{-1}L(x,\pm t)$, written $L^\downarrow(x)$ when it does not
depend on the sign (e.g.\ for pointwise $L$); when the symbol $\sigma_T(x,\pm)$ at $\xi=\pm1$ of the top-order part of $T$, row by row, has rank $\kp$, $L^\downarrow(x,\pm)=\ker\sigma_T(x,\pm)$ (Appendix~\ref{app:grass}).

\begin{definition}[Shapiro--Lopatinski ellipticity]\label{def:SL}
The pair $(H,L)$ is \emph{Shapiro--Lopatinski elliptic} at $(x,\pm)$ when the frozen data $V_\pm(x)$ and the principal field $L^{\downarrow}(x,\pm)$, two $\kp$-planes, are transverse, i.e.,
\begin{equation}\label{eq:SLdef}
  \cB=V_\pm(x)\oplus L^{\downarrow}(x,\pm),
\end{equation}
and \emph{uniformly} so when, for some boundary condition $T$ defining the domain $\cD_L$, $\sigma_T(x,\pm)$ restricted to $V_\pm(x)$ is invertible with inverse bounded uniformly in $x$. Then $\sigma_T$ has rank $\kp$, $L^\downarrow=\ker\sigma_T$ is continuous in $x$, and the transversality is uniform.
\end{definition}
\noindent This is the condition of \cite[\S20.1]{Hormander3}, equivalently the well-posedness of
\cite{BBW} in Calder\'on-projector form. We pair it with an additional condition on the edge
spectrum.

\begin{definition}[escape]\label{def:escape}
For $H$ and $L$ independent of $x$, let $\Hh_L(\xi)$ be the fiber operator $\Hh(\xi)$ of Definition~\ref{def:calderon} on the domain $\{u\in H^p(\R_+):\gamma u\in L(\xi)\}$, and for $K\subset\Delta$ compact put
$S(K):=\{\xi\in\R:\spec\Hh_L(\xi)\cap K\ne\emptyset\}$. The edge spectrum \emph{escapes} if $S(K)$ is
bounded for every such $K$, i.e., if the edge branches leave every compact subset of the bulk gap as
$|\xi|\to\infty$.
\end{definition}
\medskip\noindent The third hypothesis on the termination, \emph{admissibility}, asks that the limits of
$L$ at infinite tangential momentum be uniformly transverse to the endpoints $\Lm^+(x,\pm)$. It is stated in
Definition~\ref{def:adm}, once the frame in which both are compared has been introduced. Admissibility implies escape when the Calder\'on Lagrangian converges to its endpoints uniformly in energy, while without escape, even in measure, the density of states does not localize (Proposition~\ref{prop:exact}).

\medskip
\noindent\textbf{(iii) The edge observable.} Let $\varphi\in C^\infty\cS[0,1]$ be a smooth non-decreasing switch with $\varphi'\ge0$ supported in $\Delta$, and let $\chi\in C^\infty\cS$ be a smooth spatial switch in the variable $x$ along $\Gam$ while $\Pi=\Pi_{x_0}:=\mathbf 1_{x\ge x_0}$, with $x_0\in\R$ arbitrary, is a sharp switch. We have $\varphi'(H_L)$ as a density of states confined to the gap and introduce the current operator $i[H_L,\chi]$. Contrasting with the Euclidean setting of \cite{Bal2019,Bal2022,QuinnBal}, the termination enters the definition of the commutator to account for boundary effects. As a {\em form} on $\cD_L$,
\[
  \ip{u}{i[H_L,\chi]v}:=i\bigl(\ip{H_Lu}{\chi v}-\ip{\chi u}{H_Lv}\bigr)
  =\ip{u}{i[H,\chi]v}+\ip{\gamma u}{\cA(\chi\gamma v)}_{L^2(\Gam)},
\]
with $i[H,\chi]$ the purely bulk expression and $\cA$ the \emph{Green form}, the operator on $\Gam$ defined by
$\ip{v}{Hu}-\ip{Hv}{u}=i\ip{\gamma v}{\cA\gamma u}_{L^2(\Gam)}$ (Proposition~\ref{prop:green}). The boundary
term vanishes for pointwise terminations, for which $\chi\cD_L\subseteq\cD_L$. It does not vanish for
terminations with tangential derivatives, for which the bulk expression alone is not quantized
(Remark~\ref{rem:E2current}). The form is bounded for the graph norm of $H_L$, so it defines
$i[H_L,\chi]=i(H_L\chi-\chi H_L)$, with $H_L$ extended by duality, as a bounded operator from $\cD(H_L)$
to its antidual $\cD(H_L)^*$, the space of bounded conjugate-linear functionals on $\cD(H_L)$. It takes values in $L^2$ only where the boundary term vanishes, a proper subspace for
tangential derivatives. As $\varphi_k(H_L)$, $\varphi_k\in C_c^\infty(\Delta;\R)$, maps the antidual into
$\cD(H_L)$, the operator $\varphi_1(H_L)\,i[H_L,\chi]\,\varphi_2(H_L)$ is bounded on $L^2$, and we compute its trace.

\begin{definition}[the edge current observable]\label{def:sigmaI}
Let $U=e^{2\pi i\varphi(H_L)}$. When $[\Pi,U]$ is trace-class, set
\begin{equation}\label{eq:Idef}
  \mathcal I:=\Ind\bigl(\Pi U\Pi+(\mathbf 1-\Pi)\bigr)=-\Tr\bigl(U^*[\Pi, U]\bigr).
\end{equation}
Let $\mathcal P:=\cS$ be the set of switches $P(x)$, which contains $\chi$ and $\Pi_{x_0}$. For $P\in\mathcal P$ and $\varphi_1,\varphi_2\in C_c^\infty(\Delta;\R)$ with
$\varphi_1\varphi_2=\varphi'$, let
\begin{equation}\label{eq:TP}
  \sigma_I[H_L]:=\Tr\,\cT_P,\quad \cT_P:=i\bigl(\varphi_1(H_L)\,H_L P\,\varphi_2(H_L)-\varphi_1(H_L)\,P H_L\,\varphi_2(H_L)\bigr),
\end{equation}
assuming in the first definition that $\cT_P$ is trace-class. 

Here, $\varphi(H_L)H_L$ is interpreted as $\psi(H_L)$ for $\psi(E)=E\varphi(E)$.
We have $\ip f{\cT_Pg}=i\bigl(\ip{H_Lu}{Pv}-\ip{Pu}{H_Lv}\bigr)$ with $u={\varphi_1}(H_L)f$ and
$v=\varphi_2(H_L)g$. So $\cT_\chi=\varphi_1(H_L)\,i[H_L,\chi]\,\varphi_2(H_L)$ is the form commutator placed between $\varphi_1(H_L)$ and
$\varphi_2(H_L)$. With $\varphi'=\rho^2$, $\rho\in C_c^\infty(\Delta;\R)$, $\sigma_I=\Tr\,\rho(H_L)\,i[H_L,\chi]\,\rho(H_L)$ is the expectation
of the current with density of states $\varphi'(H_L)$, the current being the form commutator: for
$u\in\cD(H_L)$ and $u_t=e^{-itH_L}u$, $\frac d{dt}\ip{u_t}{\chi u_t}=i(\ip{H_Lu_t}{\chi u_t}-\ip{\chi u_t}{H_Lu_t})$ \cite{ProdanSchulzBaldes16}. For $v\in\cD(H_L)$, $i[H_L,\chi]v\in L^2$ if and only if $\chi v\in\cD(H_L)$: if the form equals $\ip ug$ for all
$u\in\cD(H_L)$, then $\ip{H_Lu}{\chi v}=\ip u{\chi H_Lv-ig}$, so $\chi v\in\cD(H_L^*)=\cD(H_L)$, and conversely.
Hence $i[H_L,\chi]\varphi_2(H_L)$ is $L^2$-valued if and only if $\chi\Ran\varphi_2(H_L)\subseteq\cD(H_L)$, as for
pointwise terminations, where $\chi\cD_L\subseteq\cD_L$. It then equals the bulk expression $i[H,\chi]\varphi_2(H_L)$,
trace-class as in Theorem~\ref{thm:traceclass}, and cyclicity gives
\begin{equation}\label{eq:sigmaI}
  \sigma_I[H_L]=\Tr\,i[H_L,\chi]\,\varphi'(H_L).
\end{equation}
\end{definition}
\begin{remark}\label{rem:E2current}
For $H=-\Delta$ on $y>0$ with termination $T=\gamma_0+iD_x\gamma_1$, $\cA=\bigl(\begin{smallmatrix}0&1\\1&0\end{smallmatrix}\bigr)$ and $\gamma_0=-\partial_x\gamma_1$ on $\cD_L$, so that
$i[H_L,\chi]=i[H,\chi]+\gamma_1^*\chi'\gamma_1$ as forms on $\cD_L$, with $\gamma_1^*:L^2(\Gam)\to\cD(H_L)^*$ the adjoint of the trace
$\gamma_1$, a layer supported on $\Gam$. Here $i[H,\chi]\varphi'(H_L)$ is trace-class, as in Theorem~\ref{thm:traceclass}, but
$i[H_L,\chi]\varphi'(H_L)$ is not $L^2$-valued, a fortiori {\em not} trace-class: for $0\ne v\in\Ran\varphi'(H_L)$,
$T(\chi v)=\chi'\gamma_1v\ne0$, $\gamma_1v$ being real-analytic in $x$, so that $\chi v\notin\cD(H_L)$. Thus
\eqref{eq:sigmaI} is not available, and \eqref{eq:TP} defines the current.
We verify that $\bigl(|\chi'|^{1/2}\gamma_1\overline{\varphi_1}(H_L)\bigr)^*\sgn(\chi')\,|\chi'|^{1/2}\gamma_1\varphi_2(H_L)$ is
indeed trace-class, a product of Hilbert--Schmidt operators. 

The edge states are $\psi=e^{-y/\xi}$, $\xi>0$, with $E=\xi^2-\xi^{-2}$ and $E'=2\xi+2\xi^{-3}$. The bulk
expression contributes the velocity $\ip{\psi}{\partial_\xi\Hh\,\psi}/\norm\psi^2=2\xi$, and the boundary
term of the form commutator the rest, $|\partial_y\psi(0)|^2/\norm\psi^2=2\xi^{-3}$. The bulk
trace $\Tr\,i[H,\chi]\varphi'(H_L)$ gives an incorrect $\int\varphi'(E)\,\xi^4(1+\xi^4)^{-1}dE<1$, while $\mathcal I=2\pi\sigma_I[H_L]=1$. The boundary part of $\cT_\chi$ carries the missing $\int\varphi'(E)(1+\xi^4)^{-1}dE$.
\end{remark}
\noindent Uniform Shapiro--Lopatinski ellipticity and, for $x$-independent data, escape are the two routes
to the existence of $\mathcal I$ (Theorems~\ref{thm:traceclass}, \ref{thm:quantisation} and
\ref{thm:fiberexist}); stability for $x$-dependent data is considered only for elliptic realizations as we use there the Boutet de Monvel calculus recalled in the appendix (Remark~\ref{rem:escape}).
\medskip

\noindent\textbf{(iv) Main results and their hypotheses.} \textup{[H1]--[H4]} are assumed throughout. The
invariant is defined along two routes: the \emph{elliptic route}, for uniformly Shapiro--Lopatinski elliptic
terminations with $B_j$ differential along $\Gam$; and the \emph{spectral route}, for
$x$-independent data with $L$ real-analytic in $\xi$ and edge spectrum escaping the gap.
\begin{itemize}[leftmargin=*,nosep]
\item \emph{Quantization.} On either route, $[\Pi,U]$ is trace-class, $\mathcal I$ is the index of $\Pi U\Pi$,
and $2\pi\sigma_I=\mathcal I$ (Theorems~\ref{thm:quantisation} and \ref{thm:fiberexist}).
\item \emph{Stability and locality.} On the elliptic route, $\mathcal I$ does not depend on the switches, is
constant along continuous paths of realizations, and is unchanged by compactly supported perturbations of the
operator and of the termination (Corollary~\ref{cor:stab}, Theorem~\ref{thm:local}). For $x$-independent data, or at $p=1$ for pointwise terminations, when the principal symbol does not depend on $y$ on the collar, the $y$-dependence of the lower-order terms there does not change $\mathcal I$, once $L$ is transported accordingly (Theorem~\ref{thm:collar}).
\item \emph{Spectral flow.} On the spectral route, $\mathcal I$ is the spectral flow of the edge branches
across $E_0$, a Maslov index, and a count of eigenphase passages through $2\pi\Z$ (Theorem~\ref{thm:three}).
\item \emph{Bulk--edge correspondence.} On the spectral route, $\mathcal I$ is a signed count of eigenphase passages of the relative unitary $U_L^{-1}U_+$ of termination and Calder\'on Lagrangian, equal at any momentum cutoff to a winding plus two spectral asymmetries, and also at infinite momentum when that unitary has limits
(Theorem~\ref{thm:APS}). The asymmetries are regularized for inadmissible terminations
(Proposition~\ref{prop:wall}), and the formula becomes \eqref{eq:APSintro} for admissible differential terminations
(Theorem~\ref{thm:main}). The endpoints are principal-symbol data under the weight condition \textup{[P]} of section~\ref{sec:green}, and
$\cR$ is a bulk quantity up to endpoint data (Corollary~\ref{cor:bulkR}).
\item \emph{Range.} Over admissible transported pointwise terminations with escaping edge spectrum,
$\mathcal I$ takes exactly $\kp-\kp_\cap+1$ values (Theorem~\ref{thm:range}), $\kp_\cap$ being the dimension of the intersection of the two endpoints (Proposition~\ref{prop:admtopology}). Differential terminations
shift them by any integer $\cR_L$ (Remark~\ref{rem:polyrange}).
\item \emph{Interfaces.} An interface is a half-space problem for a doubled operator
(Theorem~\ref{thm:double}). Under the interface hypotheses \textup{[C1]--[C3]} of section~\ref{sec:interface} and a transversality \eqref{eq:C2frozen} of the frozen principal data of the two media, the transparent condition $L_{\mathrm{tr}}$, continuity of the Cauchy data across the interface, gives
$\BDI$ (Theorem~\ref{thm:R2}). At $p=1$, for $y$-independent media with a common principal symbol, $\mathcal I(L_{\mathrm{tr}})=\cR[H_+]-\cR[H_-]$, so that $\BDI$ is a difference of two half-plane rotations (Corollary~\ref{cor:etaglue}).
\end{itemize}

\section{The Green form and the Calder\'on Lagrangian}\label{sec:green}
This section converts the boundary problem into finite-dimensional geometry. The coefficients of $H$ on $\Gam$ assemble into a
Hermitian form $\cA$ on the Cauchy data (Proposition~\ref{prop:green}), of signature zero
(Theorem~\ref{thm:signature}), whose Lagrangian subspaces form a copy of $\U(\kp)$
(Proposition~\ref{prop:U}); Proposition~\ref{prop:SA} says when a Lagrangian field defines a termination. A congruence removes the dependence of $\cA$ on the tangential momentum
(Theorem~\ref{thm:graded}). All Lagrangian data are read after this transport
(Definition~\ref{def:transported}). The operator contributes the Calder\'on Lagrangian $\Lm^+(\xi,E)$
(Theorem~\ref{thm:calderonB}), real-analytic in $(\xi,E)$ (Lemma~\ref{lem:analytic}), with limits at
$\xi=\pm\infty$ (Lemma~\ref{lem:endpoints}), explicit under the weight conditions \textup{[P]} or \textup{[M]}. Its image in $\U(\kp)$ is called the Calder\'on path (Definition~\ref{def:Uplus}). A
termination is admissible when its limits at infinite momentum are transverse to the endpoints of that
path (Definition~\ref{def:adm}). We mention under which condition admissibility of a pointwise termination implies uniform Shapiro--Lopatinski ellipticity (Remark~\ref{rem:SL}). The topology of the admissible set is described in Proposition~\ref{prop:admtopology}. Finally,
we show that the edge spectrum reduces to one determinant on $\U(\kp)$ (Theorem~\ref{thm:edge}), which makes the
computation of edge invariants a problem of finite-dimensional linear algebra.

\paragraph{The Green form.}\label{ssec:green0}

We assume \textup{[H1]--[H2]} throughout the section, \textup{[H1]--[H4]}  starting with Theorem~\ref{thm:calderonB}, and will use the weight conditions \textup{[P]} and \textup{[M]} introduced before Lemma~\ref{lem:endpoints} in some results. Lemma~\ref{lem:essspec} and Theorem~\ref{thm:edge} consider in addition that $H$ and $L$ are $x$-independent; elsewhere $x\in\Gam$ enters as a parameter, often not displayed.
\begin{proposition}\label{prop:green}
For $H$  in \eqref{eq:Hsplit}, there is a unique $\cA$, formally self-adjoint as an operator on
$\Gam$, such that
\begin{equation}\label{eq:green}
  \ip{v}{Hu}-\ip{Hv}{u}=i\,\ip{\gamma v}{\cA\,\gamma u}_{L^2(\Gam)},\qquad u,v\in H^p(\Om;\C^N).
\end{equation}
Throughout, $\cA(x,\xi)$ denotes the \emph{Weyl} symbol of $\cA$, a smooth family of Hermitian matrices on
$\cB$:
\[
  \cA(x,\xi)\in\Herm(\cB).
\]
In the grading of $\cB$ by order of trace, its blocks $\cA_{jl}$, $0\le j,l\le p-1$, operators along $\Gam$ with Weyl symbols $\cA_{jl}(x,\xi)$, satisfy the following:
\begin{enumerate}[label=\textup{(\roman*)},nosep]
\item with $m:=k-1-j-l$,
\begin{equation}\label{eq:blocks}
  \cA_{jl}=\sum_{k\ge j+l+1}i^{m}\binom{k-1-l}{m}\,\bigl(\partial_y^{\,m}A_k\bigr)\big|_\Gam ,
\end{equation}
the term $k=j+l+1$ carrying no normal derivative and coefficient $1$, and the same identity holds for the Weyl symbols, with $a_k$ in place of $A_k$;
\item $\ord\cA_{jl}\le p-1-j-l$, and $\cA_{jl}=0$ whenever $j+l\ge p$, the sum in \eqref{eq:blocks} being empty;
\item $\cA_{j,p-1-j}=A_p|_{y=0}$, of order zero and invertible by \textup{[H2]};
\item $H\mapsto\cA$ is linear in the coefficients of $H$ and their normal derivatives. If
$\cA'(x,\xi)\in\Herm(\cB)$ vanishes on the blocks $j+l\ge p-1$  then
\begin{equation}\label{eq:sigstable}
  \sig\bigl(\cA(x,\xi)+t\cA'(x,\xi)\bigr)=\sig\cA(x,\xi)\qquad\text{for every }t\in\R .
\end{equation}
\end{enumerate}
Finally
\begin{equation}\label{eq:detA}
  \det\cA(x,\xi)=\pm\bigl(\det A_p(x,0)\bigr)^p\ne0
\end{equation}
independently of $\xi$. We write $q(w):=w^*\cA(x,\xi)w$ and $\omega(v,w):=i\ip{v}{\cA(x,\xi)w}$ on $\cB$, and, where $(x,\xi)$ is fixed, $\cA$ also for the matrix $\cA(x,\xi)$.  A typical example of $\cA'(x,\xi)$ is the Weyl symbol of $\cA_\Vp$, defined as in \eqref{eq:green} with $H$ replaced by a perturbation $\Vp=\Vp^*$ of order $\le p-1$.
\end{proposition}
\begin{proof}
Existence of the Green form, and its triangular structure with respect to the anti-diagonal, are classical \cite[(1.3.10)--(1.3.14)]{Grubb96}; we derive the blocks explicitly. It suffices to treat $u,v\in C^\infty_c(\overline\Om)$, dense in $H^p$, both sides of \eqref{eq:green} being continuous on $H^p$ by (ii), since $\cA_{jl}$ maps $H^{p-l-1/2}(\Gam)$ to $H^{j+1/2}(\Gam)$. Integrating by parts $p$ times in $y$ produces the boundary term, while the tangential operators $A_k$ are moved across the pairing by taking their adjoints in $L^2(\Gam)$, which produces no boundary term on $\Gam=\R$. This leaves a sesquilinear form in $(\gamma u,\gamma v)$ whose entries are operators along $\Gam$ built from the $A_k$ and their normal derivatives at $y=0$.
Uniqueness holds because $\gamma$ is onto, and $\cA^*=\cA$ because \textup{[H1]} makes the left side
of \eqref{eq:green} anti-Hermitian.

By induction on $k$, $\ip{f}{D_y^kg}-\ip{D_y^kf}{g}=i\sum_{j'+l=k-1}\ip{\gamma_{j'}f}{\gamma_lg}_{L^2(\Gam)}$, so $A_kD_y^k$ contributes the boundary pairing
$i\sum_{j'+l=k-1}\ip{\gamma_{j'}(A_k^*v)}{\gamma_lu}$, the adjoint $A_k^*$ in $L^2(\Gam)$ having been moved to $v$.
The Leibniz rule for $D_y=-i\partial_y$ gives $\gamma_{j'}(A_k^*v)=\sum_{m\le j'}\binom{j'}{m}(-i)^m(\partial_y^mA_k^*)|_\Gam\,\gamma_{j'-m}v$,
and moving $\bigl((-i)^m\partial_y^mA_k^*\bigr)^*=i^m\partial_y^mA_k$ to the right of the
pairing distributes the term over the blocks $(j'-m,l)$ with coefficient $i^m\binom{j'}{m}\,\partial_y^mA_k|_\Gam$. Since $j'=k-1-l$, the block $(j,l)$ receives a contribution from $k$ exactly
when $m:=k-1-j-l$ satisfies $0\le m\le j'$, that is exactly when $k\ge j+l+1$, and then only from
$\partial_y^{\,m}A_k|_\Gam$, with coefficient $i^m\binom{k-1-l}{m}$. Collecting these gives \eqref{eq:blocks} and (i).

This also gives (ii), since $\ord\partial_y^mA_k\le\ord A_k\le p-k\le p-1-j-l$ for every $k\ge j+l+1$;
and it gives (iii), since $j+l=p-1$ together with $j+l+1\le k\le p$ forces $k=p$ and $m=0$, leaving
$\cA_{j,p-1-j}=A_p|_{y=0}$, a matrix, invertible by \textup{[H2]}.

For \eqref{eq:detA}, let $J$ be the block flip $(Jw)_l=w_{p-1-l}$, a permutation with $\det J=\pm1$.
Then $(\cA(x,\xi)J)_{jl}=\cA_{j,p-1-l}(x,\xi)$ vanishes when $j>l$, by (ii), so $\cA(x,\xi)J$ is block upper-triangular
with diagonal blocks $A_p(x,0)$ by (iii); hence
$\det\cA(x,\xi)=\pm\det(\cA(x,\xi)J)=\pm(\det A_p(x,0))^p\ne0$ by \textup{[H2]}, for every $\xi$. {The determinant depends on the anti-diagonal alone}, leading to (iv) as follows: $(\cA+t\cA')(x,\xi)$ is Hermitian, has the anti-diagonal blocks of $\cA(x,\xi)$ and vanishes on the blocks $j+l\ge p$, so
it is nondegenerate with $t$-independent determinant, and a connected family of nondegenerate Hermitian
forms has constant signature. Linearity in (iv) is clear from \eqref{eq:blocks}, and for $\cA_\Vp$ only $k\le p-1$ contribute, so its blocks vanish for $j+l\ge p-1$.
\end{proof}

\begin{theorem}[vanishing signature]\label{thm:signature}
For every $(x,\xi)$, $\sig\cA(x,\xi)=0$. Equivalently $\cA(x,\xi)$ admits maximal $q$-isotropic subspaces
of dimension $\kp$, so that
\[
  \Lag\bigl(\cA(x,\xi)\bigr)\ne\emptyset .
\]
\end{theorem}
\begin{proof}
We reduce to the anti-diagonal and count. (a) By Proposition~\ref{prop:green}\textup{(ii)--(iii)}, $\cA=\cA_0+\cA'$ with $\cA_0$ the block anti-diagonal matrix with blocks $\cA_{j,p-1-j}=A_p(x,0)$ and $\cA'$ Hermitian, vanishing on the blocks $j+l\ge p-1$; so $\sig\cA=\sig\cA_0$ by \eqref{eq:sigstable}. (b) In $\cA_0$, each pair of slots $\{j,p-1-j\}$ with $j\ne p-1-j$ carries $\bigl(\begin{smallmatrix}0&A_p\\A_p&0\end{smallmatrix}\bigr)$, whose eigenvalues are $\pm$ those of $A_p$, of signature zero. For even $p$ these pairs exhaust the slots, and $\sig\cA_0=0$. (c) For odd $p$ the middle slot $j=\tfrac{p-1}2$ carries $A_p(x,0)=\sigma_p(x,0;0,1)$ itself. By homogeneity of odd degree, $-A_p(x,0)=\sigma_p(x,0;0,-1)$, and $\sigma_p(x,0;\omega)$ is invertible and Hermitian on the unit circle $|\omega|=1$ by \textup{[H1]--[H2]}, so $A_p$ and $-A_p$ have the same signature, which is therefore zero. Finally, a nondegenerate Hermitian form of signature zero on $\C^{2\kp}$ has maximal isotropic subspaces of dimension $\kp$.
\end{proof}

\noindent The Cauchy data of decaying solutions are Lagrangian by the following classical facts,
applied below to the principal symbol at $\xi=\pm1$ and to the bulk fibers $\Hh^\natural(\xi)-E$.

\begin{lemma}[constant-coefficient Hermitian problems]\label{lem:hermODE}
Let $P(\zeta)=\sum_{k\le p}P_k\zeta^k$ be an $N\times N$ matrix polynomial, Hermitian for real $\zeta$, with
$P_p$ invertible and $\det P$ without real roots. Then:
\begin{enumerate}[label=\textup{(\alph*)},nosep]
\item $\det P$ has $\kp=\tfrac12pN$ roots in each open half plane, so that $pN$ is even;
\item the Cauchy data $\Lm^\pm$ of the solutions of $P(D_y)\psi=0$ decaying as $y\to\pm\infty$ are the
spectral subspaces of the companion matrix $M$ for the roots $\zeta$ in the two half planes, complementary and of
dimension $\kp$;
\item $\Lm^\pm$ are Lagrangian for the Green form $\cA_P$ of $P(D_y)$, whose blocks are the coefficients
of $\zeta^{j+l+1}$ in $P$ by \eqref{eq:blocks}.
\end{enumerate}
\end{lemma}
\begin{proof}
(a) As $P(\zeta)$ is Hermitian for real $\zeta$, $\det P$ is a polynomial real on $\R$, hence with real coefficients, of degree $pN$ since $P_p$ is invertible. Having no real roots, its roots come in conjugate pairs, so each open half plane contains $\tfrac12pN=\kp$ of them, counted with multiplicity.
(b) With $\Psi(y)=(\psi,D_y\psi,\dots,D_y^{p-1}\psi)(y)$, the equation $P(D_y)\psi=0$ reads $\partial_y\Psi=M\Psi$ with $M$ the companion matrix, whose eigenvalues are the $i\zeta$ over the roots $\zeta$ of $\det P$. So $\Psi(y)=e^{yM}\Psi(0)$, and since $e^{i\zeta y}$ decays as $y\to+\infty$ exactly when $\Im\zeta>0$, and as $y\to-\infty$ when $\Im\zeta<0$, $\Lm^\pm$ are the spectral subspaces of $M$ for $\pm\Im\zeta>0$: each of dimension $\kp$ by (a), and complementary since $M$ has no imaginary eigenvalue.
(c) For two solutions $\psi_1,\psi_2$, \eqref{eq:green} on an interval $[y_1,y_2]$ gives $i\ip{\Psi_2}{\cA_P\Psi_1}\big|_{y_1}^{y_2}=\ip{\psi_2}{P(D_y)\psi_1}-\ip{P(D_y)\psi_2}{\psi_1}=0$, so the pairing $\ip{\Psi_2(y)}{\cA_P\Psi_1(y)}$ does not depend on $y$. If both solutions decay as $y\to+\infty$, it vanishes, so $\Lm^+$ is isotropic; likewise $\Lm^-$ as $y\to-\infty$. As $\cA_P$ is nondegenerate by \eqref{eq:detA}, an isotropic subspace of dimension $\kp=\tfrac12\dim\cB$ is Lagrangian.
\end{proof}

The next proposition identifies the Lagrangians of $\cA$ with $\U(\kp)$. Diagonalize
\[
\cA=\Upsilon^*\,\mathrm{diag}(D_+,-D_-)\,\Upsilon
\]
with $\Upsilon$ unitary and fixed once and for all, $D_+=\mathrm{diag}(a_1,\dots,a_\kp)$ the positive eigenvalues and $D_-=\mathrm{diag}(b_\kp,\dots,b_1)$ the moduli of the negative ones, $\kp$ of each by Theorem~\ref{thm:signature}; the two lists need not agree, and we set $\varkappa_\cA:=(\det D_+/\det D_-)^{1/2}>0$. Write $\alpha=\Upsilon w=(\alpha_+,\alpha_-)$, so that
$q(w)=\alpha_+^*D_+\alpha_+-\alpha_-^*D_-\alpha_-$.
\begin{proposition}\label{prop:U}
Every $L\in\Lag(\cA)$ is the graph $L=\{(B\alpha_-,\alpha_-):\alpha_-\in\C^\kp\}$ over the $q$-\emph{negative} block of a unique
$B$ with $B^*D_+B=D_-$, and
\begin{equation}\label{eq:UL}
  L\ \longmapsto\ U_L:=D_+^{1/2}B\,D_-^{-1/2}
\end{equation}
is a real-analytic diffeomorphism $\Lag(\cA)\to\U(\kp)$; in particular $\dim_\R\Lag(\cA)=\kp^2$. For $L,\Lm\in\Lag(\cA)$,
\begin{equation}\label{eq:transdet}
  \dim(L\cap\Lm)=\dim\ker\bigl(U_L-U_\Lm\bigr),\qquad
  L\pitchfork\Lm\iff\det\bigl(U_L-U_\Lm\bigr)\ne0 ,
\end{equation}
and
\begin{equation}\label{eq:detUL}
  \det U_L=\varkappa_\cA\,\det B , \qquad |\det B|=\varkappa_\cA^{-1}.
\end{equation}
\end{proposition}

\begin{proof}
This is von Neumann's parametrization in graph coordinates \cite[Ch.~1--2]{BHS20}, \cite[\S1]{Furutani04}, \cite[Ch.~12, 15]{BBW}; the graph is taken over the $q$-negative block, which fixes the sign conventions of the paper, the chart over the positive block being $L\mapsto U_L^{-1}$. If $(\alpha_+,0)\in L$, then $q=\alpha_+^*D_+\alpha_+=0$ gives $\alpha_+=0$, so $L$ projects bijectively onto the second factor, and $L$ is the graph of a unique $B$. By polarization, isotropy is $B^*D_+B=D_-$, that is $U_L^*U_L=\mathbf 1$; conversely every unitary $U$ gives such a $B=D_+^{-1/2}UD_-^{1/2}$, whose graph is isotropic of dimension $\kp$, hence Lagrangian. For \eqref{eq:transdet}, $(B\alpha,\alpha)=(B'\alpha',\alpha')$ forces $\alpha=\alpha'$, so $L\cap\Lm\cong\ker(B-B')$, and $U_L-U_\Lm=D_+^{1/2}(B-B')D_-^{-1/2}$. Reordering the eigenvectors within each block replaces $U_L$ by $P_+U_LP_-^{-1}$ with fixed permutations $P_\pm$; this conjugates every $U_L^{-1}U_\Lm$ and multiplies $\det U_L$ by a constant sign, so nothing defined below from relative unitaries or from $\arg\det U_L$ depends on it.
\end{proof}

\noindent We now briefly explore when the Lagrangian field of a boundary condition defines a termination as well as the converse statement.
\begin{proposition}\label{prop:SA}
Let $T$ be a boundary condition with field $L(x,\xi)\in\Lag(\cA(x,\xi))$.
\begin{enumerate}[label=\textup{(\roman*)},nosep]
\item If $\cA$ and $L$ are both independent of $\xi$, as for pointwise $L$ under \textup{[S]}, or both independent of $x$, then $H_L$ is symmetric, and self-adjoint on $\cD_L$ when $(H,L)$ is uniformly Shapiro--Lopatinski elliptic. For $x$-independent data it is essentially self-adjoint without ellipticity, its closure being the direct integral of the $\Hh_L(\xi)$. Conversely, for $x$-independent $H$ and $B_j$, every termination has a Lagrangian field.
\item If $\cA$ is constant and $\cD_L=\{u:\alpha_+(\gamma u)=\mathfrak B\,\alpha_-(\gamma u)\}$ for an operator $\mathfrak B$ along $\Gam$, then $H_L$ is symmetric if and only if $D_+^{1/2}\mathfrak B D_-^{-1/2}$ is isometric in $L^2(\Gam;\C^\kp)$ on $D_-^{1/2}\alpha_-(\gamma\cD_L)$. If moreover $T$ is joined to an $x$-independent condition by an affine path $T_t$ of such symmetric conditions, uniformly Shapiro--Lopatinski elliptic with uniformly bounded right inverses $R_{T_t}$, then $H_L$ is self-adjoint.
\end{enumerate}
\end{proposition}

\begin{proof}
(i) The right-hand side of \eqref{eq:green} vanishes on $\cD_L\times\cD_L$: pointwise in $x$ in the first case, $\cA(x)$ being a matrix and $\gamma u(x),\gamma v(x)\in L(x)$; after Fourier transform in $x$ in the second, since $\gamma\widehat u(\xi),\gamma\widehat v(\xi)\in L(\xi)$. Under uniform ellipticity $\sigma_T$ has rank $\kp$ with $L^\downarrow$ continuous, so the rows of a pointwise $T$ may be chosen adapted to the filtration with leading coefficients of full rank, i.e., normal (Definition~\ref{def:local}), and the adjoint of $H_L$ is the realization of the adjoint condition \cite[Thms.~1.6.9, 1.6.11]{Grubb96}, whose field is $L^{\perp_q}=L$. In the $x$-independent elliptic case, self-adjointness on $\cD_L$ follows from essential self-adjointness, shown next, and the closedness of $H_L$ on $\cD_L$ given by Lemma~\ref{lem:ellreg}\textup{(i)}. In the $x$-independent case each fiber $\Hh_L(\xi)$ is self-adjoint on $\{u\in H^p(\R_+):\gamma u\in L(\xi)\}$, its maximal domain being $H^p(\R_+)$ \cite[Ch.~2--3]{BHS20}, and the functions of $\cD_L$ with $\widehat u$ smooth and compactly supported in $\xi$ form a core of the direct integral. Conversely, symmetry forces isotropy of $L(\xi)$ for a.e.\ $\xi$, hence for all $\xi$ by continuity, and the rank condition gives $\dim L(\xi)=\kp$.

(ii) In the coordinates $\alpha$, the right-hand side of \eqref{eq:green} is $i\ip{\alpha_-v}{(\mathfrak B^*D_+\mathfrak B-D_-)\alpha_-u}$, which vanishes on $\cD_L\times\cD_L$ exactly under the stated isometry. For $z\notin\R$, Lemma~\ref{lem:ellreg}\textup{(i)} below, the bound $\norm{u_0}\le|\Im z|^{-1}\norm{(H-z)u_0}$ for $u_0:=u-R_{T_t}T_tu\in\cD_{L_t}$, and $R_{T_t}$ give $\norm u_{H^p}\le C\bigl(\norm{(H-z)u}+\norm{T_tu}\bigr)$ uniformly in $t$. The system $\{H-z,T_t\}$, affine in $t$, is onto at the $x$-independent end by \textup{(i)} and $R_{T_t}$, hence at $T$ by the method of continuity \cite[Thm.~5.2]{GilbargTrudinger}, so $H_L-z$ is onto.
\end{proof}

\noindent Outside the cases of Proposition~\ref{prop:SA}\textup{(i)}, symmetry may be checked with \textup{(ii)}. For $H=-\Delta$, $\cA=\bigl(\begin{smallmatrix}0&1\\1&0\end{smallmatrix}\bigr)$ is constant, with $\alpha_\pm=(\gamma_0\pm\gamma_1)/\sqrt2$ and $D_\pm=1$. The symmetrized termination
\[
  T=\gamma_0+iB\gamma_1,\qquad B:=\tfrac12\bigl(b(x)D_x+D_xb(x)\bigr),\qquad b\in C_b^\infty(\R;\R),
\]
reads $\alpha_+=\mathfrak B\alpha_-$ with $\mathfrak B=-(1-iB)(1+iB)^{-1}$, the unitary Cayley transform of self-adjoint $B$ on $L^2(\Gam)$. So $H_L$ is symmetric although $L(x,\xi)=\ker\bigl(1,\ ib(x)\xi\bigr)$ depends on both variables, whereas $\gamma_0+ib(x)D_x\gamma_1$ is not symmetric when $b'\ne0$. When $|b|\ge b_0>0$, the top-order part of the row is $iB\gamma_1$, so $L^\downarrow=\C(1,0)$, transverse to $V_\pm=\C(1,i)$, and $\sigma_T(x,\pm)(1,i)=\mp b(x)$ has modulus $\ge b_0$: the realization is uniformly Shapiro--Lopatinski elliptic. Moreover $T$ has a bounded right inverse $R_T:H^{-1/2}(\Gam)\to H^2(\Om)$, $TR_T=\mathbf 1$, constructed for such rows in Appendix~\ref{app:import}. As $b$ has constant sign, $b_t:=(1-t)b+t\,b(0)$ keeps $|b_t|\ge b_0$, so the right inverses $R_{T_t}$ are uniformly bounded, and Proposition~\ref{prop:SA}\textup{(ii)} shows that $H_L$ is self-adjoint. At $b\equiv1$ this is the $x$-independent termination of Remark~\ref{rem:E2current}, to which Proposition~\ref{prop:SA}\textup{(i)} applies directly.

\medskip

So $\U(\kp)$ is the whole moduli space of pointwise terminations under \textup{[S]}, and of transported pointwise terminations in general (Corollary~\ref{cor:graded}). For $U_0\in\U(\kp)$ define
\begin{equation}\label{eq:wall}
  \Sigma_{U_0}:=\bigl\{U\in\U(\kp):\det(U-U_0)=0\bigr\}
  =\bigl\{U:\ 1\in\spec\bigl(U_0^{-1}U\bigr)\bigr\} ,
\end{equation}
the set of $U$ not transverse to $U_0$ called the \emph{wall} of $U_0$ or the \emph{Maslov cycle}.
By \eqref{eq:transdet} it is exactly the locus where the corresponding Lagrangian meets $L_{U_0}$. It
is a real hypersurface of $\U(\kp)$, whose points with $\dim\ker(U-U_0)=j$ form a set of codimension $j^2$ (Appendix~\ref{app:import}).

\paragraph{Trivializing the fiber.}\label{ssec:trivial}

The chart \eqref{eq:UL} lives in one fiber. However, for $p\ge2$ the fiber may move as $\cA(\xi)$ depends on $\xi$ a priori. As a central step of the paper, we now show that the dependence can always be removed by a congruence since $\cA(\xi)$ is polynomial (at $x$ fixed). We introduce the grading $\cB=\bigoplus_{j<p}\cB_j$ by order of trace of Proposition~\ref{prop:green}, $\cB_j\cong\C^N$ being the slot of $\gamma_j$, and define the filtrations we will be using:
\begin{equation}\label{eq:filtration}
  F_k:=\bigoplus_{j\ge k}\cB_j,\qquad G^m:=\bigoplus_{j\le m}\cB_j,\qquad
  \cB_{\mathrm{low}}:=G^{\lceil p/2\rceil-1},\qquad \cB_{\mathrm{high}}:=F_{\lceil p/2\rceil}.
\end{equation}
So $\cB_{\mathrm{high}}$ consists of the Cauchy data whose traces $\gamma_j$, $j<\lceil p/2\rceil$, vanish. At $p=2$ it is
$\{\gamma_0=0\}$, the Lagrangian of the Dirichlet boundary condition, which we call the \emph{Dirichlet Lagrangian},  and $\cB_{\mathrm{high}}=0$ at $p=1$. Note that $\cB=\cB_{\mathrm{low}}\oplus\cB_{\mathrm{high}}$, with dimensions $N\lceil p/2\rceil$ and $N\lfloor p/2\rfloor$.

\begin{definition}[the graded class $\mathfrak A$]\label{def:graded}
In that grading we say that a Hermitian polynomial $\cA(\xi)=\sum_{d\ge0}\xi^d\cA^d$ lies in $\mathfrak A$ if
\begin{equation}\label{eq:gradedclass}
  \bigl(\cA^d\bigr)_{jl}=0\qquad \text{for} \qquad j+l\geq p-d.
\end{equation}
By Proposition~\ref{prop:green}\textup{(ii)} the Green form of any $H$ satisfying \textup{[H1]--[H2]}
lies in $\mathfrak A$; in particular $\deg_\xi\cA\le p-1$, and \textup{[S]} implies $\cA^d=0$ for
$d\ge1$.
\end{definition}

\begin{lemma}\label{lem:antiinv}
If $\cA\in\mathfrak A$ is invertible then $(\cA^{-1})_{jl}=0$ unless $j+l\ge p-1$.
\end{lemma}
\begin{proof}
With $J$ the block flip of the proof of Proposition~\ref{prop:green}, $(J\cA)_{jl}=\cA_{p-1-j,l}$ vanishes unless $l\le j$, so $J\cA$ is block lower triangular with invertible diagonal blocks, the anti-diagonal blocks of $\cA$; hence $(J\cA)^{-1}=\cA^{-1}J$ is block lower triangular, which is the claim.
\end{proof}
\begin{lemma}[triangular congruence]\label{lem:tricong}
Let $\mathcal C$ be an invertible Hermitian form on $\cB$ and let $K$ be an endomorphism with
$KF_k\subseteq F_{k+1}$ for every $k$, hence $K^p=0$, and with $\mathcal CK$ Hermitian. Then
$\mathcal C(\mathbf 1+K)$ is Hermitian and nondegenerate, and
\[
  \Phi:=\bigl(\mathbf 1+K\bigr)^{-1/2}
\]
is of the form $\mathbf 1+(\text{nilpotent})$ and satisfies $\Phi^*\,\mathcal C(\mathbf 1+K)\,\Phi=\mathcal C$; in particular
$\mathcal C(\mathbf 1+K)$ has the signature of $\mathcal C$.
\end{lemma}
\begin{proof}
$\mathcal C(\mathbf 1+K)$ is Hermitian because $\mathcal C$ and $\mathcal CK$ are, and nondegenerate because
$\mathbf 1+K$ is unipotent. That $\mathcal CK$ is Hermitian says that $K$ is self-adjoint for the form
$\mathcal C$, $K^*\mathcal C=\mathcal CK$, and then so is every real polynomial in $K$. Since $K$ is nilpotent,
$\Phi$ is such a polynomial, the binomial series of $(1+t)^{-1/2}$ truncated after the term of degree $p-1$, and
$\Phi^2(\mathbf 1+K)=\mathbf 1$ because this identity holds for the power series. Hence $\Phi^*\mathcal C=\mathcal C\Phi$
and
\[
  \Phi^*\mathcal C(\mathbf 1+K)\Phi=\mathcal C\,\Phi(\mathbf 1+K)\Phi=\mathcal C\,\Phi^2(\mathbf 1+K)=\mathcal C,
\]
$\Phi$ commuting with $K$. Congruence preserves signature which yields the result.
\end{proof}

We now obtain the following graded trivialization of the Green form.
\begin{theorem}\label{thm:graded}
Let $\cA\in\mathfrak A$ with $\cA(0)$ invertible, and define $K(\xi):=\cA(0)^{-1}\bigl(\cA(\xi)-\cA(0)\bigr)$. Then $KF_k\subseteq F_{k+1}$, so $\cA(\xi)=\cA(0)\bigl(\mathbf 1+K(\xi)\bigr)$ is
invertible for every $\xi$, and
\begin{equation}\label{eq:gradedPhi}
  \Phi(\xi):=\bigl(\mathbf 1+K(\xi)\bigr)^{-1/2}
\end{equation}
of Lemma~\ref{lem:tricong} is a polynomial in $\xi$ of degree $\le p-1$, and so is $\Phi(\xi)^{-1}=\bigl(\mathbf 1+K(\xi)\bigr)^{1/2}$. It satisfies
\begin{equation}\label{eq:graded}
  \Phi(\xi)^*\,\cA(\xi)\,\Phi(\xi)=\cA(0)\qquad\text{for every }\xi\in\R .
\end{equation}
\end{theorem}
\begin{proof}
Write $\mathcal C:=\cA(0)$, so that $K=\mathcal C^{-1}R$ with
$R:=\cA-\mathcal C=\sum_{d\ge1}\xi^d\cA^d$ and $\cA=\mathcal C(\mathbf 1+K)$. By Lemma~\ref{lem:antiinv} $(\mathcal C^{-1})_{jk}=0$ unless $j+k\ge p-1$, and by \eqref{eq:gradedclass} $(\cA^d)_{kl}=0$ unless
$k+l\le p-1-d$; so a nonzero $(j,l)$ block of the coefficient of $\xi^d$ in $K$ needs some $k$ with
$p-1-j\le k\le p-1-d-l$, whence $j-l\ge d\ge1$. Thus $KF_k\subseteq F_{k+1}$ and $K^p=0$; moreover a
nonzero $(j,l)$ block of $K^m$ at degree $d$ needs $j-l\ge d$, and $j-l\le p-1$, so
$\deg K^m\le p-1$ for every $m$, and $\Phi^{\pm1}$, polynomials in $K$, have degree $\le p-1$. $\mathcal CK=R$ is Hermitian, so Lemma~\ref{lem:tricong} applies and gives
\eqref{eq:graded}.
\end{proof}

When $p=1$ the above result is vacuous since $\cA$ is constant. When $p=2$ one has $K^2=0$ and
$\Phi=\mathbf 1-\tfrac12K$.

\paragraph{The Calder\'on Lagrangian.}\label{ssec:calderon}

For the rest of the section, we assume \textup{[H1]--[H4]}.
\begin{theorem}\label{thm:calderonB}
For every $\xi\in\R$ and every $E\in\Delta$, $\dim\Lm^+(\xi,E)=\kp$ and $\Lm^+(\xi,E)$ is $q$-Lagrangian.
\end{theorem}
\begin{proof}
This is also classical \cite[\S20.1]{Hormander3}.
Recall from \eqref{eq:Lam} that $\Lm^+(\xi,E)=\{\gamma\psi:\cL\psi=0,\ \psi\in L^2(\R_+)\}$ is the space of Cauchy data of the solutions of the fiber equation that decay into the interior, $\cL=\Hh(x,\xi)-E$. By [H4], the coefficients of $\Hh(x,\xi)$ are constant for $y\ge y_0$, equal to those of $\Hh^\natural$, so on $[y_0,\infty)$ the $L^2$ solutions of $\cL\psi=0$ are the decaying exponential solutions of $\Hh^\natural-E$. That space has dimension $\kp$ by Lemma~\ref{lem:hermODE} applied to the symbol of $\Hh^\natural(\xi)-E$, which has no real characteristic root for any $\xi\in\R$. A real root would put $E\in\spec\Hh^\natural(\xi)\subseteq\spec H^{\R^2}$, excluded by [H3]. Restriction to $[y_0,\infty)$ is a bijection from $\{\psi\in L^2(\R_+):\cL\psi=0\}$ onto that space: injective by ODE uniqueness and surjective because the leading coefficient $A_p(y)$ is invertible for {every} $y$ (which is [H2] at $(\xi,\zeta)=(0,1)$) so a solution on $[y_0,\infty)$ extends uniquely backwards to $[0,y_0]$. Hence $\dim\Lm^+=\kp$, $\gamma$ being injective on solutions. {Isotropy} is the argument of Lemma~\ref{lem:hermODE}\textup{(c)} applied to $\cL$ on $\R_+$, with \eqref{eq:green} in its one-dimensional form and its solutions decaying at $+\infty$. Since $q$ has signature $0$ on $\cB=\C^{2\kp}$, then a $q$-isotropic subspace of dimension $\kp$ is maximal isotropic and $\Lm^+(\xi,E)$ is $q$-Lagrangian.
\end{proof}

\medskip\noindent\textbf{Transported data.}\label{ssec:graded}
When $\cA$ depends on $\xi$, the Lagrangian data live in the moving fibers $\Lag(\cA(x,\xi))$. The
congruence $\Phi$ of Theorem~\ref{thm:graded} carries them to one fixed fiber, and from now on all
Lagrangian data are read there: without \textup{[S]}, limits at $\xi=\pm\infty$ of fields and Calder\'on Lagrangians, endpoints, walls, rotations and asymmetries are always those of the transported data, also where $\Phi$ is not displayed. Under \textup{[S]}, $\Phi=\mathbf 1$ and nothing is transported.

\begin{definition}[transported data]\label{def:transported}
For a Lagrangian field $L(x,\xi)$ and for the Calder\'on Lagrangian, define $L^\Phi:=\Phi^{-1}L$ and
$\Lm^\Phi:=\Phi^{-1}\Lm^+$, both in $\Lag(\cA(x,0))$. A termination whose $L^\Phi$ does not depend on
$\xi$, that is $L(\xi)=\Phi(\xi)L_0$ with $L_0\in\Lag(\cA(x,0))$, is a \emph{transported pointwise
termination}. For $x$-independent data it is differential of order $\le p-1$ along $\Gam$,
$L(\xi)=\ker\bigl(B_0\Phi(\xi)^{-1}\bigr)$ with $\ker B_0=L_0$ and $\Phi^{-1}$ polynomial.
\end{definition}

\begin{lemma}[the Calder\'on Lagrangian is real-analytic]\label{lem:analytic}
$(\xi,E)\mapsto\Lm^+(\xi,E)\in\Gr(\kp,\cB)$ is real-analytic on $\R_\xi\times\Delta$. Moreover, read
in the transported frame of Theorem~\ref{thm:graded}, where $\cA(0)$ is fixed and \eqref{eq:UL} is
one real-analytic chart, the \emph{edge determinant}
\begin{equation}\label{eq:Fanalytic}
  F^\Phi(\xi,E):=\det\bigl(U_{L^\Phi(\xi)}-U_{\Lm^\Phi(\xi,E)}\bigr)
\end{equation}
is real-analytic on $\R_\xi\times\Delta$ for every $L$ real-analytic in $\xi$. This holds in particular for
every pointwise termination, and for every termination differential along $\Gam$, where $L(\xi)$ is the
kernel of a polynomial matrix family of constant rank. For smooth $L$ independent of $E$,  it is real-analytic in
$E$ at each fixed $\xi$. Its zero set, with the intersection dimensions
$\dim\bigl(L\cap\Lm^+(\xi,E)\bigr)$, does not depend on the frame since $\Phi(\xi)^{-1}$ is
invertible.
\end{lemma}
\begin{proof}
The companion matrix $M(y;\xi,E)$ of the fiber equation $(\Hh(\xi)-E)\psi=0$, written as the first-order system $\partial_y\Psi=M\Psi$ for $\Psi=\gamma\psi$, is built from $A_p^{-1}a_k(x,y,\xi)$ and is polynomial in $(\xi,E)$, so everything below extends holomorphically to
complex parameters near each real point. Beyond $y_0$ the system has constant coefficients by
\textup{[H4]}, $M=M^\natural(\xi,E)$, the companion matrix of $\Hh^\natural(\xi)-E$, and by \textup{[H3]} and the root count of Theorem~\ref{thm:calderonB} the spectrum of
$M^\natural$ avoids $i\R$ with $\kp$ eigenvalues in each half plane, a splitting that persists near
each real point; so the Riesz projection $P^s(\xi,E)$ onto the decaying part is holomorphic there
\cite[Ch.~II \S1.4]{kato1966perturbation}. On $[0,y_0]$ the transition matrix
$\mathsf F_{\xi,E}(0,y_0)$ is entire, by its Picard series \cite[Ch.~1, \S\S7--8]{CoddingtonLevinson}.  Hence
\begin{equation}\label{eq:LamRiesz}
  \Lm^+(\xi,E)=\mathsf F_{\xi,E}(0,y_0)\,\Ran P^s(\xi,E)
\end{equation}
is real-analytic. The congruence $\Phi$ of Theorem~\ref{thm:graded} is polynomial in $\xi$ with
$\Phi^*\cA(\xi)\Phi=\cA(0)$, so \eqref{eq:UL} is one fixed real-analytic chart on $\Lag(\cA(0))$,
which gives \eqref{eq:Fanalytic}; the zero set and the intersection dimensions are frame-independent since
$\Phi(\xi)^{-1}$ is invertible.
\end{proof}

\paragraph{Calder\'on endpoints.}
We now show that these endpoints are defined and how they may be computed. We need the following notation on Grassmannian geometry
collected from Appendix~\ref{app:grass}. A basis $e_1,\dots,e_{2\kp}$ of $\cB$ adapted to the grading has
$\Th_te_i=t^{\lambda_i}e_i$, with $\lambda_i=j$ on $\cB_j$ and $\Th_t=\mathrm{diag}(1,t,\dots,t^{p-1})\otimes\mathbf 1_N$. A $\kp$-plane $P$ with frame $v_1,\dots,v_\kp$ has Pl\"ucker vector
$\varpi_P=v_1\wedge\dots\wedge v_\kp\in\bigwedge^\kp\cB$, whose coordinates $p_I(P)$ in the basis
$e_I=e_{i_1}\wedge\dots\wedge e_{i_\kp}$ are the $\kp\times\kp$ minors of the frame. Vectors of this form are
called \emph{decomposable}; they form a closed cone, and each spans a unique $\kp$-plane
\cite[Ch.~1 \S5]{GriffithsHarris}. The vector $e_I$ has weight
$w_I=\sum_{i\in I}\lambda_i$. Then $w^*$ is the largest weight of a $\kp$-set, $\overline w(P)$ the largest weight with
$p_I(P)\ne0$, and the graded limit $P^\uparrow:=\lim_{t\to\infty}\Th_tP$ keeps exactly the coordinates of
weight $\overline w(P)$, by \eqref{eq:torusreps}; for a curve $\epsilon\mapsto P(\epsilon)$ we write $P^\dagger:=\lim_{t\to\infty}\Th_tP(t^{-1})$ (Lemma~\ref{lem:gradedcurve}).

\medskip\noindent\textbf{The weight conditions [P] and [M].} Two conditions describe the Calder\'on endpoints in closed form. The first asks that they be quantities of the principal symbol alone, independent of $E$, specifically as:
\[
  \text{\textup{[P]}}\qquad
  \Lm^\Phi(x,\pm,E)=V^\Phi_\pm(x)^{\uparrow}\quad\text{for every }E\in\Delta,\ \text{the limits in $\xi$ being uniform in }x.
\]
Here $\Lm^\Phi(x,\pm,E)$ are the limits of the transported Calder\'on Lagrangian $\Phi(\xi)^{-1}\Lm^+(x,\xi,E)$, which exist by Corollary~\ref{cor:graded}\textup{(iii)} below, and $V^\Phi_\pm:=(\mathbf 1+\cN_\pm)V_\pm$ are the transported frozen data, with $(\cN_\pm)_{jl}=(\pm1)^{j-l}(N_{j-l})_{jl}$ for $j>l$ and zero otherwise, read off $\Phi(\xi)^{-1}=\mathbf 1+\sum_{D\ge1}\xi^DN_D$ (Theorem~\ref{thm:graded}). Under \textup{[S]}, $\Phi=\mathbf 1$, $\Lm^\Phi=\Lm^+$ and $V^\Phi_\pm=V_\pm$. \textup{[P]} is used for the symbolic evaluation of the endpoints (section~\ref{sec:APS}). Without it the limits still exist, but may retain lower-order data and depend on $E_0$.  \textup{[P]} holds for all examples of this paper, numerically for \textup{(E7)} and \textup{(E8)}. The second condition, on the frozen data $V_\pm(x)$ of section~\ref{sec:setting} alone, asks that they carry maximal weight:
\begin{equation}\label{eq:Ehyp}
  \text{\textup{[M]}}\qquad
  \max_{w_I=w^*}\bigl|p_I\bigl(V_\pm(x)\bigr)\bigr|\ \ge\ c>0
  \quad\text{for both signs and every }x,
\end{equation}
for orthonormal frames of $V_\pm(x)$ (so that $\sum_I|p_I(V_\pm(x))|^2=1$). Since $w^*$ is the largest weight, \textup{[M]} says that $\overline w(V_\pm(x))=w^*$ for every $x$, with a lower bound uniform in $x$. At even $p$ the maximal weight is carried by the single index set of $\cB_{\mathrm{high}}$, so \textup{[M]} says that $V_\pm\cap\cB_{\mathrm{low}}=0$ uniformly, which is easily verified: the frozen principal problem at $\xi=\pm1$ with the Neumann-type conditions $\gamma_j\psi=0$, $j\ge p/2$, has no decaying solution.

Both conditions are vacuous at $p=1$. By Lemma~\ref{lem:endpoints}, \textup{[M]} implies \textup{[P]} under \textup{[S]}; it holds for second-order systems with a scalar principal symbol (Lemma~\ref{lem:scalarp2}) and fails for \textup{(E8)}, where \textup{[P]} is verified numerically.
\begin{lemma}\label{lem:endpoints}
Fix $E\in\Delta$. For $x\in\Gam$ the limits
$\Lm^+(x,\pm,E):=\lim_{\xi\to\pm\infty}\Lm^+(x,\xi,E)$ exist in $\Gr(\kp,\cB)$, with
\begin{equation}\label{eq:endrate}
  \partial_\xi\Lm^+(x,\xi,E)=O\bigl(|\xi|^{-2}\bigr) .
\end{equation}
If $\overline w(V_\pm(x))=w^*$, then $\Lm^+(x,\pm,E)=V_\pm(x)^{\uparrow}$ for every $E$, a quantity of the principal symbol at $(x,0)$ alone, independent of $E$ and of the lower-order coefficients of $H$. Under \textup{[M]} this holds at every $x$ and \eqref{eq:endrate} is uniform in $x$ and in $E$ in compact subsets of $\Delta$, so that \textup{[M]} implies \textup{[P]} under \textup{[S]}.
\end{lemma}

\begin{proof}
The lemma concerns $\Lm^+$ before the transport of Definition~\ref{def:transported}. The transported endpoints are treated in Corollary~\ref{cor:graded}\textup{(iii)}.

{Step 1.} With $\Psi:=(\psi,D_y\psi,\dots,D_y^{p-1}\psi)$, $\cL\psi=0$ becomes
$\partial_y\Psi=M(y;\xi,E)\Psi$ with $M$ the companion matrix of $\Hh(\xi)-E$, built from $A_p^{-1}a_k(x,y,\xi)$. The entry
in the last block row and $k$-th block column is $O(|\xi|^{p-k})$ since $\ord A_k\le p-k$. Define the rescaling
$z:=|\xi|y$, $\epsilon:=|\xi|^{-1}$ and $\Psi=\Th_{|\xi|}\widetilde\Psi$. The $(j,l)$ entry of the
new coefficient is $|\xi|^{-1+l-j}M_{jl}$; on the superdiagonal $l=j+1$ this is $|\xi|^0$, and in the
last row $j=p-1$ it is $|\xi|^{-p+k}O(|\xi|^{p-k})=O(1)$. Retaining the leading part gives
\begin{equation}\label{eq:rescaled}
  \partial_z\widetilde\Psi=\bigl(M_\infty(x,z\epsilon;\sgn\xi)+O(\epsilon)\bigr)\widetilde\Psi ,
\end{equation}
with $M_\infty(x,y;\pm)$ the companion matrix of the \emph{principal} symbol at $\xi=\pm1$, with remainder $O(\epsilon)$ in $L^\infty(dz)$, uniformly in $x$ and in $E$ in compact sets: it is a finite sum $\sum_{r\ge1}\epsilon^rc_r(x,z\epsilon)$, the $c_r$ being coefficients of $A_p^{-1}H$, bounded by \textup{[H4]}.

Step 2. By Lemma~\ref{lem:hermODE} applied to $\sigma_p(x,y;\pm1,\cdot)$, without real roots by \textup{[H2]},
$\spec M_\infty(x,y;\pm)\cap i\R=\emptyset$ with $\kp$ eigenvalues in each open half plane and a gap $2\nu>0$
uniform in $(x,y)$. The coefficient $\widetilde M(\cdot\,;\epsilon)$ of
\eqref{eq:rescaled} is then hyperbolic at each $z$ with that uniform gap, has $z$-derivative
$O(\epsilon)$ since $\partial_z M_\infty(x,z\epsilon;\pm)=\epsilon\,\partial_yM_\infty$, and is
constant for $z\ge y_0/\epsilon$ by \textup{[H4]}. {It is a slowly varying hyperbolic system},
and so has an exponential dichotomy on $\R_+$ for small $\epsilon$, with any exponent $\nu''<\nu$ and constants
uniform in $(x,\epsilon)$ \cite[Lecture~6, Prop.~1]{Coppel}. Fix such a $\nu''$, an integer $k\ge1$
and a chain $0<\nu'_0<\dots<\nu'_k<\nu''$, and put $X_i:=L^2(e^{-\nu'_iz}dz;\R_+)$ and
$D_i:=\{u\in X_i:\partial_zu\in X_i\}$, so that $X_0\subset\dots\subset X_k$ and multiplication by $z^j$ is
bounded $X_i\to X_{i+j}$ for $i+j\le k$. The dichotomy makes $\mathsf T(\epsilon):=\partial_z-\widetilde
M(\cdot\,;\epsilon)$ surjective from $D_i$ onto $X_i$, with right inverses bounded uniformly in $\epsilon$ (cf.\ \cite[Lecture~3]{Coppel}),
and with a kernel of dimension $\kp$ that is the same
subspace for every $i\le k$ and consists of the $L^2(\R_+)$ solutions, an unstable solution growing
at least like $e^{\nu''z}$ and lying in no $X_i$. So
$V(x,\pm;\epsilon):=\{u(0):u\in\ker\mathsf T(\epsilon)\}$ is the $\kp$-dimensional space of their
Cauchy data, and $V_\pm(x):=V(x,\pm;0)$ is the stable subspace of the frozen $M_\infty(x,0;\pm)$.
Differentiating in $\epsilon$ brings down a power of $z$,
\[
  \partial_\epsilon\bigl[M_\infty(x,z\epsilon;\pm)\bigr]=z\,\partial_yM_\infty(x,z\epsilon;\pm),
\]
so $\partial_\epsilon^{\,j}\mathsf T$ is bounded $D_i\to X_{i+j}$: each $\epsilon$-derivative costs one weight.
Normalize a frame $u(\epsilon)$ of the kernel by requiring $u(\epsilon)(0)-u(0)(0)$ to lie in a fixed complement of $V_\pm(x)$ in $\cB$. Difference quotients satisfy $\mathsf T(\epsilon)\,\delta_hu=-(\delta_h\mathsf T)\,u(\epsilon+h)$, with $\delta_h\mathsf T$ bounded $D_i\to X_{i+1}$ uniformly in $h$. The right inverses on $X_{i+1}$ and the normalization, which fixes the kernel component, bound $\delta_hu$ in $D_{i+1}$ and give its convergence, so that inductively $\partial_\epsilon^{\,j}u$ exists in $D_j$ for $j\le k$. The trace $u\mapsto u(0)$ being
continuous on $D_k$, $V$ is $C^k$ in $\epsilon$, and as $k$ is arbitrary it is $C^\infty$ on $[0,\epsilon_0)$:
\begin{equation}\label{eq:Veps}
  V(x,\pm;\epsilon)=V_\pm(x)+O(\epsilon),\qquad
  \partial_\epsilon^{\,j}V=O(1)\quad\text{for every }j.
\end{equation}
Undoing the rescaling,
\begin{equation}\label{eq:LamTheta}
  \Lm^+(x,\xi,E)=\Th_{|\xi|}\,V\bigl(x,\sgn\xi;|\xi|^{-1}\bigr).
\end{equation}

Step 3. By Lemma~\ref{lem:gradedcurve}, applied to $P=V(x,\pm;\cdot)$, which is $C^\infty$ by \eqref{eq:Veps}, and by
\eqref{eq:LamTheta}, the limits $\Lm^+(x,\pm)=V(x,\pm;\cdot)^\dagger$ exist, with rate
\eqref{eq:endrate}. If $\overline w(V_\pm)=w^*$, the same lemma gives $V_\pm(x)^\uparrow$ as the limit, determined by the principal symbol, $M_\infty$ involving no other data. For the uniformity under \textup{[M]}: the dichotomy constants depend only on $\nu$ and on the $C^1$ bounds of $M_\infty$, uniform in $x$ by \textup{[H4]}, so the right inverses of $\mathsf T(\epsilon)$ on the $X_i$, and hence the constants in \eqref{eq:Veps}, are uniform in $x$ and in $E$ in compact sets; Lemma~\ref{lem:gradedcurve} then gives \eqref{eq:endrate} uniformly, its constant depending only on these bounds and on the lower bound of the maximal-weight coordinates, which \textup{[M]} makes uniform in $x$.
\end{proof}

\begin{corollary}[transported data]\label{cor:graded}
(i) $L\mapsto\Phi(\xi)^{-1}L$ trivializes $\xi\mapsto\Lag(\cA(x,\xi))$, whose fiber is the fixed
  compact $\Lag(\cA(x,0))\cong\U(\kp)$.
\\
(ii)   For $x$-independent data every field $\Phi L_0$, $L_0\in\Lag(\cA(0))$, defines an essentially self-adjoint realization by the fiber argument of Proposition~\ref{prop:SA}\textup{(i)}, hence a transported pointwise termination. For $x$-dependent data $\Phi L_0$ is Lagrangian only at the level of symbols.
\\
(iii)  The limits $\Lm^\Phi(x,\pm)$ exist, with the rate \eqref{eq:endrate}, and equal $(V^\Phi_\pm)^{\uparrow}$ when $\overline w(V^\Phi_\pm)=w^*$, with $V^\Phi_\pm$ as in \textup{[P]}.
\end{corollary}
\begin{proof}
(i) $\Phi(\xi)^{-1}$ is a congruence from $\cA(\xi)$ to $\cA(0)$ by \eqref{eq:graded}, so it carries
$\Lag(\cA(\xi))$ bijectively to $\Lag(\cA(0))$, smoothly in $\xi$. (ii) is the same read backwards with
Proposition~\ref{prop:SA}\textup{(i)}.

(iii) By \eqref{eq:LamTheta}, $\Lm^+(\xi)=\Th_{|\xi|}V(\pm;\epsilon)$ with $\epsilon=|\xi|^{-1}$ and
$\pm=\sgn\xi$, where $V(\pm;\cdot)$ is a $C^\infty$ curve by \eqref{eq:Veps}, to which Lemma~\ref{lem:gradedcurve} applies. By Theorem~\ref{thm:graded},
$\Phi(\xi)^{-1}=\mathbf 1+\sum_D\xi^DN_D$ with $(N_D)_{jl}=0$ unless $j\ge l+D$. Conjugation by
$\Th_{|\xi|}$ multiplies the entry $(j,l)$ by $|\xi|^{l-j}$, so
\[
  G_\pm(\epsilon):=\Th_{|\xi|}^{-1}\Phi(\xi)^{-1}\Th_{|\xi|}
  =\mathbf 1+\sum_D(\pm1)^D\sum_{j,l}\epsilon^{\,j-l-D}(N_D)_{jl}
\]
is a polynomial in $\epsilon$, invertible, with $G_\pm(0)=\mathbf 1+\cN_\pm$, where $\cN_\pm$ collects the entries
with $j-l=D$. Hence $\Lm^\Phi(\xi)=\Th_{|\xi|}V^\Phi(\pm;\epsilon)$ with $V^\Phi(\pm;\epsilon):=G_\pm(\epsilon)V(\pm;\epsilon)$,
again a $C^\infty$ curve, now through $V^\Phi_\pm:=(\mathbf 1+\cN_\pm)V_\pm$ at $\epsilon=0$, and Lemma~\ref{lem:gradedcurve} gives the limits,
the rate, and the value $(V^\Phi_\pm)^{\uparrow}$ when $\overline w(V^\Phi_\pm)=w^*$.
\end{proof}
\noindent Since \textup{[M]} for $V_\pm$ does not imply it for $V^\Phi_\pm$, the few results that use \textup{[M]} itself are stated under \textup{[S]}; without \textup{[S]} they hold when \eqref{eq:Ehyp} is verified directly for $V^\Phi_\pm$, as in \textup{(E7)}.

\begin{definition}[the Calder\'on path]\label{def:Uplus}
$U_+(\xi,E)\in\U(\kp)$ denotes the image of $\Lm^\Phi(\xi,E)$ under \eqref{eq:UL} in the fixed fiber
$\Lag(\cA(0))$, $U_+(\pm,E)$ that of the {endpoints of the Calder\'on path} $\Lm^\Phi(\pm,E)$, with $U_+(\pm):=U_+(\pm,E_0)$, and
$U_L:=U_{L^\Phi}$. At fixed $E_0$, the \emph{Calder\'on path} is
\[
  \gamma_{E_0}:\ \overline\R\ni\xi\longmapsto\Lm^\Phi(\xi,E_0)\in\Lag(\cA(0)),
\]
a path from $\Lm^\Phi(-)$ to $\Lm^\Phi(+)$, whose endpoints exist in general by
Corollary~\ref{cor:graded}\textup{(iii)}. Under \textup{[S]}, $\Phi=\mathbf 1$.
\end{definition}

\begin{definition}[admissibility]\label{def:adm}
A boundary condition in the calculus is \emph{admissible} if the limits $L(x,\pm):=\lim_{\xi\to\pm\infty}L^\Phi(x,\xi)\in\Lag(\cA(x,0))$ of its transported field, the \emph{endpoints of the termination}, (Definition~\ref{def:transported}) exist and:
\\
(a)  it is eventually constant in $x$: $L(x,\pm)=L^{\mathrm{as}}(\pm)$ for $|x|\ge R$; and
\\
(b) it is transverse to the two endpoints of the Calder\'on path, uniformly in $x$:
\begin{equation}\label{eq:admdef}
  \inf_x\bigl|\ell(x,\pm)\bigr|>0,\qquad
  \ell(x,\pm)=\det\bigl(U_{L(x,\pm)}-U_+(x,\pm)\bigr)
\end{equation}
the \emph{endpoint determinant}, whose nonvanishing at a fixed $x$ says
$L(x,\pm)\pitchfork\Lm^\Phi(x,\pm)$.
\end{definition}
Here $U_+(x,\pm)=U_+(x,\pm,E_0)$ is read at the reference energy. Under \textup{[P]} the endpoints do not depend on the energy, so admissibility at $E_0$ is admissibility at every $E\in\Delta$. For a transported pointwise termination it is $L_0(x)$.

\begin{remark}[admissibility and ellipticity]\label{rem:SL}
The endpoint determinant \eqref{eq:admdef} is not the Lopatinski determinant. Ellipticity
\eqref{eq:SLdef} tests the principal field $L^\downarrow(x,\pm)$ against the frozen data $V_\pm(x)$, both
$\kp$-planes in the $2\kp$-dimensional $\cB$ by Lemma~\ref{lem:hermODE}\textup{(b)}, while admissibility tests
the endpoints $L(x,\pm)$ of the termination against the endpoints $\Lm^\Phi(x,\pm)$ of the Calder\'on path. The two coincide at $p=1$, where
$L^\downarrow(\pm)=L(\pm)$ and $\Lm^+(\pm)=V_\pm$, a $T$ with orthonormal rows giving the uniform version. Under \textup{[M]}, Lemma~\ref{lem:endpoints} and
\eqref{eq:torusreps} give $\varpi_{\Lm^+(x,\pm)}=\sum_{w_I=w^*}p_I(V_\pm(x))\,e_I$, so that
$\Lm^+(x,\pm)\supseteq\cB_{\mathrm{high}}$ by \eqref{eq:highinlimit}, and, both having dimension $\kp$ at even $p$,
\begin{equation}\label{eq:Ehigheven}
  \Lm^+(x,\pm)=\cB_{\mathrm{high}}\qquad(p\ \text{even}).
\end{equation}
Under \textup{[S]} and \textup{[M]}, an admissible pointwise $L$ makes $(H,L)$ uniformly Shapiro--Lopatinski
elliptic: $\varpi_L\wedge\varpi_{\Lm^+(\pm)}=\sum_{w_I=w^*}\pm\,p_{I^c}(L)\,p_I(V_\pm)$ involves only the
coordinates of $L$ of the smallest weight $w_{\mathrm{tot}}-w^*$; admissibility makes it nonzero, so these are
the coordinates of $\varpi_{L^\downarrow}$ by \eqref{eq:torusreps}, and the same sum is
$\varpi_{L^\downarrow}\wedge\varpi_{V_\pm}$. Hence $L^\downarrow\pitchfork V_\pm$, uniformly in $x$ by
\eqref{eq:Ehyp} and \eqref{eq:transdet}; rows of $T$ adapted to the filtration, with leading coefficients given by these coordinates, then satisfy the uniform bound of Definition~\ref{def:SL}. The converse
fails: for $-\Delta$, $\Lm^+(\xi,E)=\C(1,ik)$ with $k=k(\xi,E):=\sqrt{\xi^2-E}$, so $\Lm^+(\pm)=\C(0,1)$ is the
Dirichlet Lagrangian, and the Dirichlet condition is elliptic, $V_\pm=\C(1,i)$, but not admissible, with $\mathcal I=0$ since
$\spec\Hh_L(\xi)=[\xi^2,\infty)$ misses $\Delta$. Under \textup{[M]} at $p=2$ the wall always contains the
Dirichlet Lagrangian, by \eqref{eq:Ehigheven}. Without \textup{[M]} the implication fails: in \textup{(E8)} the
generic transported pointwise termination is admissible and (numerically) not elliptic. Section~\ref{sec:obs}
uses ellipticity on the elliptic route and escape on the spectral route, the latter provided by admissibility when the Calder\'on path converges uniformly in energy (Proposition~\ref{prop:exact}\textup{(i)}), so the remark is
only a criterion.
\end{remark}

\paragraph{Topology of admissible set.}\label{ssec:admset}

Admissibility asks that the endpoints $L(\pm)$ of the termination be transverse to those $\Lm^\Phi(\pm)$ of the Calder\'on path. The derivation of the topology of that set follows standard analyses of Lagrangian Grassmannians: an affine chart of the Lagrangians transverse to one of them, reduction to the non-degenerate part, and Sylvester's law of inertia; see \cite[Ch.~2]{McDuffSalamon17} for the chart and the Maslov cycle, and \cite{Furutani04,BHS20} in the unitary parametrization.

\begin{proposition}\label{prop:admtopology}
Let $\Lm^\Phi(\pm)\in\Lag(\cA(0))\cong\U(\kp)$ be the endpoints of the Calder\'on path. Write $U_1:=U_+(-)$, $U_2:=U_+(+)$ \textup{(}Definition~\ref{def:Uplus}\textup{)} and $\cU:=\U(\kp)\setminus(\Sigma_{U_1}\cup\Sigma_{U_2})$ for the set of $\kp$-planes transverse to both endpoints, an open subset of one fiber. For $x$-independent data it is the set of endpoint values $L_0$ of admissible transported pointwise
terminations $\Phi L_0$. Let $W:=U_1^{-1}U_2$ be the relative unitary of the two endpoints and
\begin{equation}\label{eq:defect}
  \kp_\cap:=\dim\bigl(\Lm^\Phi(-)\cap\Lm^\Phi(+)\bigr)=\dim\ker(\mathbf 1-W)\in\{0,1,\dots,\kp\}
\end{equation}
be the \emph{defect}, $\kp_\cap=0$ meaning transversality of the endpoints and $\kp_\cap=\kp$ their
coincidence. Define $\mathsf K:=\ker(\mathbf 1-W)$, $W':=W|_{\mathsf K^\perp}$ and
$\mathsf B:=0_{\mathsf K}\oplus i(\mathbf 1+W')(\mathbf 1-W')^{-1}$ in the splitting $\C^\kp=\mathsf K\oplus \mathsf K^\perp$.

\begin{enumerate}[label=\textup{(\roman*)},nosep]
\item \textup{(The chart.)} The map
\[
  U\ \longmapsto\ \mathsf A:=\mathsf A_0-\mathsf B,\qquad
  \mathsf A_0:=i(\mathbf 1+U_1^{-1}U)(\mathbf 1-U_1^{-1}U)^{-1},
\]
is a diffeomorphism of $\U(\kp)\setminus\Sigma_{U_1}$ onto $\Herm(\kp)$ which carries $\cU$ onto the matrices whose \emph{compression} $\mathsf A_{22}:=\Pi_{\mathsf K^\perp}\mathsf A|_{\mathsf K^\perp}$ to $\mathsf K^\perp$ is invertible:
\begin{equation}\label{eq:cayleyimage}
  \cU\ \cong\ \bigl\{\mathsf A\in\Herm(\kp):\det\mathsf A_{22}\ne0\bigr\},\qquad
  \dim\bigl(L\cap\Lm^\Phi(+)\bigr)=\dim\ker\mathsf A_{22} .
\end{equation}
At $\kp_\cap=0$ one has $\mathsf A_{22}=\mathsf A$ and $\cU$ is the set of \emph{invertible Hermitian matrices}; at $\kp_\cap=\kp$ one has $\mathsf B=0$ and $\cU$ is all of $\Herm(\kp)$.
\item \textup{(The chambers.)} $\cU$ has exactly $\kp-\kp_\cap+1$ connected components, called {\em chambers}, indexed by
\begin{equation}\label{eq:chamberlabel}
  j(L):=\#\bigl\{\text{positive eigenvalues of }\mathsf A_{22}(L)\bigr\}\in\{0,1,\dots,\kp-\kp_\cap\}.
\end{equation}
Every chamber is open and \emph{nonempty}.
\item \textup{(Constraints on defects.)} Not every value in \eqref{eq:defect} is attained: under
\textup{[S]} and \textup{[M]} one has $\kp_\cap\ge N\lfloor p/2\rfloor$, while $\kp_\cap=0$ at $p=1$. So under
these conditions, without which the defect can vanish for every $p$ as in \textup{(E8)}, the defect
is $0$ at $p=1$, is $\kp$ at even $p$, and satisfies $\kp-\tfrac12N\le\kp_\cap\le\kp$ at odd
$p\ge3$. Correspondingly $\cU$ has $\kp+1$ chambers at $p=1$, one at even $p$, and at most
$\tfrac12N+1$ at odd $p\ge3$.
\end{enumerate}
\end{proposition}

\begin{proof}
$W$ is unitary, so $\mathsf K$ and $\mathsf K^\perp$ are both $W$-invariant and $1\notin\spec W'$. Hence $\mathsf B$ is
defined and Hermitian.

(i) Left translation by $U_1^{-1}$ carries $\Sigma_{U_1}$ to $\Sigma_{\mathbf 1}$ and $\Sigma_{U_2}$ to $\Sigma_W$. The Cayley transform
\[
  \mathrm{cay}:\Herm(\kp)\longrightarrow\U(\kp)\setminus\Sigma_{\mathbf 1},\qquad
  \mathrm{cay}(\mathsf A_0)=(\mathsf A_0-i)(\mathsf A_0+i)^{-1},
\]
is a diffeomorphism with inverse $V\mapsto i(\mathbf 1+V)(\mathbf 1-V)^{-1}$, no eigenvalue $(a-i)/(a+i)$ equalling $1$ for real $a$ \cite[Ch.~1]{BHS20}. So $U\mapsto\mathsf A_0$, and with it $U\mapsto\mathsf A=\mathsf A_0-\mathsf B$, is a diffeomorphism of $\U(\kp)\setminus\Sigma_{U_1}$ onto $\Herm(\kp)$. For the second wall, $\mathrm{cay}(\mathsf A_0)=\mathbf 1-2i(\mathsf A_0+i)^{-1}$ gives
\begin{equation}\label{eq:caydiff}
  \mathrm{cay}(\mathsf A_0)-W=\bigl[(\mathbf 1-W)\mathsf A_0-i(\mathbf 1+W)\bigr](\mathsf A_0+i)^{-1},
\end{equation}
so the two sides have the same kernel dimension. In the splitting $\mathsf K\oplus \mathsf K^\perp$, where $\mathbf 1-W$ vanishes on the first summand and $\mathbf 1+W$ is $2$ there, the right-hand factor is block lower triangular:
\[
  (\mathbf 1-W)\mathsf A_0-i(\mathbf 1+W)
  =\begin{pmatrix}-2i\,\mathbf 1_{\kp_\cap} & 0\\[3pt]
     (\mathbf 1-W')(\mathsf A_0)_{21} & (\mathbf 1-W')\bigl((\mathsf A_0)_{22}-\mathsf B'\bigr)
   \end{pmatrix},\qquad
  \mathsf B':=i(\mathbf 1+W')(\mathbf 1-W')^{-1},
\]
$W'$ being normal, so that $(\mathbf 1-W')^{-1}$ and $\mathbf 1+W'$ commute. Its determinant is $(-2i)^{\kp_\cap}\det(\mathbf 1-W')\det\mathsf A_{22}$ and its kernel is isomorphic to $\ker\mathsf A_{22}$, the block on $\mathsf K$ being invertible. With \eqref{eq:transdet} this gives \eqref{eq:cayleyimage}.

(ii) The condition $\det\mathsf A_{22}\ne0$ constrains one block of $\mathsf A$, so \eqref{eq:cayleyimage} is the product of the contractible $\Herm(\mathsf K)\times\mathrm{Hom}(\mathsf K^\perp,\mathsf K)$ with the invertible elements of $\Herm(\mathsf K^\perp)$. By Sylvester's law of inertia the latter fall into $\kp-\kp_\cap+1$ classes by signature, each nonempty and connected, being an orbit of the connected group $\mathrm{GL}(\mathsf K^\perp)$ acting by congruence; so the chambers are open, nonempty and connected.

(iii) Both endpoints contain $\cB_{\mathrm{high}}$ by \eqref{eq:highinlimit}, so
$\kp_\cap\ge\dim\cB_{\mathrm{high}}=N\lfloor p/2\rfloor$, which is $\kp$ at even $p$. At $p=1$,
$\Th_t=\mathbf 1$ and the endpoints are the frozen data themselves; $\sigma_p$ is then odd, so
$M_\infty(x,0;-)=-M_\infty(x,0;+)$ and $\Lm^+(-)$ is the \emph{unstable} subspace of
$M_\infty(x,0;+)$, complementary to $\Lm^+(+)$ by Lemma~\ref{lem:hermODE}\textup{(b)}. The
chamber counts are (ii).
\end{proof}

\paragraph{Reduction of the edge spectrum.}\label{ssec:edgeset}

We can now recast the spectrum of $\Hh_L(\xi)$, an unbounded operator on a half line, as the
vanishing of one determinant on $\U(\kp)$, together with the discreteness and multiplicity bounds
that allow us to count eigenvalues. That reduction allows for finite-dimensional computations.

\begin{lemma}[no essential spectrum in the gap]\label{lem:essspec}
Let $H$ and $L$ be $x$-independent and let $L$ be a termination. Then for every $\xi\in\R$,
$\spec_{\mathrm{ess}}\Hh_L(\xi)\subseteq\spec\Hh^\natural(\xi)\subseteq\spec H^{\R^2}$, so
$\spec_{\mathrm{ess}}\Hh_L(\xi)\cap\Delta=\emptyset$.
\end{lemma}
\begin{proof}
Let $E\in\spec_{\mathrm{ess}}\Hh_L(\xi)$ and let $(u_n)$ be a singular Weyl sequence: $\norm{u_n}=1$, $u_n\rightharpoonup0$
and $(\Hh_L(\xi)-E)u_n\to0$. On each interval $(0,R)$ the equation bounds $u_n$ in $H^p(0,R)$, the leading
coefficient $A_p$ being invertible, so $u_n\to0$ in $H^{p-1}(0,R)$ by compact embedding, after passing to a subsequence. Let $\chi(y)$ be smooth,
equal to $0$ for $y\le y_0$ and to $1$ for $y\ge y_0+1$, and extend $\chi u_n$ by zero to $\R$. The coefficients
of $\Hh(\xi)$ and $\Hh^\natural(\xi)$ agree for $y\ge y_0$, so
$(\Hh^\natural(\xi)-E)(\chi u_n)=\chi(\Hh_L(\xi)-E)u_n+[\Hh(\xi),\chi]u_n\to0$, the commutator involving only
derivatives of order $\le p-1$ of $u_n$ on $(y_0,y_0+1)$, while $\norm{\chi u_n}\to1$. Hence
$E\in\spec\Hh^\natural(\xi)$. By Fourier transform in $y$, $\spec\Hh^\natural(\xi)=\bigcup_\zeta\spec a^w(y_0;\xi,\zeta)$,
which does not meet $\Delta$ by \textup{[H3]}.
\end{proof}

We now identify edge spectrum with pullback of the Maslov cycle.
\begin{theorem}\label{thm:edge}
Let $H$ and $L$ be $x$-independent and $E\in\Delta$. Then
\begin{equation}\label{eq:edgespec}
  E\in\spec\Hh_L(\xi)\iff L(\xi)\cap\Lm^+(\xi,E)\ne\{0\},
\end{equation}
with multiplicity $\dim\bigl(L(\xi)\cap\Lm^+(\xi,E)\bigr)\le\kp$, and only finitely many eigenvalues lie in each compact $K\subset\Delta$. In the chart \eqref{eq:UL} of the transported data, the right-hand side of \eqref{eq:edgespec} is the vanishing of the edge determinant \eqref{eq:Fanalytic}, $\det\bigl(U_L(\xi)-U_+(\xi,E)\bigr)=0$.

For a transported pointwise termination $U_L$ is a single point of $\U(\kp)$, and the edge spectrum in
$\R_\xi\times\Delta$ is the preimage of the wall $\Sigma_{U_L}$ under $(\xi,E)\mapsto U_+(\xi,E)$.

The edge states are \emph{exponentially localized}: a solution with data in $\Lm^+$ decays into the interior by definition, and \textup{[H3]} forces that decay to be exponential at a positive rate, uniform for $(\xi,E)$ in compact sets, the characteristic roots of $\Hh^\natural(\xi)-E$ being non-real (Lemma~\ref{lem:hermODE}).
\end{theorem}
\begin{proof}
That the eigenvalues are the parameters at which the boundary condition meets the Cauchy-data space is classical as the zeros of a Weyl-function determinant \cite[Ch.~2--3]{BHS20}, or the pullback of the Maslov cycle \cite{Nicolaescu95}. By Lemma~\ref{lem:essspec}, $\spec\Hh_L(\xi)\cap\Delta$ consists of isolated eigenvalues of finite multiplicity. An eigenfunction at $E\in\Delta$ is an $L^2(\R_+)$ solution of $(\Hh(\xi)-E)\psi=0$ with Cauchy data in $L$; since $\gamma$ is injective on such solutions, the eigenspace is isomorphic to $L\cap\Lm^+(\xi,E)$, which is \eqref{eq:edgespec}, and its dimension is at most $\dim L=\kp$. Both $L$ and $\Lm^+(\xi,E)$ are $\cA(\xi)$-Lagrangian, the first by Proposition~\ref{prop:SA}\textup{(i)}, the second by Theorem~\ref{thm:calderonB}. So \eqref{eq:transdet} applied after transport by $\Phi(\xi)^{-1}$, which preserves intersection dimensions, turns non-transversality into the vanishing of $\det\bigl(U_L(\xi)-U_+(\xi,E)\bigr)$, with the multiplicity as stated.

Finiteness in compact subsets of $\Delta$ comes from Lemma~\ref{lem:essspec}. Accumulation at $\partial\Delta$ is possible.
\end{proof}

\section{Existence, quantization and stability of $\mathcal I$}\label{sec:obs}

We now show conditions under which $\mathcal I$ is well defined, quantized and stable, along two routes.

The \emph{elliptic route} takes as hypothesis the uniform Shapiro--Lopatinski estimate and assumes that the
$B_j$ of \eqref{eq:traceop} are differential along $\Gam$ with $C_b^\infty$ coefficients and that $T$ has a bounded right inverse $R_T:\prod_iH^{p-m_i-1/2}(\Gam)\to H^p(\Om)$, $m_i$ the row orders of $T$. Normal trace operators have one \cite[Prop.~1.6.5, Rem.~1.6.5$'$]{Grubb96}, and Appendix~\ref{app:import} constructs one for the non-normal rows $\gamma_0+iB\gamma_1$ of the examples. It includes $x$-dependent perturbations on a
bounded set. We show localization of $f(H_L)$ near the boundary (Theorem~\ref{thm:sG}), the trace-class
property (Theorem~\ref{thm:traceclass}) and the index formula (Theorem~\ref{thm:quantisation}), and then
stability along paths of elliptic realizations (Corollary~\ref{cor:stab}) and under changes on a bounded set
(Theorem~\ref{thm:local}). When the principal symbol does not depend on $y$ on the collar ($0\le y\le y_0$), the lower-order coefficients there do not contribute to the invariant
(Theorem~\ref{thm:collar}).

The \emph{spectral route} applies to $x$-independent operators but is more general otherwise. A fiber argument gives
existence and quantization under the escape condition alone (Theorem~\ref{thm:fiberexist}). Escape is implied
by admissibility when the Calder\'on path converges to its endpoints uniformly in energy, and by ellipticity
(Remark~\ref{rem:escape}); without it the density of states does not localize (Proposition~\ref{prop:exact}).
This route gives no stability under $x$-dependent perturbations (Remark~\ref{rem:escape}).

\paragraph{Functional calculus at the boundary.}\label{ssec:noparam}

Write
\begin{equation}\label{eq:res}
  R_L(z):=(H_L-z)^{-1},\qquad R_b(z):=\bigl(H^{\R^2}-z\bigr)^{-1},\qquad
  \cG(z):=R_L(z)-r^+R_b(z)e^+ ,
\end{equation}
$\cG(z)$ being the part of the resolvent due to the boundary and the collar. The BdM boundary calculus
enters this section for elliptic trace operators with differential
$B_j$ in  the uniform local regularity of Lemma~\ref{lem:ellreg} and in the bounded right inverse $R_T$ of the
trace operator in the proof of Lemma~\ref{lem:CT}.

\paragraph{Localization and trace-class property under ellipticity.}\label{ssec:trace}

We start with a standard regularity result.
\begin{lemma}\label{lem:ellreg}
Call $B\subset\overline\Om$ a \emph{unit box} if $B=([x_1,x_1+1]\times[y_1,y_1+1])\cap\overline\Om$, and write $B^*$ for the concentric box of side $3$. Assume \textup{[H1]--[H2]} and uniform Shapiro--Lopatinski ellipticity.
\begin{enumerate}[label=\textup{(\roman*)},nosep]
\item There are $C_n$ \emph{independent of the box} with
$\norm w_{H^{pn}(B)}\le C_n\sum_{j\le n}\norm{H^jw}_{L^2(B^*)}$ for every unit box $B$ and every
$w\in\cD(H_L^n)$. For $n=1$ the same holds, with $w\in\cD_{L'}$, for every pair $(H',T')$, symmetric or not, with the principal symbols and ellipticity bounds of $(H,T)$ and bounded lower-order coefficients.
\item If $s>1+k$ and $T$ is bounded $H^{-s}(B')\to H^{s}(B)$ for unit boxes $B,B'$, then its kernel is of class
$C^k$ on $B\times B'$ with $|\partial_w^\alpha\partial_{w'}^\beta K_T(w,w')|\le C_s\norm{T}_{H^{-s}(B')\to H^{s}(B)}$ for
$|\alpha|+|\beta|\le k$, with $C_s$ independent of the boxes.
\end{enumerate}
\end{lemma}
\begin{proof}
(i): On a fixed box this is the a priori estimate for an elliptic realization, given by its parametrix mapping into the domain \cite[Prop.~1.4.2]{Grubb96}; rows of $T$ of order $\ge p$ are first brought to order $<p$ by composition with pseudodifferential operators over $\Gam$, which leaves the realization unchanged \cite[\S1.4]{Grubb96}. The constant is uniform in the
box: freezing the coefficients at its centre, it depends only on the uniform Shapiro--Lopatinski
bound, the ellipticity constant and $\norm{A_k}_\infty$, and the difference, of size
$O(\mathrm{diam}\,B^*)$, is absorbed on boxes of a fixed smaller size (equivalently, by precompactness in $C^k$ of the translated local data). For $n>1$ apply this to
$H^jw$, $j<n$, on nested boxes. (ii) is Sobolev embedding in dimension two: $\partial^\beta\delta_{w'}\in H^{-s}$ uniformly in $w'$
for $s>1+|\beta|$, and $H^s\subset C^k$ for $s>1+k$.
\end{proof}

The next lemma is a Combes--Thomas estimate \cite{CombesThomas} in the form a boundary problem
requires. Conjugating by an exponential weight $e^{\alpha\phi}$, $\phi=\phi(x,y)$, we have
\begin{equation}\label{eq:conjshift}
  e^{\alpha\phi}\,D\,e^{-\alpha\phi}=D+i\alpha\nabla\phi .
\end{equation}
Write $H_\alpha$ and
$T_\alpha$ for $H$ and $T$ with $D$ replaced by $D+i\alpha\nabla\phi$, so that, for $\phi=\phi(x)$,
$L_\alpha(x,\xi)=L\bigl(x,\xi+i\alpha\phi'(x)\bigr)$, and
\[
  e^{\alpha\phi}\,H_T\,e^{-\alpha\phi}=\bigl(H_\alpha\bigr)_{T_\alpha}
  \qquad\text{on}\qquad \cD_{L_\alpha}=e^{\alpha\phi}\cD_L .
\]
For a {pointwise} termination and $\phi=\phi(x)$, $L$ does not depend on $\xi$, so $T_\alpha=T$. For a termination differential along $\Gam$ the domain does move, and the conjugated
realization is self-adjoint only for the weighted inner product $\langle e^{-\alpha\phi}\cdot,e^{-\alpha\phi}\cdot\rangle$, whose comparison with $L^2$ costs the factor $e^{2|\alpha|\norm\phi_\infty}$. So $\norm{(H_T-z)^{-1}}\le1/|\Im z|$ coming from
self-adjointness must be replaced, for the conjugated operator, by an a priori
estimate.

\begin{lemma}\label{lem:CT}
Let $T$ be a boundary condition \eqref{eq:traceop} defining a termination, with $(H,T)$
uniformly Shapiro--Lopatinski elliptic, so that its realization $H_T$, closed by Lemma~\ref{lem:ellreg}\textup{(i)}, is self-adjoint. Then there
are $c,C>0$ depending only on $p$, $N$, the coefficient bounds and the ellipticity constants
such that for every bounded $\phi=\phi(x,y)\in C^\infty$ with $|\partial^\beta\phi|\le1$ for
$1\le|\beta|\le p$, every $z$ in a fixed
bounded set with $d:=\dist(z,\spec H_T)\in(0,1]$, and every $|\alpha|\le c\,d$,
\[
  \norm{e^{\alpha\phi}(H_T-z)^{-1}e^{-\alpha\phi}}\le C/d .
\]
Consequently, for bounded functions $g_1,g_2$ on $\Om$ whose supports are at distance $\rho$, the multiplication operators satisfy, in operator norm on $L^2(\Om;\C^N)$,
\begin{equation}\label{eq:CTdecay}
  \bigl\|g_1(H_T-z)^{-1}g_2\bigr\|\le\frac Cd\norm{g_1}_\infty\norm{g_2}_\infty e^{-c\,d\,\rho}.
\end{equation}
The same holds for $H^{\R^2}$ in the absence of any boundary condition.
\end{lemma}
\begin{proof}
\emph{Step 1.} Since $\phi$ is bounded with bounded derivatives, $e^{\pm\alpha\phi}$ is bounded on $L^2$ and on
$H^p$, so $H_{T,\alpha}:=(H_\alpha)_{T_\alpha}$ is defined on $\cD_{L_\alpha}=e^{\alpha\phi}\cD_L$ and
$H_{T,\alpha}-z$ is a bijection onto $L^2$ with inverse $e^{\alpha\phi}(H_T-z)^{-1}e^{-\alpha\phi}$, to be bounded independently of $\norm\phi_\infty$.

\emph{Step 2.} By \eqref{eq:conjshift} the perturbations $H_\alpha-H$
and $T_\alpha-T$ are $O(\alpha)$ with coefficients polynomial in the derivatives of $\phi$, bounded
for $|\alpha|\le1$ because $|\nabla\phi|\le1$, and each is of one order less than the term it
perturbs; in $T_\alpha-T$ this includes the lower traces produced by $\partial_y\phi$ in $\gamma_j(e^{-\alpha\phi}u)$. So $(H_\alpha,T_\alpha)$ has the interior and boundary principal symbols of $(H,T)$,
uniform Shapiro--Lopatinski ellipticity persists for $|\alpha|\le\alpha_0$ with constants
independent of $\alpha$, and Lemma~\ref{lem:ellreg}\textup{(i)} with $n=1$, summed over unit boxes
whose triples have bounded overlap, gives
\begin{equation}\label{eq:aprioriconj}
  \norm w_{H^p(\Om)}\le C\bigl(\norm{H_{T,\alpha} w}+\norm w\bigr),\qquad w\in\cD_{L_\alpha},
\end{equation}
with $C$ independent of $\alpha$.

\emph{Step 3}. Write $\partial\cH:=\prod_iH^{p-m_i-1/2}(\Gam)$,
$m_i$ the orders of the rows of $T$. By hypothesis, $T$ has a bounded right inverse $R_T:\partial\cH\to H^p(\Om)$. By Step~2, $\norm{T-T_\alpha}_{H^p\to\partial\cH}\le C|\alpha|$.
For $f\in L^2$, put $\psi:=(H_{T,\alpha}-z)^{-1}f\in\cD_{L_\alpha}$ (so $T_\alpha\psi=0$) and $\psi_0:=\psi-R_TT\psi$. Then $T\psi_0=0$, so
$\psi_0\in\cD_L$, and $T\psi=(T-T_\alpha)\psi$ gives $\norm{\psi-\psi_0}_{H^p}\le C|\alpha|\norm \psi_{H^p}$, and
\[
  (H_T-z)\psi_0=(H-z)\psi-(H-z)(\psi-\psi_0)=f-(H_\alpha-H)\psi-(H-z)(\psi-\psi_0),
\]
so $\norm{(H_T-z)\psi_0}\le\norm f+C|\alpha|\norm \psi_{H^p}$. Self-adjointness of $H_T$ gives
$\norm{\psi_0}\le d^{-1}\norm{(H_T-z)\psi_0}$, while $\norm \psi\le\norm{\psi_0}+C|\alpha|\norm \psi_{H^p}$. Inserting
\eqref{eq:aprioriconj} in the form $\norm \psi_{H^p}\le C(\norm f+\norm \psi)$ into both, and using
$|\alpha|\le1$ and $d\le1$,
\[
  \norm \psi\ \le\ \frac{C_1}{d}\,\norm f+\frac{C_2|\alpha|}{d}\,\norm \psi ,
\]
so that $|\alpha|\le d/(2C_2)$ yields $\norm \psi\le 2C_1d^{-1}\norm f$, which is the assertion. For \eqref{eq:CTdecay}, let $\delta(w):=\min\bigl(\rho,\dist(w,\supp g_2)\bigr)$ on $\R^2$, a $1$-Lipschitz function with $\delta=0$ on $\supp g_2$ and $\delta=\rho$ on $\supp g_1$, and put $\phi:=c_0\,(\delta*\mathsf m)|_\Om$, with $\mathsf m\ge0$ a fixed mollifier supported in the unit ball, $\int\mathsf m=1$. Since $\partial^\beta(\delta*\mathsf m)=(\partial_k\delta)*\partial^{\beta-e_k}\mathsf m$ for $|\beta|\ge1$, a fixed $c_0\in(0,1]$ gives $|\partial^\beta\phi|\le1$ for $1\le|\beta|\le p$, while $|\delta*\mathsf m-\delta|\le1$; so $\phi\ge c_0(\rho-1)$ on $\supp g_1$ and $\phi\le c_0$ on $\supp g_2$. Writing $g_1(H_T-z)^{-1}g_2=\bigl(g_1e^{-\alpha\phi}\bigr)\bigl(e^{\alpha\phi}(H_T-z)^{-1}e^{-\alpha\phi}\bigr)\bigl(e^{\alpha\phi}g_2\bigr)$ with $\alpha=c\,d$ gives \eqref{eq:CTdecay}, with $c$ replaced by $c\,c_0$ and $C$ by $Ce^{2cc_0}$, as $d\le1$. For a pointwise termination and $\phi=\phi(x)$ the domain does not move, and the argument is that of \cite{CombesThomas}. For $H^{\R^2}$ there is no boundary condition and the same argument applies.
\end{proof}

We can now state the main localization estimate.
\begin{theorem}[localization of $f(H_L)$]\label{thm:sG}
Let $(H,L)$ be uniformly Shapiro--Lopatinski elliptic. Let $f\in C_c^\infty(K)$
for a fixed compact $K\subset\Delta$. Then the kernel of $f(H_L)$ is smooth with
\begin{equation}\label{eq:decay}
  \bigl|\partial^\alpha K_f(x,y;x',y')\bigr|\le C_{\alpha,M}\,\langle x-x'\rangle^{-M}e^{-c(y+y')}
  \qquad\text{for every }M ,
\end{equation}
where $C_{\alpha,M}=C'_{\alpha,M}\norm f_{C^{m(M)}}$ with $m(M):=4M+1$ and constants $C'_{\alpha,M}$ independent of
$f$, and $c>0$ is independent of $M$ and of $f$.
\end{theorem}
\noindent The decay normal to $\Gam$ is exponential, at a rate fixed by the bulk gap; along $\Gam$ it is only rapid, as for a convolution in $x$ with a Schwartz function.
\begin{proof}
Let $\widetilde f$ be an almost-analytic extension of $f$ of order $J$: supported in a compact neighborhood of $K$ in $\C$ disjoint from $\spec H^{\R^2}$, with
$\widetilde f|_\R=f$ and $|\dbar\widetilde f(z)|\le C_{f,J}|\Im z|^J$ \cite[Ch.~2]{Davies}. The Helffer--Sj\"ostrand
formula \cite[Ch.~2]{Davies} gives, 
\begin{equation}\label{eq:HS}
  f(H_L)=-\frac1\pi\int_\C\dbar\widetilde f(z)\,R_L(z)\,dA(z)
  =-\frac1\pi\int_\C\dbar\widetilde f(z)\,\cG(z)\,dA(z),
\end{equation}
the second equality because $z\mapsto r^+R_b(z)e^+$ is holomorphic near $\supp\widetilde f$, which avoids
$\spec H^{\R^2}$, so that its contribution vanishes by Stokes' theorem.

\emph{Decay in $y+y'$.} By \eqref{eq:HS} it suffices to bound $g_1\cG(z)g_2$ for bounded cutoffs $g_1,g_2$, and by duality it suffices to take $g_1$ supported in
$\{y\ge y_1\}$, $y_1>y_0$, with $g_2$ arbitrary, since $\max(y,y')\ge\tfrac12(y+y')$. With $\theta=\theta(y)$ a
cutoff equal to $0$ for $y\le y_0+\frac13(y_1-y_0)$ and to $1$ for $y\ge y_0+\frac23(y_1-y_0)$, so that $H=H^{\R^2}$ on $\supp\theta$ by \textup{[H4]}, one has $e^+\theta R_Lg_2=R_be^+\bigl(\theta g_2+[H,\theta]R_Lg_2\bigr)$, hence
\[
  g_1\cG(z)g_2
  =g_1r^+R_be^+(\theta-\mathbf 1)g_2+g_1r^+R_be^+[H,\theta]R_Lg_2 ,
\]
where in both terms $g_1r^+R_be^+$ acts across a distance $\ge\tfrac13(y_1-y_0)$. Lemma~\ref{lem:CT}
for $H^{\R^2}$ with a weight in $y$ bounds it by $Cd_b^{-1}e^{-cd_b(y_1-y_0)/3}$,
$d_b:=\min\bigl(1,\dist(z,\spec H^{\R^2})\bigr)\ge\delta>0$ on $\supp\widetilde f$, $K$ being a compact subset of the gap, and $\norm{[H,\theta]R_L}\le C\norm{R_L}_{L^2\to H^p}\le C|\Im z|^{-1}$ by
Lemma~\ref{lem:ellreg}\textup{(i)}, $[H,\theta]$ having order $p-1$ and $z$ ranging in a bounded set.

\emph{Decay in $|x-x'|$.} We use Lemma~\ref{lem:CT} for $H_T=H_L$, applied to $R_L$ in the first form of \eqref{eq:HS}, with $\alpha=c|\Im z|\le c\,d$ since $d\ge|\Im z|$. The resulting bound is $C|\Im z|^{-1}e^{-c|\Im z|\rho}$, $\rho:=|x-x'|$.

Both bounds are $O(|\Im z|^{-1})$, and the rate in $y$ does not involve $z$. The bound
$|\dbar\widetilde f(z)|\le C_{f,J}|\Im z|^J$ makes the integral \eqref{eq:HS}
converge absolutely, and splitting at $|\Im z|=\rho^{-1/2}$ gives
$O(\rho^{-J/2})$. Then $J=4M$, for which $C_{f,J}\le C\norm f_{C^{4M+1}}$ \cite[Ch.~2]{Davies}, gives the rate $\rho^{-2M}$. The geometric mean of the bounds
in $x$ and in $y$ gives the product form of \eqref{eq:decay}.

\emph{Pointwise bounds.} Write $f(H_L)=(H_L+i)^{-n}g(H_L)(H_L+i)^{-n}$ with $g(E):=(E+i)^{2n}f(E)$, again in
$C_c^\infty(K)$, and insert partitions of unity $\sum_{B}\mathbf 1_{B}=\mathbf 1$ over unit boxes between the three
factors. For unit boxes $B,B_1$, $\mathbf 1_B(H_L+i)^{-n}\mathbf 1_{B_1}$ is bounded $L^2\to H^{pn}(B)$ with norm
$O(e^{-c\,\mathrm{dist}(B,B_1)})$, by Lemma~\ref{lem:ellreg}\textup{(i)} and Lemma~\ref{lem:CT} at $z=-i$ applied to
the functions $E^j(E+i)^{-n}$, $j\le n$, of $H_L$, which are combinations of powers $(H_L+i)^{-k}$, $0\le k\le n$, by binomial expansion of $E=(E+i)-i$, the term $k=0$ being local.
By duality, the last factor is bounded $H^{-pn}\to L^2$ with
the same decay, and the middle factor obeys the operator bounds above. Summing over the intermediate
boxes, $\mathbf 1_Bf(H_L)\mathbf 1_{B'}$ is bounded $H^{-pn}(B')\to H^{pn}(B)$ with the decay of
\eqref{eq:decay}, and Lemma~\ref{lem:ellreg}\textup{(ii)} with $pn>1+|\alpha|$ turns this into \eqref{eq:decay}.
\end{proof}
This yields the first main result of this section.
\begin{theorem}\label{thm:traceclass}
Under the hypotheses of Theorem~\ref{thm:sG}, $\vartheta(x)f(H_L)$ and $f(H_L)\vartheta(x)$ are trace-class for every $\vartheta\in C_c^\infty(\R)$, with the bound in trace norm $\norm\cdot_1$
\begin{equation}\label{eq:trbd}
  \norm{\vartheta(x)f(H_L)}_1\le C\norm f_{C^{k_1}}\qquad\text{for all }f\in C_c^\infty(\Delta)
  \text{ with }\supp f\subset K,
\end{equation}
$k_1$ and $C$ depending only on $H$, $L$, $\Delta$, $\vartheta$ and the fixed compact $K\subset\Delta$. Consequently, for a pointwise termination, $i[H_L,\chi]\varphi'(H_L)$ is trace-class, so that the trace \eqref{eq:sigmaI} is well defined and finite.
\end{theorem}
\begin{proof}
Write $f=f\cdot f_2$ with $f_2\in C_c^\infty(\Delta)$ equal to $1$ on $K$, and
$\vartheta f(H_L)=\bigl(\vartheta f(H_L)\langle x\rangle\bigr)\bigl(\langle x\rangle^{-1}f_2(H_L)\bigr)$. By
\eqref{eq:decay} with $M=2$ both factors have square-integrable kernels: the first because
$\langle x'\rangle\le C\langle x-x'\rangle$ for $x\in\supp\vartheta$, the second because of the factor
$\langle x\rangle^{-1}$. So both are Hilbert--Schmidt, with norms bounded by $C\norm f_{C^{m(2)}}$ and a
constant, and their product is trace-class with \eqref{eq:trbd}, $k_1=m(2)=9$. The same holds for
$f(H_L)\vartheta$ by taking adjoints.

For the last claim, since $\chi=\chi(x)$, Leibniz's rule makes $i[H,\chi]$ a differential operator of order $\le p-1$ with coefficients supported in $\supp\chi'$. Applying it to the kernel of $\varphi'(H_L)$ in the first pair of variables preserves \eqref{eq:decay}, bounding {all} derivatives with the same exponential factor, and the support condition provides the compactly supported factor $\vartheta$. So each such term again obeys \eqref{eq:decay} and the above argument applies.
\end{proof}

\paragraph{Fedosov formula and quantization.}\label{ssec:quant}

Define
\[
  g:=e^{2\pi i\varphi}-1\in C_c^\infty(\Delta),\qquad
  U:=e^{2\pi i\varphi(H_L)}=\mathbf 1+\mathsf W,\qquad \mathsf W:=g(H_L),
\]
so that $\mathsf W$ obeys \eqref{eq:decay} by Theorem~\ref{thm:sG}. With
$\overline g=e^{-2\pi i\varphi}-1$,
\begin{equation}\label{eq:alg}
  (1+\overline g)\,g'
  =e^{-2\pi i\varphi}\cdot 2\pi i\,\varphi'\,e^{2\pi i\varphi}
  =2\pi i\,\varphi' .
\end{equation}
Recall the set $\mathcal P$ of switch functions of Definition~\ref{def:sigmaI}. We show independence of the switch function; see also \cite{Graf2013,ProdanSchulzBaldes16,Bal2019,BalReview}.
\begin{lemma}\label{lem:comm-tc}
Let $P\in\mathcal P$ and let $K\subset\Delta$ be compact. There are $C$ and $m$ such that for
$f\in C_c^\infty(K)$:
\begin{enumerate}[label=\textup{(\roman*)},nosep]
\item $\norm{[P,f(H_L)]}_1\le C\norm f_{C^m}$;
\item $\Tr[P,f(H_L)]=0$.
\end{enumerate}
In particular $[P,U]$ and $U^*PU-P=U^*[P,\mathsf W]$ are trace-class.
\end{lemma}
\begin{proof}
(i) Let $P=0$ on $(-\infty,-R]$ and $P=1$ on $[R,\infty)$. The kernel of $[P,f(H_L)]$ is
$(P(x)-P(x'))K_f$, and it vanishes unless $\max(|x|,|x'|)\le R+|x-x'|$. By \eqref{eq:decay} with $M=6$ it is
therefore bounded by $C\norm f_{C^{m}}\langle x\rangle^{-3}\langle x'\rangle^{-3}e^{-c(y+y')}$, with
$m=m(6)$. With $f_2\in C_c^\infty(\Delta)$ equal to $1$ near $K$,
\[
[P,f]=\bigl([P,f]\langle x\rangle\bigr)\bigl(\langle x\rangle^{-1}f_2\bigr)+\bigl(f\langle x\rangle^{-1}\bigr)\bigl(\langle x\rangle[P,f_2]\bigr)
\]
is a sum of products of Hilbert--Schmidt operators. The factors containing $f$ have norms at most
$C\norm f_{C^m}$. No smoothness of $P$ is used.

(ii) Write $W=f(H_L)$ and $P=\Pi+Q$, with $Q$ bounded and compactly supported in $x$.

The trace-class operator $[\Pi,W]$ has vanishing diagonal blocks with respect to $\Pi$, so its trace
is $0$.

With $\vartheta\in C_c^\infty(\R)$ satisfying $\vartheta Q=Q$, the operators $\vartheta W$ and
$W\vartheta$ are trace-class (Theorem~\ref{thm:traceclass}). Then
\[
\Tr(QW)=\Tr(\vartheta W\,Q)=\Tr(W\vartheta Q\,\vartheta)=\Tr(WQ),
\]
by $\Tr(AB)=\Tr(BA)$ for $A$ trace-class and $B$ bounded.
\end{proof}
\noindent The next identity is of Fedosov type, related trace identities for switch functions appearing in
\cite{EGS05,ProdanSchulzBaldes16}; its proof needs no invariance of the domain under $P$.
\begin{lemma}\label{lem:FD}
Let $P\in\mathcal P$ and $\mathfrak t(u,v):=\Tr\bigl(u(H_L)[P,v(H_L)]\bigr)$ for
$u,v\in C_c^\infty(\Delta)$. There is $c\in\C$ with
\begin{equation}\label{eq:FD}
  \mathfrak t(u,v)=c\int_\Delta u\,v'\,dE .
\end{equation}
\end{lemma}
\begin{proof}
Write $u$ for $u(H_L)$. Lemma~\ref{lem:comm-tc} and $\Tr(AB)=\Tr(BA)$ give four properties.
\begin{enumerate}[label=\textup{(\alph*)},nosep]
\item $\mathfrak t(u,v)+\mathfrak t(v,u)=\Tr[P,uv]=0$.
\item $\mathfrak t(u,vw)=\mathfrak t(uv,w)+\mathfrak t(wu,v)$.
\item $\mathfrak t(u,v)=0$ if $\supp u\cap\supp v=\emptyset$: apply (b) with $w=1$ near $\supp v$ and
$w=0$ near $\supp u$.
\item $\mathfrak t(u,\theta)=0$ if $\theta=1$ near $\supp u$: by (b) and (c),
$\mathfrak t(u,\theta)=\mathfrak t(u,\theta^2)=2\mathfrak t(u,\theta)$.
\end{enumerate}

For a compact $K\subset\Delta$ fix $e\in C_c^\infty(\Delta)$ with $e(E)=E$ near $K$, and let
$\lambda(\rho):=\mathfrak t(\rho,e)$ for $\rho\in C_c^\infty(K)$. By (c), $\lambda$ is a well-defined
linear functional on $C_c^\infty(\Delta)$, with $|\lambda(\rho)|\le C_K\norm\rho_\infty$, since
$|\Tr\rho(H_L)[P,e(H_L)]|\le\norm{\rho(H_L)}\,\norm{[P,e(H_L)]}_1\le\norm\rho_\infty\norm{[P,e(H_L)]}_1$ and
$C_K:=\norm{[P,e(H_L)]}_1<\infty$ by Lemma~\ref{lem:comm-tc}\textup{(i)}.

Let $u\in C_c^\infty(K)$. By (b) and induction,
$\mathfrak t(u,e^{n+1})=\mathfrak t(ue^n,e)+\mathfrak t(eu,e^n)=(n+1)\,\mathfrak t(ue^n,e)$. So
$\mathfrak t(u,Q(e))=\lambda(uQ')$ for polynomials $Q$ with $Q(0)=0$.

For general $v$, let $\theta\in C_c^\infty(\Delta)$ equal $1$ near $K$ with $e=E$ on $\supp\theta$, and
let $Q_n\to v$ in $C^m(\supp\theta)$. By (c) and (d),
$\mathfrak t(u,\theta Q_n)=\mathfrak t(u,(Q_n-Q_n(0))(e))=\lambda(uQ_n')$. Letting $n\to\infty$, using
Lemma~\ref{lem:comm-tc}(i) and (c), gives $\mathfrak t(u,v)=\lambda(uv')$.

Finally, if $\int\rho=0$ then $f:=\int_{-\infty}^E\rho\in C_c^\infty(\Delta)$. With $\theta=1$ near
$\supp f$, (a) gives $\lambda(\rho)=\lambda(f')=\mathfrak t(\theta,f)=-\mathfrak t(f,\theta)=-\lambda(f\theta')=0$.
Hence $\lambda=c\int$, with $c:=\lambda(\rho_0)$ for any $\rho_0$ with $\int\rho_0=1$.
\end{proof}
\begin{lemma}\label{lem:indep}
The constant $c$ of Lemma~\ref{lem:FD} does not depend on $P\in\mathcal P$, and $\Tr(U^*PU-P)=2\pi i\,c$ for
every switch $\varphi$. In particular the smooth switch and every sharp projection $\Pi_{x_0}$ give the same
value.
\end{lemma}
\begin{proof}
By Lemma~\ref{lem:comm-tc}\textup{(ii)}, \eqref{eq:FD}, \eqref{eq:alg} and $\int g'=0$,
$\Tr(U^*PU-P)=\Tr[P,\mathsf W]+\mathfrak t(\overline g,g)=c\int(1+\overline g)g'\,dE=2\pi i\,c$, for every
switch $\varphi$. For $P,P'\in\mathcal P$, $\mathcal Q:=P-P'$ has compact support, so that $\mathcal Qv(H_L)$ and
$(uv)(H_L)\mathcal Q$ are trace-class (Theorem~\ref{thm:traceclass}) and cyclicity gives
$\mathfrak t_P(u,v)-\mathfrak t_{P'}(u,v)=\Tr(u\,\mathcal Qv)-\Tr(uv\,\mathcal Q)=\Tr(\mathcal Q\,vu)-\Tr(\mathcal Q\,uv)=0$.
\end{proof}

\noindent The next formula is the Fedosov formula for the Fredholm module $(\Pi,U)$
\cite{ASS94,ProdanSchulzBaldes16}.
\begin{lemma}\label{lem:fedproj}
Let $U$ be unitary, $\Pi$ an orthogonal projection, and $[\Pi,U]$ trace-class. Then $\Pi U\Pi|_{\Ran\Pi}$ is Fredholm, $U^*\Pi U-\Pi$ is trace-class, and
\[
  \Ind\bigl(\Pi U\Pi\big|_{\Ran\Pi}\bigr)=-\Tr\bigl(U^*\Pi U-\Pi\bigr).
\]
\end{lemma}
\noindent For the sign: with $T=\Pi U\Pi|_{\Ran\Pi}$ and $\Pi^\perp=\mathbf 1-\Pi$, $\mathbf 1-T^*T=\Pi U^*\Pi^\perp U\Pi$ and
$\mathbf 1-TT^*=\Pi U\Pi^\perp U^*\Pi$ are trace-class, so $\Ind T$ is the difference of their traces, while the
diagonal blocks of $U^*\Pi U-\Pi$ have traces $-\Tr(\Pi U^*\Pi^\perp U\Pi)$ and $\Tr(\Pi U\Pi^\perp U^*\Pi)$.
\begin{theorem}\label{thm:quantisation}
Under the hypotheses of Theorem~\ref{thm:sG}, let $U=e^{2\pi i\varphi(H_L)}$, $\Pi=\Pi_{x_0}$ and
$P\in\mathcal P$.
\begin{enumerate}[label=\textup{(\roman*)},nosep]
\item $[\Pi,U]$ is trace-class, $\Pi U\Pi$ is Fredholm on $\Ran\Pi$, and
\begin{equation}\label{eq:quant}
  \mathcal I=\Ind\bigl(\Pi U\Pi\big|_{\Ran\Pi}\bigr)=-\Tr\bigl(U^*PU-P\bigr)\ \in\ \Z .
\end{equation}
\item For $u,v\in C_c^\infty(\Delta)$,
\begin{equation}\label{eq:trPU}
  \Tr\bigl(u(H_L)[P,v(H_L)]\bigr)=\frac{i\mathcal I}{2\pi}\int_\Delta u\,v'\,dE .
\end{equation}
\item For $\varphi_1,\varphi_2\in C_c^\infty(\Delta;\R)$, $\cT_P$ is trace-class and
\[
  \Tr\,\cT_P=\frac{\mathcal I}{2\pi}\int_\Delta\varphi_1\varphi_2\,dE .
\]
In particular $2\pi\sigma_I[H_L]=\mathcal I$ for every factorization of $\varphi'$ and every
$P\in\mathcal P$.
\end{enumerate}
\end{theorem}
\begin{proof}
(i) By Lemma~\ref{lem:comm-tc}, $[\Pi,U]$ is trace-class. Lemma~\ref{lem:fedproj} gives the index
formula, and Lemma~\ref{lem:indep} replaces $\Pi$ by $P$.

(ii) By (i) and Lemma~\ref{lem:indep}, $-\mathcal I=2\pi i\,c$, and \eqref{eq:FD} gives \eqref{eq:trPU}.

(iii) Since $(E\varphi_1)\varphi_2=\varphi_1(E\varphi_2)$, we have
$\cT_P=i\bigl((E\varphi_1)(H_L)[P,\varphi_2(H_L)]-\varphi_1(H_L)[P,(E\varphi_2)(H_L)]\bigr)$, which is trace-class by
Lemma~\ref{lem:comm-tc}. By (ii),
\[
\Tr\,\cT_P=-\frac{\mathcal I}{2\pi}\int\bigl(E\varphi_1\varphi_2'-\varphi_1(E\varphi_2)'\bigr)dE
=\frac{\mathcal I}{2\pi}\int\varphi_1\varphi_2\,dE .
\]
\end{proof}
\noindent Only \eqref{eq:decay} with the constants of Theorem~\ref{thm:sG}, Theorem~\ref{thm:traceclass}
and Lemmas~\ref{lem:comm-tc}--\ref{lem:fedproj} were used. The theorem therefore also holds on the spectral
route (Theorem~\ref{thm:fiberexist}), and for tangential derivatives.

\paragraph{Stability and locality.}\label{ssec:stab} By Theorem~\ref{thm:quantisation}, every statement below on $\mathcal I$ holds for $\sigma_I=\mathcal I/2\pi$.

\begin{proposition}\label{prop:normres}
Let $s\mapsto(H_s,T_s)$, $s\in[0,1]$, satisfy \textup{[H1]--[H4]} with constants uniform in $s$ and
coefficients continuous in $s$ in $C_b^\infty$, and let $T_s$ be boundary conditions \eqref{eq:traceop}
with $B_j$ differential, row orders independent of $s$ and coefficients continuous in $s$, such that each $T_s$ is a termination, with field $L_s$ \eqref{eq:Lfield}, and the Shapiro--Lopatinski condition holds uniformly in $s$. Then, for any fixed $z_0\in\C\setminus\R$, $s\mapsto(H_{s,L_s}-z_0)^{-1}$ is norm-continuous, and so is $s\mapsto f(H_{s,L_s})$ for every $f\in\C\oplus C_0(\R)$, in particular for $U_s=e^{2\pi i\varphi(H_{s,L_s})}$.
\end{proposition}
\begin{proof}
The $H_{s,L_s}$ act on different domains; a right inverse of the trace operator corrects one domain into another.

\emph{The a priori estimate.} Summing Lemma~\ref{lem:ellreg}\textup{(i)} with $n=1$ over unit boxes, whose triples $B^*$ have bounded overlap, and using $\norm u\le\norm{(H_{s,L_s}-z_0)u}/|\Im z_0|$,
\begin{equation}\label{eq:apriori-unif}
  \norm u_{H^p(\Om)}\le C\,\norm{(H_{s,L_s}-z_0)u},\qquad u\in\cD_{L_s},
\end{equation}
with $C$ independent of $s$ by uniformity of the Shapiro--Lopatinski condition. Writing
$R_s:=(H_{s,L_s}-z_0)^{-1}$, this says that $\norm{R_s}_{L^2\to H^p}\le C$ uniformly in $s$.

\emph{A correction.} By hypothesis $T_0$ has a bounded right inverse $R_{T_0}:\partial\cH\to H^p(\Om)$, with $\partial\cH$ as in Step~3 of the proof of Lemma~\ref{lem:CT}. Put $\alpha(s):=\norm{H_0-H_s}_{H^p\to L^2}$ and $\beta(s):=\norm{T_0-T_s}_{H^p\to\partial\cH}$; both tend to $0$ as $s\to0$, the coefficients being continuous in $s$ and the row orders fixed.

\emph{The estimate.} Fix $h\in L^2(\Om)$, put $u:=R_sh\in\cD_{L_s}$ and $v:=u-R_{T_0}T_0u$.
(a) Since $T_su=0$, $T_0u=(T_0-T_s)u$, so by \eqref{eq:apriori-unif} $\norm{u-v}_{H^p}\le C\beta(s)\norm h$.
(b) $T_0v=0$, so $v\in\cD_{L_0}$, and
\[
  (H_{0,L_0}-z_0)v=h+(H_0-H_s)u-(H_0-z_0)(u-v).
\]
Applying $R_0$, bounded on $L^2$, and using (a) and \eqref{eq:apriori-unif}, $\norm{v-R_0h}\le C(\alpha(s)+\beta(s))\norm h$.
(c) Hence $\norm{(R_s-R_0)h}\le\norm{u-v}+\norm{v-R_0h}\le C(\alpha(s)+\beta(s))\norm h$, so $\norm{R_s-R_0}\to0$ as $s\to0$. The same argument with any base point $s_0\in[0,1]$ in place of $0$ gives continuity on $[0,1]$.

\emph{Functional calculus.} For $f=c+f_0$, $f_0\in C_0(\R)$, approximate $f_0$ uniformly by polynomials in $(E-z_0)^{-1}$ and $(E-\bar z_0)^{-1}$ (Stone--Weierstrass), using $\norm{f_0(H_{s,L_s})}\le\norm{f_0}_\infty$.
\end{proof}

\begin{corollary}[homotopy invariance and independence of the switch]\label{cor:stab}
\begin{enumerate}[label=\textup{(\roman*)},nosep]
\item $\mathcal I$ is constant along every path meeting the hypotheses of
Proposition~\ref{prop:normres}. This holds for a path of admissible pointwise terminations under \textup{[S]} and \textup{[M]} by  Remark~\ref{rem:SL};
\item $\mathcal I$ does not depend on the choice of switch $\varphi$.
\end{enumerate}
\end{corollary}
\begin{proof}
(i) Uniform Shapiro--Lopatinski ellipticity gives uniform constants in \eqref{eq:decay}, hence a uniform trace-norm bound on $[\Pi,U_s]$ by Lemma~\ref{lem:comm-tc}, so every $\Pi U_s\Pi+(\mathbf 1-\Pi)$ is Fredholm; and $s\mapsto U_s$ is norm-continuous by Proposition~\ref{prop:normres}. So the index is locally constant.

(ii) is Lemma~\ref{lem:indep} with \eqref{eq:quant}.
\end{proof}

\begin{theorem}[locality]\label{thm:local}
Let $(H,T)$ and $(H',T')$ both satisfy \textup{[H1]--[H4]}, with $B_j,B'_j$ differential along
$\Gam$ and uniform Shapiro--Lopatinski ellipticity, and let their coefficients
coincide outside a compact set $K\subset\overline\Om$. Then $U-U'$ is compact and
$\mathcal I(H,T)=\mathcal I(H',T')$. In particular, if $(H,T)$ coincides outside a compact set with an
$x$-independent pair $(H_\infty,T_\infty)$ satisfying the same hypotheses, then
$\mathcal I(H,T)=\mathcal I(H_\infty,T_\infty)$, computed in sections~\ref{sec:APS}--\ref{sec:range}.
\end{theorem}
\begin{proof}
Let $\theta\in C_c^\infty(\R^2)$ equal $1$ near $K$ and $\tilde\theta:=1-\theta$, so that $H=H'$ and $T=T'$ near
$\supp \tilde\theta$. Write $R,R'$ for the resolvents at a fixed $z\notin\R$ and $R_T$ for a bounded right inverse
of $T$ (Step~3 of the proof of Lemma~\ref{lem:CT}). For $f\in L^2$ and $w:=R'f$, locality gives
$T(\tilde\theta w)=[T,\tilde\theta]w$, $\tilde\theta$ vanishing where $T\ne T'$; so $v:=\tilde\theta w-R_T[T,\tilde\theta]w\in\cD_L$ and
$(H-z)v=\tilde\theta f+[H,\tilde\theta]w-(H-z)R_T[T,\tilde\theta]w$. Applying $R$,
\[
  \tilde\theta R'-R\tilde\theta=R[H,\tilde\theta]R'+\bigl(\mathbf 1-R(H-z)\bigr)R_T[T,\tilde\theta]R' .
\]
$[H,\tilde\theta]$ has order $\le p-1$ and coefficients supported in $\supp\nabla\theta$, and $R'$ is bounded
$L^2\to H^p$ by ellipticity, so $[H,\tilde\theta]R'$ is compact by Rellich. By Leibniz in $y$ and in $x$, $[T,\tilde\theta]=\sum_j[B_j,\tilde\theta|_\Gam]\gamma_j+\sum_{j}\sum_{1\le k\le j}\binom jkB_j(D_y^k\tilde\theta)|_\Gam\gamma_{j-k}$
has rows of order one less than those of $T$ and coefficients supported in $\supp\nabla\theta$, so
$[T,\tilde\theta]R'$ is compact into $\partial\cH$ by the trace theorem and Rellich on $\Gam$. Hence $\tilde\theta R'-R\tilde\theta$ is compact, and so are $\theta R'$ and $R\theta$, whence $R-R'$, for every $z\notin\R$. By \eqref{eq:HS}, $U-U'=g(H_L)-g(H'_{L'})=-\frac1\pi\int\dbar\widetilde g(z)\bigl(R(z)-R'(z)\bigr)\,dA(z)$, a norm-convergent integral of compact operators, hence compact.

Put $V:=UU'^*$, unitary with $V-\mathbf 1$ compact. Then
$\Pi U\Pi=\Pi V\Pi\cdot\Pi U'\Pi+\Pi(V-\mathbf 1)(\mathbf 1-\Pi)U'\Pi$, the last term compact, so
$\Ind(\Pi U\Pi)=\Ind(\Pi V\Pi)+\Ind(\Pi U'\Pi)$ on $\Ran\Pi$, the Fredholm property coming from
Theorem~\ref{thm:traceclass} and Lemma~\ref{lem:comm-tc}. Finally $\Pi V\Pi|_{\Ran\Pi}=\mathbf 1+\Pi(V-\mathbf 1)\Pi$ is the identity plus a compact operator, so $\Ind(\Pi V\Pi)=0$.
\end{proof}

\begin{theorem}[lower-order terms on the collar]\label{thm:collar}
Let $H$ and $L$ be $x$-independent, satisfy \textup{[H1]--[H4]}, and let the principal symbol not depend on $y$ on the
collar,
\begin{equation}\label{eq:flatsymbol}
  \sigma_p(y;\xi,\zeta)=\sigma_p(\xi,\zeta)\qquad\text{for }0\le y\le y_0 .
\end{equation}
Let $H^\circ$ be obtained from $H$ by freezing $a^w$ at $y=y_0$, so that
$H^\circ=H^{\R^2}$ on $\Om$, put $\Vp:=H-H^\circ$, of order $\le p-1$ by \eqref{eq:flatsymbol} and
supported in $\{y<y_0\}$, and $H_t:=H-t\Vp$. With $K_1:=\cA_H^{-1}\cA_\Vp$ and
$\Xi_t:=(\mathbf 1-tK_1)^{-1/2}$, so that $\Xi_t^*\cA_{H_t}\Xi_t=\cA_H$, let $L_t:=\Xi_tL$ and
$L^\circ:=L_1$. If $(H_t,L_t)$ is uniformly Shapiro--Lopatinski elliptic for $t\in[0,1]$, which is automatic at
$p=1$ for admissible $L$ since then $\cA_\Vp=0$ and $L_t=L$, then
\begin{equation}\label{eq:collar}
  \mathcal I[H,L]=\mathcal I[H^\circ,L^\circ] .
\end{equation}
\end{theorem}
\begin{proof}
By Proposition~\ref{prop:green}\textup{(iv)} $\cA_\Vp$ vanishes on the blocks $j+l\ge p-1$, so by
Lemma~\ref{lem:antiinv} $K_1$ raises the filtration index and Lemma~\ref{lem:tricong} applies. Each
$\Xi_t(\xi)$ is a polynomial Fourier multiplier whose $\xi^d$-coefficient raises the slot index by at least $d$,
so $L_t$ is a differential boundary condition with the row orders of $L$, Lagrangian for $\cA_{H_t}(\xi)$ at every
$\xi$, hence self-adjoint by Proposition~\ref{prop:SA}\textup{(i)}. \textup{[H3]} and \textup{[H4]} hold
along the path, $\Vp$ vanishing for $y\ge y_0$, and Corollary~\ref{cor:stab}\textup{(i)} gives
\eqref{eq:collar}. At $p=1$ ellipticity is admissibility, and the endpoints $V_\pm$ of $\sigma_p$ do not move along the path thanks to
\eqref{eq:flatsymbol}.
\end{proof}
\noindent At $p=1$, $\cA_\Vp=0$ and $L_t=L$, so the $x$-independence of $H$ and $L$ is not needed when $L$ is pointwise: the $(H_t,L)$ are terminations by Proposition~\ref{prop:SA}\textup{(i)}, and the argument is unchanged.

\paragraph{The spectral route.}\label{ssec:fiberroute}
For $x$-independent data,  the realization fibers over $\xi$, Lemma~\ref{lem:essspec}, and Theorem~\ref{thm:edge} suffice to define the invariant without ellipticity condition on $(H,L)$; see also \cite{Bal2022,graf2021topology,tauber2020anomalous,JudTauber,Gontier23}.
The hypothesis replacing ellipticity is escape (Definition~\ref{def:escape}).  We first give a sufficient condition for it and show that it cannot be avoided.

\medskip\noindent\textbf{Escape.} For $K\subset\Delta$ compact, $S(K)$ of Definition~\ref{def:escape} is closed since by Theorem~\ref{thm:edge} it is the restriction to $\R_\xi$ of the zero set of the continuous function $(\xi,E)\mapsto\det\bigl(U_L(\xi)-U_+(\xi,E)\bigr)$ on $\R_\xi\times K$ (Lemma~\ref{lem:analytic}). The edge spectrum \emph{escapes in measure} if every $|S(K)|$ is finite, which escape implies.
\begin{proposition}\label{prop:exact}
Let $H$ and $L$ be $x$-independent.
\begin{enumerate}[label=\textup{(\roman*)},nosep]
\item If $L$ is admissible and $U_+(\xi,E)\to U_+(\pm)$ uniformly for $E$ in compact subsets of $\Delta$, the edge spectrum escapes.
\item If the edge spectrum does not escape in measure, then for some $u\in C_c^\infty(\Delta)$ and every $\rho\in L^2(\R)\setminus\{0\}$, $\rho(x)\,u(H_L)$ is not trace-class.
\end{enumerate}
\end{proposition}
\noindent The uniformity in (i) holds at $p=1$, and under \textup{[M]} for the transported data $V^\Phi_\pm$, which under \textup{[S]} is \textup{[M]} itself (Lemma~\ref{lem:endpoints}, Corollary~\ref{cor:graded}\textup{(iii)} and the remark following it). By (ii), without escape in measure the localization of the density of states fails, and with it the hypothesis under which $\mathcal I$ is shown to exist. Escape is the right hypothesis, and (i) makes admissibility a verifiable sufficient condition for it.
\begin{proof}
Let $E_j(\xi)$ enumerate the eigenvalues of $\Hh_L(\xi)$ in $\Delta$ with multiplicity. By Lemma~\ref{lem:essspec} and Theorem~\ref{thm:edge}, $u(\Hh_L(\xi))$ has finite rank with
$\norm{u(\Hh_L(\xi))}_2^2=\sum_j|u(E_j(\xi))|^2$, $\norm\cdot_2$ the Hilbert--Schmidt norm, so Plancherel in the convolution variable gives, for
$u\in C_c^\infty(\Delta)$ and $\rho\in L^2(\R)$,
\[
  \norm{u(H_L)\,\rho}_2^2=\frac{\norm\rho_{L^2}^2}{2\pi}\int_\R\sum_j\bigl|u(E_j(\xi))\bigr|^2d\xi .
\]

\textup{(i)} Admissibility is $\det\bigl(U_{L(\pm)}-U_+(\pm)\bigr)\ne0$, and transversality is an {open} condition. Given a compact $K\subset\Delta$, $U_{L(\xi)}\to U_{L(\pm)}$ by definition of the limits $L(\pm)$ (Definition~\ref{def:adm}) and $U_+(\xi,E)\to U_+(\pm)$ uniformly for $E\in K$ by hypothesis, so $\det\bigl(U_L(\xi)-U_+(\xi,E)\bigr)\ne0$ for all $E\in K$ once $|\xi|\ge R(K)$, and Theorem~\ref{thm:edge} turns that into the absence of eigenvalues in $K$ there. Hence $S(K)\subseteq[-R(K),R(K)]$ for every compact $K$, which is escape.

\textup{(ii)} Fix a compact $K\subset\Delta$ with $|S(K)|=\infty$ and $u\in C_c^\infty(\Delta)$ with
$u\equiv1$ on $K$. Then $\sum_j|u(E_j(\xi))|^2\ge1$ for every $\xi\in S(K)$, so the identity above is
infinite, and $\norm\cdot_2\le\norm\cdot_1$ rules out trace-class.
\end{proof}

\begin{theorem}\label{thm:fiberexist}
Let $H$ and $L$ be $x$-independent, with $L$ real-analytic in $\xi$ and escaping edge spectrum. Then the kernel of $f(H_L)$, $f\in C_c^\infty(\Delta)$, obeys \eqref{eq:decay} with $m(M)=M$, and the conclusions of Theorem~\ref{thm:quantisation} hold; in particular $2\pi\sigma_I=\mathcal I\in\Z$.
\end{theorem}
\noindent No ellipticity is assumed; $H_L$ is the direct integral of Proposition~\ref{prop:SA}\textup{(i)}. Pointwise and differential terminations are real-analytic in $\xi$, and admissible ones escape under the hypotheses of Proposition~\ref{prop:exact}\textup{(i)}.
\begin{proof}
\emph{Resolvent.} Write $\cD_\xi:=\{u\in H^p(\R_+):\gamma u\in L(\xi)\}$ and fix $\xi_0$. Real-analyticity
gives $G(\xi)\in GL(\cB)$, analytic near $\xi_0$, with $G(\xi_0)=\mathbf 1$ and
$G(\xi)L(\xi_0)=L(\xi)$. Let $\mathcal E:\cB\to H^p(\R_+)$ be bounded, compactly supported, with
$\gamma\mathcal E=\mathbf 1$. Then $V(\xi):=\mathbf 1+\mathcal E(G(\xi)-\mathbf 1)\gamma$ is analytic and invertible on
$H^p$ near $\xi_0$ and maps $\cD_{\xi_0}$ onto $\cD_\xi$. For $z$ in a compact subset of the resolvent set
of $\Hh_L(\xi_0)$, the map $(\Hh(\xi)-z)V(\xi):\cD_{\xi_0}\to L^2$ is analytic and invertible near
$\xi_0$. Hence
\[
(\Hh_L(\xi)-z)^{-1}=V(\xi)\bigl[(\Hh(\xi)-z)V(\xi)\bigr]^{-1}
\]
is analytic in $\mathcal L(L^2,H^p)$, and $\{\Hh_L(\xi)\}$ is a self-adjoint holomorphic family
\cite[Ch.~VII \S\S1.2, 3.1]{kato1966perturbation}.

\emph{Branches.} Let $K:=\supp f\subset(a,b)$ with $a,b\in\Delta\setminus\spec\Hh_L(\xi_0)$. The Riesz
projection onto $(a,b)$ is analytic of constant finite rank near $\xi_0$. Rellich's theorem
\cite[Ch.~VII \S1.3 and Ch.~II, Thm.~6.1]{kato1966perturbation} gives real-analytic eigenvalues $\lambda_k$ and
orthonormal eigenfunctions $\psi_k$ in $H^p$ with $f(\Hh_L(\xi))=\sum_kf(\lambda_k)\psi_k\psi_k^*$ near
$\xi_0$, collisions included. For $y\ge y_0$,
$\psi_k$ is the first component of $e^{(y-y_0)M^\natural(\xi,\lambda_k)}P^s(\xi,\lambda_k)\Psi_k(\xi)$, with $\Psi_k$ its Cauchy data at $y_0$, analytic in $\xi$. By
\textup{[H3]} this is holomorphic in $\xi$ on a complex neighborhood, bounded there by $Ce^{-\nu y}$; Cauchy's estimates in $\xi$ and the equation in $y$ give $|\partial_\xi^r\partial_y^\alpha\psi_k|\le Ce^{-\nu y}$.

\emph{Kernel.} By escape, $f(\Hh_L(\xi))=0$ for $|\xi|>R$. Covering $[-R,R]$ by finitely many such
neighborhoods, the kernel $k_f(\xi;y,y')$ of $f(\Hh_L(\xi))$ is smooth, supported in $|\xi|\le R$, and
satisfies $|\partial_\xi^r\partial^\alpha k_f|\le C\norm f_{C^r}e^{-\nu(y+y')}$. Integrating by parts in
$K_f=\frac1{2\pi}\int e^{i(x-x')\xi}k_f\,d\xi$ gives \eqref{eq:decay}. Lemma~\ref{lem:comm-tc},
Theorem~\ref{thm:traceclass} and Theorem~\ref{thm:quantisation} then apply verbatim. Compact support in
$\xi$ replaces ellipticity. 
\end{proof}

\begin{remark}[escape and ellipticity]\label{rem:escape}
For $x$-independent data, uniform Shapiro--Lopatinski ellipticity implies escape. With $\mu=\pm\xi>0$ entering as a cotangent variable (regularity $+\infty$, \cite[\S2.1]{Grubb96}), the fiber problem $\{\Hh(\xi)-E,T(\xi)\}$ on $\R_+$ is uniformly parameter-elliptic \cite[Def.~3.1.4]{Grubb96}: condition (I) there is \textup{[H2]}, and (II)--(III), with no tangential variable, are the Shapiro--Lopatinski condition at $\xi=\pm1$. By \cite[Thm.~3.2.13]{Grubb96} it is invertible for $|\xi|\ge\mu_0$, with inverse of norm $O(|\xi|^{-p})$ on $L^2$; as $E$ enters at order zero, this holds uniformly for $E$ in compact sets. Either hypothesis gives existence, quantization and the index representation \eqref{eq:quant} of $\mathcal I$, but only ellipticity gives the norm continuity of Proposition~\ref{prop:normres} behind Corollary~\ref{cor:stab}, hence stability under $x$-dependent perturbations. The spectral route provides no topology in which $\mathcal I$ is locally constant. Example \textup{(E8)} shows how the two interact.
\end{remark}

\section{The bulk--edge correspondence}\label{sec:APS}

Throughout the section $H$ and $L$ are $x$-independent, except in Lemmas~\ref{lem:crossing}, \ref{lem:rate} and \ref{lem:plucker} which hold at each fixed $x$ for $x$-dependent data. We rewrite the integer
$\mathcal I$  in four ways,
\[
\begin{aligned}
  \mathcal I=\Ind\bigl(\Pi U\Pi+(\mathbf 1-\Pi)\bigr)
  &\ \longrightarrow\ \tfrac1{2\pi i}\!\int\!\tr\bigl(u^{-1}du\bigr)\\
  &\ \longrightarrow\ \sfl\bigl(\Hh_L\bigr)=-\muRS
  \ \longrightarrow\ \cR-\cR_L+\eta(+)-\eta(-) ,
\end{aligned}
\]
in Theorem~\ref{thm:three}  and Theorem~\ref{thm:APS},
defining the energy monotonicity that orients a crossing, the rotation $\cR$, and the
asymmetry $\eta$. We treat terminations on the wall (Proposition~\ref{prop:wall}),
differential terminations (Theorem~\ref{thm:main}), and describe the computation of the right-hand side.
For periodic tight-binding models with a Dirichlet truncation, corresponding steps (signed counts of edge branches, edge index equal to an edge current) appear in \cite{Hatsugai93,schulz2000simultaneous,Graf2013,AvilaSBVB13}, the last also using monotonicity in energy. 

\paragraph{Energy monotonicity.}\label{ssec:crossing}
The reduction of Theorem~\ref{thm:edge} turns the edge spectrum into the parameters at which an eigenphase
of $U_L^{-1}U_+(\xi,E)$ lies in $2\pi\Z$. Their direction of passage is fixed by the following property of the
Calder\'on path, which holds for every $L$.

\begin{lemma}\label{lem:crossing}
For every $\xi\in\R$ and $E\in\Delta$,
\begin{equation}\label{eq:Emono}
  i\,U_+(\xi,E)^*\,\partial_EU_+(\xi,E)>0 ,
\end{equation}
so that the eigenphases of $U_L^{-1}U_+(\xi,E)$ are strictly decreasing in $E$ for every $L$ independent of $E$.
\end{lemma}
\begin{proof}
Write $\Lm^+(\xi,E)=\{(B'(E)\alpha,\alpha)\}$ as a graph over the $q$-negative block, as
Proposition~\ref{prop:U} permits, in the transported frame. $\Phi(\xi)$ does not depend
on $E$ and is a congruence, so the computation below is unchanged by it. Fix $\alpha$ and let $\psi_E$ be the
$L^2(\R_+)$ solution of $\cL\psi_E=0$ with $\gamma\psi_E=v(E):=(B'(E)\alpha,\alpha)$, smooth in $E$ by
Lemma~\ref{lem:analytic}, and $\dot\psi:=\partial_E\psi_E$. By \eqref{eq:green} on $\R_+$, the boundary terms
at infinity vanishing,
\[
  i\ip{v}{\cA\,\partial_Ev}=\ip{\psi}{\cL\dot\psi}-\ip{\cL\psi}{\dot\psi}=\ip{\psi}{\cL\dot\psi}=\norm\psi^2>0 ,
\]
since $\cL\psi=0$ and, differentiating, $\cL\dot\psi=-(\partial_E\cL)\psi=\psi$. As $\partial_Ev=(\dot B'\alpha,0)$,
$\cA=\mathrm{diag}(D_+,-D_-)$ and $B'=D_+^{-1/2}U_+D_-^{1/2}$ by \eqref{eq:UL}, the left-hand side equals
$i\,\alpha^*B'^*D_+\dot B'\alpha=i\,\beta^*U_+^*\partial_EU_+\beta$ with $\beta=D_-^{1/2}\alpha$ arbitrary. This
is \eqref{eq:Emono}. For $W=U_L^{-1}U_+$ with $U_L$ independent of $E$, $-iW^*\partial_EW=-iU_+^*\partial_EU_+$ is
negative definite, and the eigenphases of a unitary family move by the eigenvalues of that matrix
\cite[Ch.~II, Thm.~5.4]{kato1966perturbation}.
\end{proof}
\noindent In the chart $L\mapsto U_L^{-1}$ over the $q$-positive block (proof of Proposition~\ref{prop:U}) the inequality reverses, since $i(U^{-1})^*\partial_E U^{-1}=-U(iU^*\partial_EU)U^*$, and every relative eigenphase changes sign. The chart \eqref{eq:UL} is thus the one in which the Calder\'on path turns clockwise as $E$ increases, and all signs below are tied to this monotonicity.

\paragraph{The edge spectrum and the spectral flow.}\label{ssec:edge} The object with intrinsic meaning is neither the termination nor the Calder\'on Lagrangian separately, but the \emph{relative unitary}
\[
  W(\xi):=U_{L(\xi)}^{-1}\,U_+(\xi,E_0)\in\U(\kp),\qquad E_0\in\Delta,
\]
the chart points being those of the transported data (Definition~\ref{def:Uplus}). A change of chart conjugates $W$, so $\spec W$ and $\arg\det W$ are intrinsic. Since $|\det(U_L-U_+)|=|\det(W-\mathbf 1)|$, the edge spectrum at $E_0$ is, by Theorem~\ref{thm:edge}, the locus where $W$ has eigenvalue $1$, with multiplicity $\dim\ker(W-\mathbf 1)$. We now count eigenvalue branches.
\begin{remark}\label{rem:branches}
Branches may collide. By Lemma~\ref{lem:analytic} and Theorem~\ref{thm:edge}, the edge spectrum is the zero set of the real-analytic edge determinant $F^\Phi$ of \eqref{eq:Fanalytic} on $\R_\xi\times\Delta$, a locally finite union of real-analytic arcs \cite[Ch.~VII \S1.3 and Ch.~II, Thm.~6.1]{kato1966perturbation}. Off the locally finite collision set it is a graph $E=E_j(\xi)$ of simple eigenvalues since $\partial_EF^\Phi\ne0$ by Lemma~\ref{lem:crossing}. We call the closures of the components of the zero set minus the collision set its arcs.
\end{remark}

\begin{definition}[spectral flow]\label{def:sf}
Let $\xi\mapsto T(\xi)$ be a self-adjoint family, continuous in the gap topology, with discrete spectrum in $\Delta$ and no eigenvalue at $E_0\in\Delta$ for $|\xi|$ large. Its \emph{spectral flow} $\sfl(T)$ across $E_0$ is the number of branches crossing $E_0$ upwards minus the number crossing downwards as $\xi$ runs from $-\infty$ to $+\infty$, in the analytic sense of \cite{APS75,Phillips}, which needs no regularity of the crossings.
\end{definition}
\noindent The fiber families $\Hh_L$ qualify: their resolvents are real-analytic in $\xi$ (proof of Theorem~\ref{thm:fiberexist}), their spectrum in $\Delta$ is discrete by Lemma~\ref{lem:essspec} and Theorem~\ref{thm:edge}, and escape (Definition~\ref{def:escape}) with $K=\{E_0\}$ confines the crossings of $E_0$ to a compact $\xi$-interval. With $K=[E_0,E_1]$ it shows that $\sfl(\Hh_L)$ does not depend on $E_0\in\Delta$, the crossings of $E_0$ and of $E_1$ counting the same branches.

\begin{theorem}\label{thm:three}
Let $L$ be real-analytic in $\xi$ with escaping edge spectrum, so that $\mathcal I$ is defined by Theorem~\ref{thm:fiberexist}, and let $E_j(\xi)$ enumerate the eigenvalues of $\Hh_L(\xi)$ in $\Delta$, and $u(\xi)=e^{2\pi i\varphi(\Hh_L(\xi))}$ the fibers of $U=e^{2\pi i\varphi(H_L)}$. Then
\begin{equation}\label{eq:three}
  \mathcal I=\underbrace{\frac1{2\pi i}\int_{\R_\xi}\tr\bigl(u^{-1}du\bigr)}_{\textup{(i) winding}}
  =\underbrace{\sfl\bigl(\Hh_L\bigr)}_{\textup{(ii) spectral}}
  =\underbrace{-\muRS\bigl(\gamma_{E_0};L\bigr)}_{\textup{(iii) symplectic}} ,
\end{equation}
and each member is the signed number of passages through $2\pi\Z$ of the continuous lifts of the eigenphases of $W$, a counterclockwise passage counting $+1$.
\end{theorem}
In \textup{(iii)}, $\muRS$ is the Maslov index of the pair $(\gamma_{E_0},L(\cdot))$ for the Hermitian form $\cA$ over $[-R,R]$, with $R$ such that, by escape, no crossing occurs for $|\xi|\ge R$: the sum of the signatures of the crossing forms \eqref{eq:crossdef}. The identity of \textup{(ii)} and \textup{(iii)} is the
Maslov--spectral-flow correspondence for Cauchy-data families in complex Hilbert spaces \cite{BBF98}; at each crossing, the crossing form is identified with the crossing operator of the operator family, as in \cite[Lemmas~22, 24 and Theorem~27]{Gontier23}. The correspondence needs no
regularity of the crossings, both sides being homotopy invariant. For $L$ independent of $\xi$ the computation in the proof of
Lemma~\ref{lem:crossing}, with $\xi$ in place of $E$, gives the crossing form
$\gamma\psi\mapsto-\int_0^\infty\psi^*(\partial_\xi\Hh)\psi\,dy$, which by Hellmann--Feynman is
$-\norm\psi^2E_j'(\xi_0)$ at a simple crossing; this accounts for the minus sign. The sign in \eqref{eq:three} descends through \textup{(i)} from the sign in Lemma~\ref{lem:fedproj}, and every orientation-dependent statement below is read off \eqref{eq:three}.  
\begin{proof}
From the above paragraph, it remains to prove \textup{(i)}$=$\textup{(ii)}, and that \textup{(ii)} is the count of eigenphase passages.

For \textup{(i)}, $\mathsf W=g(H_L)$ is the direct integral of $w:=u-\mathbf 1=g(\Hh_L)$, smooth, of rank the number of eigenvalues in $\supp g$, finite by Theorem~\ref{thm:edge}, and supported in a bounded set (proof of Theorem~\ref{thm:fiberexist}). For $P\in\mathcal P$, $U^*[P,U]=[P,\mathsf W]+\mathsf W^*[P,\mathsf W]$, and
$\Tr[P,\mathsf W]=0$ by Lemma~\ref{lem:comm-tc}\textup{(ii)}. For the second term, translation invariance and $\int(P(x+r)-P(x))\,dx=r$ make $\Tr(U^*PU-P)$ the first moment
$\int r\,\tr\bigl(K_{\mathsf W^*}(-r)K_{\mathsf W}(r)\bigr)dr$ of the kernel $K_{\mathsf W}(x-x')$ of $\mathsf W$, operator-valued in $y$ and rapidly decaying by
\eqref{eq:decay}, the trace being finite since $\mathsf W^*\langle x\rangle^{-1}$ and $\langle x\rangle[P,\mathsf W]$ are Hilbert--Schmidt (proof of Lemma~\ref{lem:comm-tc}); Parseval, with $\widehat{rK_{\mathsf W}}=i\partial_\xi w$ and $\int\tr(\partial_\xi w)\,d\xi=0$, turns it into $\frac i{2\pi}\int\tr(u^*\partial_\xi u)\,d\xi$, the classical index formula for Toeplitz operators with translation-invariant symbol \cite{ProdanSchulzBaldes16}. This is
\textup{(i)} by \eqref{eq:Idef}. No commutator with $H_L$ enters, so this holds for
$\xi$-dependent $L$ too. For \textup{(ii)}, off the collision set, we have
$\tr(u^{-1}\partial_\xi u)=2\pi i\,\partial_\xi\sum_j\varphi(E_j)$, the sum running over the finitely many
$E_j\in\supp g$. Integrating over the arcs of Remark~\ref{rem:branches} separately (a finite sum, escape confining
those that meet $\supp\varphi'$ to a compact $\xi$-interval) each arc contributes $\varphi$ evaluated at its two ends, hence $+1$, $-1$ or $0$ according as it
traverses $\supp\varphi'$ upward, downward, or not. The total is $\sfl$, which is
\textup{(ii)}.

For the eigenphase count, the edge spectrum is the preimage of $\Sigma_{U_L}$ under
$(\xi,E)\mapsto\Lm^+(\xi,E)$ by Theorem~\ref{thm:edge}: the crossings of the level $E_0$ by the branches $E_j$
are the $\xi$ at which an eigenphase $\theta$ of $W(\xi)=U_{L(\xi)}^{-1}U_+(\xi,E_0)$ lies in $2\pi\Z$. Fix a
compact interval $K\subset\Delta$. By escape and Remark~\ref{rem:branches}, finitely many arcs meet $\R\times K$,
each in a compact $\xi$-interval, and a non-constant real-analytic arc has finitely many critical values there.
So for $E_0\in K$ off a finite set of levels, every crossing of $E_0$ is by a single branch, off the collision
set, with $E_j'(\xi_0)\ne0$. There the eigenvalue $1$ of $W$ is simple, $\theta(\xi,E_j(\xi))\in2\pi\Z$ along the
arc, and $\partial_E\theta<0$ by Lemma~\ref{lem:crossing} gives $\sgn\partial_\xi\theta=\sgn E_j'(\xi_0)$: an
upward branch is a counterclockwise passage, for $\xi$-dependent $L$ as well. So the passage count at $E_0$ is
$\sfl$. Both are independent of $E_0\in K$: $\sfl$ as noted after Definition~\ref{def:sf}, and the passage count because,
by escape, the eigenphases at $\xi=\pm R$ stay off $2\pi\Z$ for $E_0\in K$, so that this integer depends
continuously on $E_0$. The exceptional levels are therefore irrelevant.
\end{proof}

\medskip
We now compute the invariant from symbolic information. We start with a bulk contribution to the edge invariant. Recall the notation of section~\ref{sec:green}: $\Lm^+(x,\xi,E)$ is the Calder\'on Lagrangian of Definition~\ref{def:calderon}, Lagrangian of dimension $\kp$ by Theorem~\ref{thm:calderonB}; $U_+(x,\xi,E)\in\U(\kp)$ is the image of $\Lm^\Phi(x,\xi,E)$ under the chart \eqref{eq:UL}, and $U_+(x,\pm)$ that of the transported endpoints $\Lm^\Phi(x,\pm)$ of Corollary~\ref{cor:graded}\textup{(iii)} (Definition~\ref{def:Uplus}).

\begin{definition}[the rotation of a path]\label{def:rot}
Let $\gamma:\R\to\Lag(\cA)$ be a continuous path with limits $\gamma(\pm\infty)$. Its \emph{rotation} is
\begin{equation}\label{eq:R}
  \cR[\gamma]:=\frac1{2\pi}\Bigl[\arg\det U_{\gamma(\xi)}\Bigr]_{-\infty}^{+\infty}\ \in\ \R ,
\end{equation}
the change of a continuous lift of $\arg\det U_{\gamma(\xi)}$. The same formula defines $\cR[W]$ for $W:\R\to\U(\kp)$ with limits. The \emph{rotation of the Calder\'on path} is $\cR:=\cR[\Lm^\Phi(\cdot,E_0)]$, and the \emph{winding of the termination} is $\cR_L:=\cR[L^\Phi]$.
\end{definition}
\noindent The lift exists, is unique up to $2\pi\Z$, and converges at $\pm\infty$, since $\det U_{\gamma(\xi)}$ converges. The rotation is \emph{additive},
\begin{equation}\label{eq:Radd}
  \cR\bigl[U_{\gamma_1}^{-1}U_{\gamma_2}\bigr]=\cR[\gamma_2]-\cR[\gamma_1] ,
\end{equation}
since $\det(U_{\gamma_1}^{-1}U_{\gamma_2})=\overline{\det U_{\gamma_1}}\det U_{\gamma_2}$, and it is an integer whenever $\det U_{\gamma(+\infty)}=\det U_{\gamma(-\infty)}$, in particular for loops. When two media are compared, $\cR[H_\pm]$ denotes $\cR$ formed for $H_\pm$. A variant $\cR[H;\Om_-]$ for an operator given on $\Om_-$ appears in Lemma~\ref{lem:gro}.

\begin{lemma}[rotation of the Calder\'on path]\label{lem:rate}
$\cR$ is well defined by Definition~\ref{def:rot}. Moreover
\begin{equation}\label{eq:Rint}
  \cR=\frac1{2\pi}\int_\R\tr\bigl(-i\,U_+(\xi)^*\partial_\xi U_+(\xi)\bigr)\,d\xi ,
\end{equation}
exhibiting $\cR$ as a \emph{local integral of the symbol} over the momentum line.
\end{lemma}
\begin{proof}
$\xi\mapsto\det U_+(\xi,E_0)\in\U(1)$ is continuous with limits $\det U_+(\pm)$, so $\cR$ is well
defined by Definition~\ref{def:rot}. Continuity holds because $\Lm^+(\xi,E)$ is the range of a Riesz projection depending continuously on $\xi$, by \eqref{eq:LamRiesz}, and the endpoint limits are Corollary~\ref{cor:graded}\textup{(iii)}.

The chart \eqref{eq:UL} is smooth near $\Lm^\Phi(\pm)$, so \eqref{eq:endrate} gives $\partial_\xi U_+=O(|\xi|^{-2})$. For \eqref{eq:Rint}, $U_+^*\partial_\xi U_+$ is anti-Hermitian, $U_+$
being unitary, so $\tr\bigl(-i\,U_+^*\partial_\xi U_+\bigr)$ is real and equals
$\partial_\xi\arg\det U_+$. The integral converges absolutely because
$|\tr(-iU_+^*\partial_\xi U_+)|\le\kp\norm{\partial_\xi U_+}=O(|\xi|^{-2})$, and integrating gives the
total change \eqref{eq:R}.
\end{proof}

We introduce the final piece used in the bulk edge correspondence.
\begin{definition}[spectral asymmetry]\label{def:Reta}
For $W\in\U(\kp)$ with $1\notin\spec W$ let $\theta_1,\dots,\theta_\kp\in(0,2\pi)$ be its eigenphases and set
\begin{equation}\label{eq:eta}
  \eta(W):=\sum_{j=1}^\kp\Bigl(\frac12-\frac{\theta_j}{2\pi}\Bigr)
  \in\Bigl(-\frac\kp2,\frac\kp2\Bigr),
\end{equation}
writing $\eta(L,\Lm):=\eta(U_L^{-1}U_\Lm)$ for transverse $L,\Lm\in\Lag(\cA)$. With $W=e^{2\pi i\cT}$, $\cT$ self-adjoint and $\spec\cT\subset(0,1)$, one has $\eta(W)=\tfrac12\tr(\mathbf 1-2\cT)$ the spectral asymmetry of $W$. Both $\eta(W)$ and $\spec W$ are conjugation-invariant, hence independent of any trivialization.
\end{definition}
\paragraph{The bulk--edge correspondence.}\label{ssec:general}

\begin{definition}[endpoint unitaries]\label{def:relunitary}
The \emph{endpoint unitaries} of the relative unitary $W$ are the limits $W(\pm):=\lim_{\xi\to\pm\infty}W(\xi)$, when they exist.
\end{definition}
\noindent They exist for differential $L$ by Lemma~\ref{lem:plucker} below. Admissibility (Definition~\ref{def:adm}) is then $1\notin\spec W(\pm)$, whereas escape alone gives $1\notin\spec W(\xi)$ for $|\xi|$ large without producing any limit. By \eqref{eq:Radd} the rotation of $W$ is $\cR-\cR_L$, and $\eta(W(\pm))=\eta(L(\pm),\Lm^\Phi(\pm))$.
\begin{theorem}\label{thm:APS}
Let $H$ and $L$ be as in Theorem~\ref{thm:three}.
\begin{enumerate}[label=\textup{(\roman*)},nosep]
\item For every $R$ with $1\notin\spec W(\xi)$ for $|\xi|\ge R$, which exists by escape, and with no limits of
$W$ assumed,
\begin{equation}\label{eq:APSfinite}
  \mathcal I=\frac1{2\pi}\Bigl[\arg\det W(\xi)\Bigr]_{-R}^{R}
  +\eta\bigl(W(R)\bigr)-\eta\bigl(W(-R)\bigr).
\end{equation}
\item If moreover $W(\pm)$ exist and $1\notin\spec W(\pm)$, as for admissible $L$, then (with the orientation of \eqref{eq:three}):
\begin{equation}\label{eq:APSgeneral}
  \mathcal I=\cR-\cR_L+\eta\bigl(W(+)\bigr)-\eta\bigl(W(-)\bigr).
\end{equation}
\item If moreover $L$ is a transported pointwise termination $L(\xi)=\Phi(\xi)L_0$
(Definition~\ref{def:transported}), then $\cR_L=0$ and
\begin{equation}\label{eq:APS}
  \mathcal I=\cR+\eta\bigl(L_0,\Lm^\Phi(+)\bigr)-\eta\bigl(L_0,\Lm^\Phi(-)\bigr).
\end{equation}
\end{enumerate}
\end{theorem}
\begin{proof}
\textup{(i)} By Theorem~\ref{thm:three}, $\mathcal I$ is the signed number of passages through $2\pi\Z$ of
the multiset of continuous lifts $\{\theta_j\}$ of the eigenphases of $W$, continuous even where
individual branches collide by Remark~\ref{rem:branches}; $\sum_j\lceil\theta_j/2\pi\rceil$ is a symmetric
function of it. Since there are no passages outside $[-R,R]$, $\mathcal I$ is given by the \emph{unsplit} form
\begin{equation}\label{eq:unsplit}
  \mathcal I=\sum_j\Bigl(\bigl\lceil\theta_j(R)/2\pi\bigr\rceil
  -\bigl\lceil\theta_j(-R)/2\pi\bigr\rceil\Bigr).
\end{equation}
Insert $\lceil t\rceil=t+1-\{t\}$, valid since $\theta_j(\pm R)/2\pi\notin\Z$, and split. The first sum
is $\tfrac1{2\pi}\sum_j(\theta_j(R)-\theta_j(-R))=\tfrac1{2\pi}[\arg\det W]_{-R}^{R}$, because
$\sum_j\theta_j=\arg\det W$; the second is $\eta(W(R))-\eta(W(-R))$, because $\{\theta_j(\pm R)/2\pi\}$
is the eigenphase of $W(\pm R)$ in $(0,2\pi)$ divided by $2\pi$.

\textup{(ii)} Let $R\to\infty$ in \eqref{eq:APSfinite}. The first term tends to $\cR-\cR_L$ by
\eqref{eq:Radd}, and the other two to $\eta(W(\pm))$ because $1\notin\spec W(\pm)$ makes the
eigenphases of $W(\pm R)$ converge in $(0,2\pi)$ to those of $W(\pm)$, and $\eta$ is continuous there.

\textup{(iii)} The transported section is the constant $L_0$, so $\cR_L=0$ and
$W(\pm)=U_{L_0}^{-1}U_+(\pm)$, whence $\eta(W(\pm))=\eta(L_0,\Lm^\Phi(\pm))$ by \eqref{eq:eta}.
\end{proof}
\noindent Given Theorem~\ref{thm:three}, \eqref{eq:APSfinite} is the classical expression of a Maslov index by a winding and endpoint corrections \cite{CappellLeeMiller}. What is specific here is the identification of the endpoints (Lemma~\ref{lem:endpoints}) and the transport $\Phi$.

\noindent The unsplit \eqref{eq:unsplit} is manifestly an integer and a signed count. What the split brings is that the $L$-dependence sits in the two
asymmetry terms while $\cR$ is common to all admissible $L$. Form \textup{(i)} needs only the hypotheses
of Theorem~\ref{thm:three}. The existence of $W(\pm)$ buys the passage to $\xi=\pm\infty$ in
\textup{(ii)}, and a transported pointwise termination yields $\cR_L=0$ in \textup{(iii)}. The latter is used in examples (E7) and (E8); in the other examples $\Phi=\mathbf 1$, so that transported pointwise terminations are pointwise. The dependence of $\cR$ on $E_0$ and on the collar is discussed in Corollary~\ref{cor:bulkR}. The formula survives two generalizations: terminations on the wall, where admissibility fails, and terminations depending on the tangential momentum.

\medskip\noindent\textbf{Terminations on the wall.} Examples (E4) and (E5) require this generalization.

\begin{definition}[regularized endpoint asymmetry]\label{def:etareg}
Let $W(\xi)=U_L^{-1}U_+(\xi,E_0)$ and let $\eta$ be as in Definition~\ref{def:Reta}, computed from
the eigenphase representatives in $(0,2\pi)$. Put
\begin{equation}\label{eq:etareg}
  \eta^{\mathrm{reg}}(\pm):=\lim_{\xi\to\pm\infty}\eta\bigl(W(\xi)\bigr).
\end{equation}
\end{definition}

\begin{proposition}[correspondence at an inadmissible termination]\label{prop:wall}
Let $H$ and $L$ be as in Theorem~\ref{thm:three}, admissible or not, and let the limits $W(\pm)$ exist, as for differential $L$. Then the limits \eqref{eq:etareg} exist and
\begin{equation}\label{eq:APSwall}
  \mathcal I=\cR-\cR_L+\eta^{\mathrm{reg}}(+)-\eta^{\mathrm{reg}}(-),
\end{equation}
with $\cR_L=0$ when $L$ is a transported pointwise termination. For admissible $L$, $\eta^{\mathrm{reg}}(\pm)=\eta(W(\pm))$ and \eqref{eq:APSwall} is \eqref{eq:APSgeneral}.
Write $n^\flat(\pm)$ for the number of eigenphases of $W$ whose limit as $\xi\to\pm\infty$ lies in
$2\pi\Z$ and is approached \emph{from below}. If $\Lm^\Phi(+)=\Lm^\Phi(-)$ and $L(+)=L(-)$, as for even $p$ under
\textup{[S]}, \textup{[M]} and differential $L$ by \eqref{eq:Ehigheven} and Lemma~\ref{lem:plucker} below, then
\begin{equation}\label{eq:IReq}
  \mathcal I=\cR-\cR_L+n^\flat(-)-n^\flat(+),
\end{equation}
an integer correction bounded by $\dim\bigl(L\cap\Lm^\Phi(\pm)\bigr)$. In particular $\mathcal I=\cR-\cR_L$ for \emph{every} such $L$, admissible or not, exactly when the
two ends are approached from the same sides.
\end{proposition}
\begin{proof}
By escape, $1\notin\spec W(\xi)$ for $|\xi|\ge R$. Fix $j$ and let $\theta_j(\xi)\in(0,2\pi)$ be the representative; it is continuous on each half-line $\pm\xi\ge R$, and converges in $\R/2\pi\Z$ because $W(\xi)\to W(\pm)$. Its continuous lift converges, as after Definition~\ref{def:rot}, so it converges in $[0,2\pi]$, to $0$ or $2\pi$ when the limit is $0$ modulo $2\pi$, so \eqref{eq:etareg} exists. Moreover $\eta^{\mathrm{reg}}(\pm)=\eta^{[0,2\pi)}(\pm)-n^{\flat}(\pm)$, where
$\eta^{[0,2\pi)}$ assigns the value $0$ to a degenerate phase: a degenerate phase approached from
below contributes $\tfrac12-1$ to $\eta^{\mathrm{reg}}$ against $\tfrac12-0$ to $\eta^{[0,2\pi)}$, and one approached from above contributes the same quantity to both.

Theorem~\ref{thm:APS}\textup{(i)} holds at every $R'\ge R$. Letting $R'\to\infty$, the first term of \eqref{eq:APSfinite} tends to $\cR-\cR_L$ by \eqref{eq:Radd}, and the others to $\eta^{\mathrm{reg}}(\pm)$, which is
\eqref{eq:APSwall}. For admissible $L$ no limit phase lies in $2\pi\Z$, so $\eta^{\mathrm{reg}}(\pm)=\eta(W(\pm))$. If $\Lm^\Phi(+)=\Lm^\Phi(-)$ and $L(+)=L(-)$ then $W(+)=W(-)$, so
$\eta^{[0,2\pi)}(+)=\eta^{[0,2\pi)}(-)$, and \eqref{eq:IReq} follows. {The sides of approach need not agree}: the
two ends are separate components of $\{|\xi|\ge R\}$, and a phase may tend to $2\pi^-$ at one and to
$0^+$ at the other. For instance, for $\kp=1$ and $W(\pm)=1$, this is the case $\cR-\cR_L=-1$, $n^\flat(-)=1$, $n^\flat(+)=0$, where $W$
runs once round the circle inside $(0,2\pi)$ without crossing, and $\mathcal I=0$ correctly.
\end{proof}

\noindent The regularized asymmetries are needed in \textup{(E4)}, for $\partial_yu=sD_xu$ with $0<|s|<1$, where
the correction in \eqref{eq:IReq} cancels $-\cR_L$.

\medskip\noindent\textbf{The main theorem for general differential terminations.}\label{ssec:poly}
For an admissible section $\xi\mapsto L(\xi)$, Theorem~\ref{thm:APS}\textup{(ii)} reads, by the identities following Definition~\ref{def:relunitary},
\begin{equation}\label{eq:RL}
  \mathcal I=\cR-\cR_L+\eta\bigl(L(+),\Lm^\Phi(+)\bigr)-\eta\bigl(L(-),\Lm^\Phi(-)\bigr),\qquad
  \cR_L=\frac1{2\pi}\Bigl[\arg\det U_{L^\Phi(\xi)}\Bigr]_{-\infty}^{+\infty}.
\end{equation}
The termination enters only through its winding $\cR_L$ and its endpoints $L(\pm)=\lim_{\xi\to\pm\infty}L^\Phi(\xi)$. For a general section, such as a pseudodifferential termination of order zero, $L(+)$ and $L(-)$ may differ and $\cR_L$ need not be an integer. For a \emph{differential} termination the next lemma shows that $L(+)=L(-)$, so that $\cR_L$ is an integer, bounded by the degree; it also shows that these limits, used in Definition~\ref{def:adm}, exist when $p\ge2$ and the blocks $B_j$ have different orders.

\begin{lemma}[endpoint value for differential terminations]\label{lem:plucker}
Let $\cA$ be fixed and $L(\xi)=\ker B(\xi)\in\Lag(\cA)$ for real $\xi$, with $B(\xi)$ a
$\kp\times2\kp$ matrix polynomial of degree $\le d_B$. In the coordinates $\alpha=(\alpha_+,\alpha_-)$ of
Proposition~\ref{prop:U}, write $B\Upsilon^*=(B_+,B_-)$. Then
\begin{enumerate}[label=\textup{(\roman*)},nosep]
\item $B_\pm(\xi)$ is invertible for every real $\xi$, and
$U_{L(\xi)}=-D_+^{1/2}B_+(\xi)^{-1}B_-(\xi)\,D_-^{-1/2}$;
\item the endpoints $L(\pm)$ exist and \emph{coincide}, $L(+)=L(-)=:L_\infty\in\Lag(\cA)$, the \emph{common endpoint};
\item $\cR_L=n_+(\det B_-)-n_+(\det B_+)\in\Z$, where $n_+$ counts zeros in $\Im\xi>0$; hence $|\cR_L|\le\kp d_B$.
\end{enumerate}
\end{lemma}
\noindent The degree $d_B$ of $B$ is not the tangential order $d$ of the termination. Under \textup{[S]} the lemma applies with $d_B=d$, and in general to the transport $L^\Phi(\xi)=\ker\bigl(B(\xi)\Phi(\xi)\bigr)$, with $d_B\le d+p-1$ since $\Phi$ is polynomial of degree $\le p-1$ (Theorem~\ref{thm:graded}). For a transported pointwise termination $L^\Phi$ is constant and $\cR_L=0$.
\begin{proof}
(i) If $B_+v=0$ then $(v,0)\in L$, and isotropy gives $v^*D_+v=0$, so $v=0$. Hence $B_+$ is invertible, and the graph matrix $B_L$ of $L$ (called $B$ in Proposition~\ref{prop:U}) satisfies $B_+B_L+B_-=0$, i.e., $B_L=-B_+^{-1}B_-$. Then $B_-=-B_+B_L$ is invertible because $B_L$ is, and \eqref{eq:UL} gives the formula.

(ii) By (i) the entries of $U_{L(\xi)}$ are rational in $\xi$, and they are bounded on $\R$ because $U_{L(\xi)}$ is unitary. A rational function bounded on $\R$ has a finite limit at infinity, the same as $\xi\to+\infty$ and as $\xi\to-\infty$. So $U_{L(\xi)}$ converges at both ends to the same unitary matrix, whose preimage under \eqref{eq:UL} is $L_\infty$.

(iii) By (i) and \eqref{eq:detUL}, $\det U_{L(\xi)}=(-1)^\kp\varkappa_\cA\det B_-(\xi)/\det B_+(\xi)$, a quotient of two polynomials of degree $\le\kp d_B$ without real zeros. As $|\det U_L|=1$ they have the same modulus up to a constant on $\R$, hence the same degree $d_B'$. A polynomial $q$ of degree $d_B'$ without real zeros has $[\arg q]_{-\infty}^{+\infty}=\pi(n_+(q)-n_-(q))=\pi(2n_+(q)-d_B')$, since $[\arg(\xi-z)]_{-\infty}^{+\infty}=\pm\pi$ according as $\pm\Im z>0$. Subtracting the two gives (iii).
\end{proof}

\begin{theorem}[bulk--edge correspondence]\label{thm:main}
Let $L$ be an admissible termination, differential of order $\le d$ along $\Gam$, with escaping edge spectrum. Then its endpoints coincide, $L_\infty:=\lim_{\xi\to\pm\infty}L^\Phi(\xi)$, and
\begin{equation}\label{eq:APSpoly}
  \ \mathcal I=\cR-\cR_L+\eta\bigl(L_\infty,\Lm^\Phi(+)\bigr)-\eta\bigl(L_\infty,\Lm^\Phi(-)\bigr)\ ,
  \qquad \cR_L\in\Z,\quad |\cR_L|\le\kp(d+p-1).
\end{equation}
\end{theorem}
\noindent Admissibility gives escape under the uniformity of Proposition~\ref{prop:exact}\textup{(i)}, in particular at $p=1$ and under \textup{[M]} for the transported data, and reads $L_\infty\pitchfork\Lm^\Phi(\pm)$. Under \textup{[S]}, $|\cR_L|\le\kp d$.
\begin{proof}
The transport is an isometry $\cA(\xi)\to\cA(0)$ applied to both factors of $W$, so the spectrum of
$W$ and every change of $\arg\det W$ are frame-independent and \eqref{eq:APSgeneral} holds in the
transported frame. Lemma~\ref{lem:plucker}, applied to $L^\Phi$, makes the two endpoints of the termination coincide
and gives the integrality and the bounds.
\end{proof}

\noindent Every admissible termination in the examples of section~\ref{sec:examples} is covered by this
theorem. Section~\ref{sec:range} specializes it to closed Calder\'on paths and to transported pointwise
terminations, and section~\ref{sec:interface} to the doubled operator of an interface.

\begin{remark}[pseudodifferential terminations]\label{rem:PDO}
The invariant sees a termination only through its section $\xi\mapsto L(\xi)$ of \eqref{eq:Lfield}. Composing $T$ on the left with an invertible tangential operator, for instance $\mathrm{diag}(\langle D_x\rangle^{-m_i})$ on its rows, leaves $\ker T$, hence $H_L$ and $\mathcal I$, unchanged while changing the PDO orders of the $B_j$. A differential termination rescaled to order zero is the same realization.
For $x$-independent data under \textup{[S]} (in general for $B\Phi$), three cases occur.

\emph{(a) Rational sections.} If the symbol matrix is rational in $\xi$ without real poles, clearing denominators row by row, by polynomials without real zeros, does not change $L(\xi)$ or the rank, so Lemma~\ref{lem:plucker} and Theorem~\ref{thm:main} still apply, with $d_B$ the degree of the corresponding polynomial matrix. For instance, for $\kp=1$, $s=\pm1$, the order-zero section $L(\xi)$ with $U_{L(\xi)}=\bigl((1-is\xi)/(1+is\xi)\bigr)^k$ is the kernel of the row $(1,-U_{L(\xi)})$ in the chart. Clearing denominators gives the row $\bigl((1+is\xi)^k,-(1-is\xi)^k\bigr)$, of degree $k$, and $\cR_L=-k\sgn s$.

\emph{(b) Even sections.} If $L(-\xi)=L(\xi)$, as for functions of $\langle D_x\rangle$ such as $\partial_yu=\mu\langle D_x\rangle u$, the two limits coincide when they exist, a continuous lift of $\arg\det U_{L(\xi)}$ is even, and $\cR_L=0$. Then $\mathcal I(L)=\mathcal I(L_\infty)$, the value of the pointwise termination $L_\infty$.

\emph{(c) General classical symbols.} Otherwise the limits, which exist for classical symbols, are $L(\pm)=\ker\bigl(b^0_i(\pm1)\bigr)_i$ when the principal rows $b^0_i$ have rank $\kp$ at $\xi=\pm1$, and are in general different. Symbols odd in $\xi$, such as that of $D_x\langle D_x\rangle^{-1}$, produce this. Then \eqref{eq:RL} holds with two independent endpoints and $\cR_L\in\R$. Neither $\cR_L$ nor $\mathcal I$ is bounded at fixed pseudodifferential order: for $\kp=1$, the order-zero section $U_{L(\xi)}=\exp\bigl(i\pi t\,\xi\langle\xi\rangle^{-1}\bigr)$ has $\cR_L=t$ for every real $t$, with endpoints $e^{\pm i\pi t}$. Its fractional part is absorbed by the $\eta$-terms, which see these endpoints: modulo $\Z$, $\eta(W(+))-\eta(W(-))\equiv t-\frac1{2\pi}\bigl(\arg U_+(+)-\arg U_+(-)\bigr)\equiv t-\cR$, so $\mathcal I\in\Z$. Without limits, as for non-classical symbols, only \eqref{eq:APSfinite} remains, and escape becomes a genuine restriction.

In summary, sections~\ref{sec:obs} and \ref{sec:APS} apply to any termination whose section $L(\xi)$ is real-analytic with escaping edge spectrum, Theorem~\ref{thm:APS}\textup{(ii)} and Proposition~\ref{prop:wall} requiring in addition the limits $W(\pm)$. A single endpoint $L(+)=L(-)$ and an integer $\cR_L$, used in \eqref{eq:IReq}, Corollary~\ref{cor:loop} and section~\ref{sec:range}, hold in cases (a) and (b), but not in general in case (c).
\end{remark}

\paragraph{How the right-hand side is computed.}\label{ssec:symbolic}
The $\eta$-terms of \eqref{eq:APSpoly} are finite-dimensional linear algebra, of size at most $2\kp=pN$, and $\cR_L$ is a root count (Lemma~\ref{lem:plucker}\textup{(iii)}); only $\cR$, which depends on $H$ and not on the termination, needs a computation at each $\xi$ of a grid, each of the same size, followed by one quadrature. The computations of section~\ref{sec:examples} proceed as follows.
\begin{enumerate}[label=\textup{(\alph*)},nosep]
\item \emph{Chart.} Assemble $\cA$ from the coefficients on $\Gam$ by
\eqref{eq:blocks}, and diagonalize it once. This fixes the chart \eqref{eq:UL}. Without \textup{[S]},
compute the polynomial $\Phi(\xi)$ of \eqref{eq:gradedPhi} once as well.
\item \emph{Calder\'on path.} At each $\xi$ of a grid, an ordered Schur decomposition of the companion
matrix $M^\natural(\xi,E_0)$ gives an orthonormal basis of its stable subspace. For coefficients constant
in $y$, with $\sigma^\circ$ the full symbol, this is
\begin{equation}\label{eq:Lamflat}
  \Lm^+(\xi,E_0)=\mathrm{span}\bigl\{(c,\zeta c,\dots,\zeta^{p-1}c):\ \Im\zeta>0,\
  (\sigma^\circ(\xi,\zeta)-E_0)c=0\bigr\},
\end{equation}
with Jordan chains at multiple roots; otherwise the basis is carried to $y=0$ by \eqref{eq:LamRiesz}.
Multiplying by $\Phi(\xi)^{-1}$ and applying the chart gives $U_+(\xi,E_0)$.
\item \emph{Rotation.} $\cR$ is the change of a continuous lift of $\arg\det U_+(\xi,E_0)$, obtained by
unwrapping along the grid after compactifying the momentum line, or equivalently by the quadrature
\eqref{eq:Rint}.
\item \emph{Endpoints.} The same decomposition for the principal symbol at $\xi=\pm1$ gives $V_\pm$.
Under \textup{[P]},
\[
  \Lm^\Phi(\pm)=\bigoplus_{m=0}^{p-1}\pi_m\bigl(V^\Phi_\pm\cap G^m\bigr),
\]
with $V^\Phi_\pm$ as in \textup{[P]}, $G^m=\bigoplus_{j\le m}\cB_j$ and $\pi_m$ the projection onto $\cB_m$: a row reduction against the filtration. At $p=1$,
$\Lm^+(\pm)=V_\pm$; at even $p$ under \textup{[S]} and \textup{[M]}, $\Lm^+(\pm)=\cB_{\mathrm{high}}$. 
\item \emph{Termination and asymmetries.} The chart gives $U_{L_\infty}$. For a differential termination
$\cR_L=n_+(\det B_-)-n_+(\det B_+)$ by Lemma~\ref{lem:plucker}\textup{(iii)}, two root counts. The
eigenphases of $W(\pm)=U_{L_\infty}^{-1}U_+(\pm)$ in $(0,2\pi)$ give the asymmetries \eqref{eq:eta}.
\end{enumerate}
That the result is an integer checks all steps at once. One evaluation of $\cR$ determines it, with the
endpoint data, on a connected component of parameter space (Lemma~\ref{lem:Rconst},
Corollary~\ref{cor:bulkR}); compactly supported $x$-dependence does not enter
(Theorem~\ref{thm:local}), and under \eqref{eq:flatsymbol} neither does the collar, $L$ being replaced by
$L^\circ$ (Theorem~\ref{thm:collar}).

\section{Rigidity and the range of the invariant}\label{sec:range}

This section studies $\mathcal I$ as a function of the termination: when it is constant, how much it can vary, and which values it takes. Throughout, $H$ is $x$-independent and admissible terminations have escaping edge spectrum, so that $\mathcal I$ exists. A transported pointwise termination $\Phi L_0$ is a point $L_0$ of the compact $\Lag(\cA(0))\cong\U(\kp)$, and a differential termination $L$ reduces to one, $\mathcal I(L)=\mathcal I(\Phi L_\infty)-\cR_L$, with $L$ admissible exactly when $\Phi L_\infty$ is (Theorem~\ref{thm:main}). We therefore take $L=\Phi L_0$ unless stated otherwise. Our main tools are \eqref{eq:APS} and the chambers of Proposition~\ref{prop:admtopology}. Section~\ref{sec:examples} illustrates some of the results.

\paragraph{Rigidity.}\label{ssec:rigid}
The two $\eta$-terms of \eqref{eq:APS} are the only $L$-dependent part of the right-hand side, so
they are also the only mechanism by which $\mathcal I$ can depend on the termination.

\begin{corollary}[rigidity for closed Calder\'on path]\label{cor:loop}
If $\Lm^\Phi(+)=\Lm^\Phi(-)$, i.e., the Calder\'on path is a closed loop, then $\mathcal I=\cR-\cR_L$ for every admissible differential termination, and in particular
\[
  \mathcal I=\cR=\wind(\det U_+)\in\Z
\]
for every transported pointwise termination.
\end{corollary}
\begin{proof}
The two $\eta$-terms of \eqref{eq:APSpoly} coincide for every $L$, and $\cR$ is the winding of a loop,
hence an integer.
\end{proof}

\noindent The path closes for an even principal symbol, $\sigma_p(x;-\xi,\zeta)=\sigma_p(x;\xi,\zeta)$, which forces
$p$ to be even since $\sigma_p(-1,0)=(-1)^p\sigma_p(1,0)$ with $\sigma_p(1,0)$ invertible by \textup{[H2]}: under
\textup{[S]} and \textup{[P]}, $\Lm^+(x,\pm)=\lim_{t\to\infty}\Th_tV_\pm(x)$ and evenness gives $V_+=V_-$.

\begin{corollary}[even order is rigid]\label{cor:evenrigid}
Let $p$ be even and let \textup{[M]} and \textup{[S]} hold. Then
$\Lm^+(+)=\Lm^+(-)=\cB_{\mathrm{high}}$ and $\mathcal I=\cR-\cR_L$ for \emph{every}
admissible differential termination, $\mathcal I=\cR=\wind(\det U_+)\in\Z$ for every pointwise one. Without
\textup{[S]} the same holds for the transported path when $\Lm^\Phi(+)=\Lm^\Phi(-)$, and since
$\Phi^{-1}$ raises the trace order both transported endpoints are again $\cB_{\mathrm{high}}$ when the transported data $(\mathbf 1+\cN_\pm)V_\pm$ still miss $\cB_{\mathrm{low}}$.
Without \textup{[S]}, \textup{[M]} for the transported data cannot be dropped: \textup{(E8)}, where it fails, takes (numerically) all $\kp+1$ values at $p=2$ and $p=4$.
\end{corollary}
\begin{proof}
By \eqref{eq:Ehigheven}, $\Lm^+(\pm)=\cB_{\mathrm{high}}$, and Corollary~\ref{cor:loop} applies.
\end{proof}

\begin{lemma}[scalar principal symbols at second order]\label{lem:scalarp2}
Let $p=2$ and $\sigma_p=\phi(x;\xi,\zeta)\,\mathsf S(x)$ with $\phi$ real scalar and $\mathsf S$ Hermitian invertible, even or not.
Then \textup{[M]} holds, and by \eqref{eq:Ehigheven}
\[
  \Lm^+(x,+)=\Lm^+(x,-)=\bigl\{\gamma:\gamma_0=0\bigr\},
\]
the Dirichlet Lagrangian, so that Corollary~\ref{cor:evenrigid} applies under \textup{[S]}. The inadmissible pointwise
terminations are exactly those meeting the Dirichlet Lagrangian, and for definite $\mathsf S$ each of them is
nevertheless elliptic (Remark~\ref{rem:SL}).
\end{lemma}
\begin{proof}
The scalar $\phi(\pm1,\cdot)$ has one simple root $\zeta_\pm$ with $\Im\zeta_\pm>0$, at which every
$c\in\C^N$ is in the kernel; so $V_\pm=\{(c,\zeta_\pm c)\}$, whose $\gamma_1$-block is invertible. That is
\textup{[M]} at $p=2$, and $\cB_{\mathrm{high}}$ is the Dirichlet Lagrangian. Admissibility of a pointwise $L$ is
$L\pitchfork\{\gamma_0=0\}$, so the inadmissible ones form the wall
$\Sigma_{\{\gamma_0=0\}}=\{L\in\Lag(\cA):L\cap\{\gamma_0=0\}\ne0\}$, a codimension-one stratified
hypersurface whose deepest stratum is the Dirichlet Lagrangian itself. Such an $L$ has
$L^\downarrow=\pi_0(L)\oplus(L\cap F_1)$, which meets $V_\pm$ exactly when the two summands share a
vector, $\zeta_\pm$ being scalar. Isotropy makes them $\mathsf S$-orthogonal, so for definite $\mathsf S$ the
termination is elliptic (Remark~\ref{rem:SL}).
\end{proof}

\noindent So the path closes in two circumstances (an even principal symbol, and even $p$ under \textup{[S]} and \textup{[M]})
while $\Lm^\Phi(+)\ne\Lm^\Phi(-)$ {forces} dependence on the
termination. How much is the subject of the rest of this section.

\begin{lemma}[locally constant rotation]\label{lem:Rconst}
Let $s\mapsto H_s$, $s$ in a connected manifold $\mathscr S$, be a continuous family satisfying \textup{[H1]--[H4]} with a common gap containing $E_0$, whose Calder\'on paths $U_+(\xi,E_0;s)$ converge to their endpoints $U_+(\pm;s)$ locally uniformly in $s$. Then, for continuous lifts $\phi_\pm(s)$ of $\arg\det U_+(\pm;s)$,
\[
  \cR(s)-\tfrac1{2\pi}\bigl(\phi_+(s)-\phi_-(s)\bigr)\ \in\ \Z
\]
is constant on $\mathscr S$. In particular $\cR$ is constant when the endpoints do not depend on $s$.
\end{lemma}
\noindent The uniform convergence holds at $p=1$, and under \textup{[M]} for the transported data, by Lemma~\ref{lem:endpoints} and Corollary~\ref{cor:graded}\textup{(iii)}.
\begin{proof}
When $\cA$ depends on $s$, the charts are taken continuous in $s$. Use orthonormal frames $Y_\pm(s)$ of the spectral subspaces of $\cA_s(0)$ (continuous in $s$ by \cite[Ch.~I \S4.6 and Thm.~6.32]{kato1966perturbation}) and $U=A_+^{1/2}BA_-^{-1/2}$ with $A_\pm=\pm Y_\pm^*\cA_sY_\pm$. For fixed $s$ this is $O_+^*U_LO_-$ with $O_\pm$ unitary and independent of $L$, which changes neither $\cR$ nor relative eigenphases; $\Phi_s$ is continuous in $s$ by Theorem~\ref{thm:graded}. The quantity is an integer, as noted after Definition~\ref{def:rot}. Near $s_0$, on a contractible neighborhood $N$, the
map $\det U_+$ on the simply connected $[-\infty,+\infty]\times N$ has a continuous lift $\tilde\phi$ of its argument, so
$2\pi\cR(s)=\tilde\phi(+\infty,s)-\tilde\phi(-\infty,s)$ is continuous on $N$, and so are
$\phi_\pm-\tilde\phi(\pm\infty,\cdot)\in2\pi\Z$. The quantity is therefore locally constant, hence constant.
\end{proof}

\begin{corollary}[the rotation is a bulk quantity up to endpoint data]\label{cor:bulkR}
Let $H^\circ$ be as in Theorem~\ref{thm:collar}, and let $H_s=H^\circ+s(H-H^\circ)$, $s\in[0,1]$, satisfy \textup{[H2]} and the hypotheses of Lemma~\ref{lem:Rconst}. Then $2\pi(\cR[H]-\cR[H^\circ])=\Delta\phi_+-\Delta\phi_-$, the changes along $s\in[0,1]$ of the lifts of $\arg\det U_+(\pm;s)$.
\end{corollary}
\begin{proof}
Lemma~\ref{lem:Rconst} with the parameter $s\in[0,1]$.
\end{proof}
\noindent \textup{[H2]} holds along the path under \eqref{eq:flatsymbol}. The correction involves finitely many boundary jets of the coefficients, but no solution across the collar. It vanishes, so that $\cR$ depends on the bulk symbol alone, at $p=1$ under \eqref{eq:flatsymbol} and when the Calder\'on paths close (even $p$ under \textup{[S]}, \textup{[M]}). Similarly, by Lemma~\ref{lem:Rconst} with the parameter $s=E_0$, for which the chart is fixed, $\cR$ does not depend on $E_0$ when the endpoints do not depend on $E$ and the convergence is uniform in $E$. Lemma~\ref{lem:Rconst} also justifies the arithmetic of section~\ref{sec:examples}. There, one
numerical evaluation, together with the endpoint data, determines $\cR$ on each connected component of
the parameter region. The non-integrality of $\cR$ lives entirely in the endpoint data.

\begin{theorem}[the range of $\mathcal I$]\label{thm:range}
Let $U_1:=U_+(-)$, $U_2:=U_+(+)$, let $W'$ be the restriction of $U_1^{-1}U_2$ to the orthogonal complement of its eigenvalue $1$, and let $j(L)$ be the chamber label \eqref{eq:chamberlabel}. Then for every admissible $L$
\[
  \mathcal I(L)=\mathcal I_0-j(L),\qquad
  \mathcal I_0=\cR+\eta(W')+\tfrac12(\kp-\kp_\cap).
\]
If each chamber contains a termination with escaping edge spectrum, $\mathcal I$ takes exactly the $\kp-\kp_\cap+1$ consecutive values $\mathcal I_0,\dots,\mathcal I_0-\kp+\kp_\cap$, one on each chamber, and
\begin{equation}\label{eq:horm}
  \bigl|\mathcal I(L')-\mathcal I(L)\bigr|\le\kp-\kp_\cap\le\kp .
\end{equation}
\end{theorem}
\noindent Here $\kp_\cap$ is the defect \eqref{eq:defect}. The chambers are nonempty and open (Proposition~\ref{prop:admtopology}), and each contains a termination with escaping edge spectrum under the uniformity of Proposition~\ref{prop:exact}\textup{(i)}, in particular at $p=1$ and under \textup{[M]} for the transported data.

\begin{proof}
By \eqref{eq:APS}, $\mathcal I(L)=\cR+F(U_L)$ with
$F(U):=\eta(U^{-1}U_2)-\eta(U^{-1}U_1)$, continuous on $\cU$ with values in $\Z-\cR$, hence locally
constant there. We show that $F+j$ is constant. The points where
$\dim\ker(U-U_i)\ge2$ form a set of codimension $\ge4$ (Appendix~\ref{app:import}), and $\Sigma_{U_1}\cap\Sigma_{U_2}$ has
codimension $\ge2$ unless $U_1=U_2$ (Appendix~\ref{app:import}); removing a closed set of codimension $\ge2$ leaves $\U(\kp)$
connected, so it suffices to show that $F+j$ does not jump across a corank-one point $U_\star$ of
one wall lying off the other. Cross it along $U(s)=e^{is}U_\star$. All eigenphases of $U^{-1}U_2$ and
of $U^{-1}U_1$ then decrease at unit rate, so the continuous parts of the two $\eta$'s cancel in $F$,
while all eigenvalues of $\mathsf A_0=i(\mathbf 1+U_1^{-1}U)(\mathbf 1-U_1^{-1}U)^{-1}$, which acts on
eigenphases by $e^{i\theta}\mapsto-\cot(\theta/2)$, increase.

At $U_\star\in\Sigma_{U_2}$ the eigenphase of $U^{-1}U_2$ at $0$ passes to $2\pi^-$, so $F$ drops by
$1$; and $\mathsf A_{22}$, increasing, has exactly one eigenvalue passing upward through $0$ by
\eqref{eq:cayleyimage}, so $j$ rises by $1$.

At $U_\star\in\Sigma_{U_1}$ the eigenphase of $U^{-1}U_1$ at $0$ passes to $2\pi^-$, so $F$ rises by
$1$. The eigenphase $\theta_1$ of $U_1^{-1}U$ increases through $0$, so
$\mathsf A_0=t\,P+R_0$ with $P$, $R_0$ smooth, $P$ of rank one, and $t=-\cot(\theta_1/2)$ tending to $+\infty$ as $\theta_1\to0^-$ and coming back from $-\infty$ as $\theta_1\to0^+$. Its compression to $\mathsf K^\perp$ is nonzero, since $\Ran P(0)\subset\mathsf K$ would put
$U_\star$ on $\Sigma_{U_2}$; so exactly one eigenvalue of $\mathsf A_{22}$ does the same, passing through $\infty$ rather than through $0$, the others staying away from $0$ because by \eqref{eq:caydiff} $\det\mathsf A_{22}$ is a
nonzero multiple of $\det(U-U_2)\det(\mathsf A_0+i)$, of exact order $|t|$. So $j$ drops by $1$.

Hence $F+j$ is constant. Its value is read off at $U=e^{2i/t}U_1$, $t\to\infty$: then
$\mathsf A_0=-\cot(1/t)\mathbf 1$, so $\mathsf A_{22}<0$ for $t$ large and $j=0$, while
$\eta(U^{-1}U_1)\to-\kp/2$ and $\eta(U^{-1}U_2)\to\eta(W')-\kp_\cap/2$, the $\kp_\cap$ eigenphases on
$\mathsf K$ tending to $2\pi^-$. This gives $\mathcal I_0$. The chambers $j=0,\dots,\kp-\kp_\cap$ are nonempty
and open by Proposition~\ref{prop:admtopology}\textup{(ii)}, which gives the rest.
\end{proof}

\noindent At $\kp_\cap=\kp$ the theorem is Corollary~\ref{cor:loop}. In \textup{(E1)}, \textup{(E2)}, \textup{(E3)} and
\textup{(E8)}, numerically for the last, $W'=-\mathbf 1$, so $\eta(W')=0$ and the values fill $[\cR-\tfrac{\kp-\kp_\cap}2,\cR+\tfrac{\kp-\kp_\cap}2]$,
symmetric about $\cR$. The bound \eqref{eq:horm} is of H\"ormander-index type \cite{Hormander71}. Momentum dependence enlarges the range, and by a computable amount.

\begin{remark}\label{rem:polyrange}
\emph{(a) Every pair $(L_\infty,\cR_L)$ with $|\cR_L|\le\kp$ occurs for terminations of order $d=1$ under \textup{[S]}} (of order $\le p$ in general, $B\Phi^{-1}$ replacing $B$). Given $Y\in\U(\kp)$,
$j\in\{0,\dots,\kp\}$ and $s=\pm1$, let
\[
  \mathsf G^\pm_j(\xi):=\mathrm{diag}(1\pm is\xi,\dots,1\pm is\xi,1,\dots,1)
\]
with $j$ entries $1\pm is\xi$, so that $\mathsf G^-_j(\xi)$ is the complex conjugate of $\mathsf G^+_j(\xi)$ for real $\xi$, and take, in the coordinates of Proposition~\ref{prop:U}, the affine boundary
condition $B(\xi)=\bigl(\mathsf G^+_j(\xi)D_+^{1/2},\,-\mathsf G^-_j(\xi)\,YD_-^{1/2}\bigr)$. As in
Lemma~\ref{lem:plucker}\textup{(i)}, $L(\xi)=\ker B(\xi)$ is the graph of $-B_+^{-1}B_-$, so by \eqref{eq:UL}
\[
  U_{L(\xi)}=\mathsf G^+_j(\xi)^{-1}\mathsf G^-_j(\xi)\,Y
  =\mathrm{diag}\Bigl(\tfrac{1-is\xi}{1+is\xi},\dots,1,\dots\Bigr)Y ,
\]
unitary for every real $\xi$; hence $L(\xi)$ is Lagrangian and a termination. Its endpoint is
$U_{L_\infty}=\mathsf E_jY$ with $\mathsf E_j=\mathrm{diag}(-1,\dots,-1,1,\dots,1)$, and $\cR_L=-sj$ by
Lemma~\ref{lem:plucker}\textup{(iii)}. As $Y$ ranges over $\U(\kp)$, $L_\infty$ is arbitrary, and admissible
outside a closed null set, whatever the value of $\cR_L\in\{-\kp,\dots,\kp\}$; in particular the bound of
Lemma~\ref{lem:plucker} is attained.

\emph{(b) The range.} By Theorems~\ref{thm:main} and \ref{thm:range}, an admissible termination differential
of order $\le d$ has $\mathcal I(L)=\mathcal I(\Phi L_\infty)-\cR_L$, with $\mathcal I(\Phi L_\infty)$ in a
window of width $\kp-\kp_\cap\le\kp$ and, under \textup{[S]}, $|\cR_L|\le\kp d$. Hence
$|\mathcal I(L)-\mathcal I(L')|\le\kp+2\kp d$, and by (a) the value $3\kp$ is attained at $d=1$ when
$\kp_\cap=0$. The bound grows linearly in $d$; over terminations of unbounded order, or pseudodifferential
ones (Remark~\ref{rem:PDO}), $\mathcal I$ is unbounded. For interfaces the same count at
$\widetilde\kp=2\kp$ (Theorem~\ref{thm:double}) leaves the current undetermined within a window of width
$2\kp$ over pointwise transmission conditions, and $6\kp$ once one tangential derivative is allowed.
\end{remark}

\section{The bulk--edge correspondence at an interface}\label{sec:interface}

\label{ssec:double}%
Reflection turns an interface between two media into a half-space problem for a doubled operator
(Theorem~\ref{thm:double}), to which sections~\ref{sec:green}--\ref{sec:range} apply. The transparent
transmission condition is the unit of the doubled unitary group (Proposition~\ref{prop:transparent}), elliptic
under a transversality hypothesis but admissible only when $p=1$ (Proposition~\ref{prop:tradm}). The
rotation splits over the two media, and at the transparent transmission condition the invariant is compared with the
bulk-difference invariant of the Euclidean plane (Theorem~\ref{thm:R2}, Corollary~\ref{cor:etaglue}).

See also \cite{KSV14} for the quantization of interface currents by the difference of the bulk invariants in tight-binding models; and \cite{GontierTauber25} for junctions of two one-dimensional materials sharing a boundary symplectic form that are treated
through the Lagrangian planes of their half-line $L^2$ solutions.


\paragraph{The doubling of half spaces and transmission conditions.}

Let $\Om_\pm:=\R_x\times\R_{\pm,y}$ with common boundary $\Gam$, let $H_\pm$ act on $\C^N$ over
$\Om_\pm$, each satisfying \textup{[H1]--[H4]} with the same order $p$, and let $\Delta$ be a common gap
of the two bulk operators. Write $\gamma^\pm u$ for the Cauchy data of $u|_{\Om_\pm}$ at $y=0^\pm$; a
{pointwise transmission condition} is a linear condition on the pair $(\gamma^+u,\gamma^-u)$.

For the medium on $\Om_-$, the relevant Cauchy data are those of solutions decaying as $y\to-\infty$, the growing ones in the orientation of $\Om_+$. For $E\in\Delta$ let
\[
  \Lm^-(x,\xi,E):=\mathsf F_{\xi,E}(0,y_0)\,\Ran P^u(\xi,E)
\]
be the \emph{growing} Calder\'on Lagrangian, defined as in \eqref{eq:LamRiesz} but with the Riesz projection $P^u$ over the open \emph{right} half plane. It consists of the Cauchy data of the solutions of the fiber equation on $\R_+$ whose data at $y_0$ lie in $\Ran P^u$, which grow exponentially. For the medium $H_-$, given on $\Om_-$, we write $\Lm^-_{H_-}(x,\xi,E)$ for the Cauchy data at $y=0^-$ of the $L^2(\R_-)$ solutions of the fiber equation, that is $\mathsf F_{\xi,E}(0,-y_0)\Ran P^u(\xi,E)$ with $P^u$ formed for the constant coefficients of $H_-$ at $y\le-y_0$; for a $y$-independent operator the two definitions agree.

\begin{lemma}[the Calder\'on Lagrangian of the other half plane]\label{lem:gro}
$\Lm^-(x,\xi,E)$ is $q$-Lagrangian and transverse to $\Lm^+(x,\xi,E)$ for every $\xi\in\R$ and $E\in\Delta$, and its limits $\Lm^-(x,\pm)$ as $\xi\to\pm\infty$ exist, with the rate \eqref{eq:endrate}.
\end{lemma}
\noindent The limits are the graded limits $(V^u_\pm)^{\uparrow}$ of the {unstable} frozen data, hence principal-symbol data under the analogue of \textup{[P]} for $V^u_\pm$, which we assume whenever this is used. The transversality holds at finite $\xi$ only: by \eqref{eq:highinlimit} the two endpoint families meet whenever $p\ge2$ and the maximal-weight condition holds for both. For an operator $H$ given on $\Om_-$, $\Lm^-_H$ is the Calder\'on Lagrangian of its half-plane problem. We write $\cR[H;\Om_-]:=\cR[\Lm^-_H]$ for its rotation, read in the chart \eqref{eq:UL}, and $\cR[H]=\cR[H;\Om_+]$.
\begin{proof}
The limits follow from the proof of Lemma~\ref{lem:endpoints}.
Let $\cA(y)$ be the Green form of Proposition~\ref{prop:green} for the fiber operator at height $y$, so that $\cA(0)=\cA$ and $q(w)=w^*\cA w$. By \eqref{eq:green} on $[y,y']$, the pairing $\ip{\Psi_1(y)}{\cA(y)\Psi_2(y)}$ of two solutions does not depend on $y$; so $\mathsf F_{\xi,E}(0,y_0)$ carries the form $w^*\cA(y_0)w$ to $q$. Beyond $y_0$ the fiber equation has constant coefficients, and $\Ran P^u$ and $\Ran P^s$ are the spaces $\Lm^-$ and $\Lm^+$ of Lemma~\ref{lem:hermODE} for $P(D_y)=\Hh^\natural(\xi)-E$: the Cauchy data of the solutions on $\R$ decaying as $y\to-\infty$ and as $y\to+\infty$. By that lemma they are Lagrangian for $\cA(y_0)$ and complementary. Transported by the invertible $\mathsf F_{\xi,E}(0,y_0)$, which carries $\cA(y_0)$ to $q$, $\Lm^-$ is $q$-Lagrangian and transverse to $\Lm^+=\mathsf F_{\xi,E}(0,y_0)\Ran P^s$. For $\Lm^-_{H_-}$ the same argument, with Green's formula on $\R_-$, gives Lagrangian and the endpoint limits.
\end{proof}

\begin{definition}[reflection and doubling]\label{def:double}
Let $\mathsf R_{\mathrm{fl}}u(x,y):=u(x,-y)$ and $\check H_-:=\mathsf R_{\mathrm{fl}}H_-\mathsf R_{\mathrm{fl}}^{-1}$, an
operator on $\Om_+$ whose coefficient of $D_y^k$ is that of $H_-$ reflected and multiplied by
$(-1)^k$. On $L^2(\Om_+;\C^{2N})$ put
\[
  \widetilde H:=H_+\oplus\check H_-,\qquad
  \widetilde u:=\bigl(u|_{\Om_+},\ \mathsf R_{\mathrm{fl}}u|_{\Om_-}\bigr),
\]
so that $u\mapsto\widetilde u$ is unitary $L^2(\R^2;\C^N)\to L^2(\Om_+;\C^{2N})$.
\end{definition}

\begin{theorem}[the interface problem]\label{thm:double}
$\widetilde H$ satisfies \textup{[H1]--[H4]} on $\Om_+$, $\Delta$ being a common gap of $H_\pm$, with
\[
  \widetilde N=2N,\quad\widetilde\kp=2\kp,\quad\widetilde\cB=\cB\oplus\cB,\quad
  \widetilde\cA=\cA_+\oplus(-\cA_-),\quad
  \widetilde\Lm^+(\xi,E)=\Lm^+_{H_+}(\xi,E)\oplus\Lm^-_{H_-}(\xi,E)
\]
in the coordinates $\widetilde\gamma\widetilde u=(\gamma^+u,\gamma^-u)$.
The map $u\mapsto\widetilde u$ intertwines the interface problem with a boundary problem for $\widetilde H$, pointwise self-adjoint transmission conditions corresponding to $\widetilde\cA$-Lagrangians,
\begin{equation}\label{eq:doublemoduli}
  \bigl\{\text{pointwise self-adjoint transmission conditions}\bigr\}\cong\Lag(\widetilde\cA)\cong\U(2\kp),
\end{equation}
and every result of sections~\ref{sec:green}--\ref{sec:range} holds for transmission conditions, with $\kp$ replaced by $2\kp$.
\end{theorem}
\noindent In particular, Theorem~\ref{thm:range} gives $2\kp-\widetilde\kp_\cap+1$ values, with $\widetilde\kp_\cap:=\dim\bigl(\widetilde\Lm^+(+)\cap\widetilde\Lm^+(-)\bigr)$; and for a transmission condition differential of order $d$ along $\Gam$, Lemma~\ref{lem:plucker} and Theorem~\ref{thm:main} give $L(+)=L(-)$ and $|\cR_L|\le2\kp d$ under \textup{[S]}. Both media must have the same order $p$. Condition \eqref{eq:Ehyp} for $\widetilde H$ is equivalent to \eqref{eq:Ehyp} for $H_+$ and $\check H_-$ (proof below). Without \textup{[S]} at $p\ge2$, pointwise transmission conditions are understood as transported ones (Definition~\ref{def:transported}); $L_{\mathrm{tr}}$ is unaffected, the doubled transport fixing the diagonal when $\cA_+(\xi)=\cA_-(\xi)$.
\begin{proof}
\textup{[H1]--[H2]} transfer unchanged since $\mathsf R_{\mathrm{fl}}$ is unitary and the principal symbol and the coefficient of $D_y^p$ of
$\widetilde H$ are block diagonal with invertible blocks. For \textup{[H3]--[H4]} take
$y_0:=\max(y_0^+,y_0^-)$ and $\widetilde H^{\R^2}=H^{\R^2}_+\oplus\check H^{\R^2}_-$. Since
$\spec\check H^{\R^2}_-=\spec H^{\R^2}_-$, \textup{[H3]} for $\widetilde H$ is exactly the requirement
that $\Delta$ be a {common} gap.

For \textup{[M]}, and by the same factorization of split Pl\"ucker coordinates, for
\textup{[P]}, the weight multiset of $\widetilde\cB=\cB\oplus\cB$ is two copies of that of $\cB$,
so $\widetilde w^*=2w^*$ by \eqref{eq:wstar}; and a Pl\"ucker coordinate of a split plane vanishes
unless the index set splits evenly, in which case it factors,
$p_{I_1\cup I_2}\bigl(V^{(1)}\oplus V^{(2)}\bigr)=\varepsilon\,p_{I_1}(V^{(1)})\,p_{I_2}(V^{(2)})$
with $|\varepsilon|=1$ and $|I_1|=|I_2|=\kp$. A nonvanishing coordinate of weight $2w^*$ therefore has $w_{I_1}=w_{I_2}=w^*$ and both factors nonzero, and
conversely. So \eqref{eq:Ehyp} for $\widetilde H$ is {equivalent} to \eqref{eq:Ehyp} for $H_+$
and for $\check H_-$.

For $\widetilde\cA$: Green's formula on $\Om_+$ gives
$\ip{v}{H_+u}-\ip{H_+v}{u}=i\ip{\gamma^+v}{\cA_+\gamma^+u}$, whereas on $\Om_-$, where the boundary
sits at the upper end of the interval, the same computation produces the opposite sign. Adding, the
Green form of the interface problem is $\cA_+\oplus(-\cA_-)$ in the identification
$\widetilde\gamma\widetilde u=(\gamma^+u,\gamma^-u)$, up to the orthogonal diagonal sign conjugation
$\mathrm{diag}((-1)^j)_{j<p}$ coming from $\gamma_j(\mathsf R_{\mathrm{fl}}u)=(-1)^j\gamma_j^-u$, which does
not affect signatures. Since $\sig\cA_\pm=0$ by Theorem~\ref{thm:signature}, also
$\sig\widetilde\cA=0$; equivalently, apply that theorem directly to $\widetilde H$.
Proposition~\ref{prop:U} then gives \eqref{eq:doublemoduli}.

For the Calder\'on Lagrangian, $\widetilde H$ is block diagonal, so an $L^2(\R_+)$ solution of
$(\widetilde H-E)\widetilde\psi=0$ is a pair of $L^2(\R_+)$ solutions for $H_+$ and $\check H_-$, the
latter carried by $\mathsf R_{\mathrm{fl}}$ to an $L^2(\R_-)$ solution of $H_--E$; the dimensions add
to $2\kp=\widetilde\kp$. The remaining statements transfer since $\widetilde H$ satisfies \textup{[H1]--[H4]} and $\widetilde\cA$, the Green form of a differential operator, lies in the class \eqref{eq:gradedclass}.
\end{proof}

\begin{proposition}[transparent transmission condition is group unit]\label{prop:transparent}
Suppose $\cA_+=\cA_-$. The \emph{transparent} transmission condition $\gamma^+u=\gamma^-u$ corresponds to the
diagonal $L_{\mathrm{tr}}=\{(w,w)\}$, which is $\widetilde\cA$-Lagrangian, and in the coordinates of
\eqref{eq:doublemoduli} it is $U_{L_{\mathrm{tr}}}=\mathbf 1$. The \emph{decoupled} transmission conditions are the
split Lagrangians $L_+\oplus L_-$, the image of an embedding
$\U(\kp)\times\U(\kp)\hookrightarrow\U(2\kp)$ on which
$\mathcal I=\mathcal I_+(L_+)+\mathcal I_-(L_-)$; no decoupled transmission condition is transparent.
\end{proposition}
\begin{proof}
$\widetilde q(w,w)=w^*\cA w-w^*\cA w=0$, so $L_{\mathrm{tr}}$ is isotropic of dimension
$\dim\cB=2\kp=\widetilde\kp$, hence Lagrangian.

Write $v_1,\dots,v_\kp$ and $u_1,\dots,u_\kp$ for the $\cA$-positive and $\cA$-negative eigenvectors
and $v^{(i)},u^{(i)}$ for their copies in the $i$-th summand. The $\widetilde\cA$-positive block is
spanned by $v^{(1)},u^{(2)}$ and the negative one by $v^{(2)},u^{(1)}$; ordering each block
with its $v$-type vectors first, as allowed by the proof of Proposition~\ref{prop:U}, gives $\widetilde D_+=\widetilde D_-=\mathrm{diag}(D_+,D_-)$, in which
coordinates $\{(w,w)\}$ is the graph of $\widetilde B=\mathbf 1$ and \eqref{eq:UL} returns
$U_{L_{\mathrm{tr}}}=\mathbf 1$.
For the second statement, $\widetilde H_L=H_{+,L_+}\oplus H_{-,L_-}$ on a split Lagrangian and $\sigma_I$ is additive
over direct sums. No decoupled transmission condition is transparent, since if
$L_{\mathrm{tr}}=L_+\oplus L_-$ then $(w,0)\in L_{\mathrm{tr}}$ for every $w\in L_+$, forcing $w=0$
and so $L_+=0$, whereas $\dim L_+=\kp\ge1$.
\end{proof}

\noindent So the transparent condition is a distinguished point of $\Lag(\widetilde\cA)\cong\U(2\kp)$, the identity in the chart above, whereas boundary problems have no such base point. At $L_{\mathrm{tr}}$ the interface problem is the operator on $\R^2$ obtained by gluing $H_+$ and $H_-$, with no condition imposed: $\mathcal I(L_{\mathrm{tr}})$ is an invariant of that operator.

\medskip
\begin{proposition}[ellipticity and admissibility of the transparent transmission condition]\label{prop:tradm}
Assume \textup{[H1]--[H4]} for $H_\pm$ with a common gap, \eqref{eq:Ehyp} for each side and its
analogue for the growing data, and $\cA_+=\cA_-$. Let $V_\pm[\sigma^+_p]$ and $V^u_\pm[\sigma^-_p]$ be the decaying and the growing frozen data of the two principal symbols at $(x,0)$ and $\xi=\pm1$, and $U^{H_+}_+(\pm)$, $U^{H_-}_-(\pm)$ the images of $\Lm^+_{H_+}(\pm)$ and $\Lm^-_{H_-}(\pm)$ under \eqref{eq:UL}.
\begin{enumerate}[label=\textup{(\roman*)},nosep]
\item $(\widetilde H,L_{\mathrm{tr}})$ is uniformly Shapiro--Lopatinski elliptic if and only if
\begin{equation}\label{eq:C2frozen}
  V_\pm[\sigma^+_p]\cap V^u_\pm[\sigma^-_p]=\{0\}\quad\text{uniformly in }x .
\end{equation}
\item $L_{\mathrm{tr}}$ is admissible if and only if
\begin{equation}\label{eq:C2det}
  \det\Bigl(U_+^{H_+}(\pm)-U_-^{H_-}(\pm)\Bigr)\ne0 ,\qquad\text{that is}\qquad
  \Lm^+_{H_+}(\pm)\pitchfork\Lm^-_{H_-}(\pm)\ \text{in }\cB .
\end{equation}
\end{enumerate}
\end{proposition}
\noindent Both are conditions on the two principal symbols at $y=0$ alone, and \eqref{eq:C2frozen} holds automatically when $\sigma^+_p=\sigma^-_p$ at $y=0$. At $p=1$ the endpoint families are the frozen data and \eqref{eq:C2det} is \eqref{eq:C2frozen}. For $p\ge2$, \eqref{eq:C2det} fails: the two endpoint families are maximal-weight graded limits, which meet by \eqref{eq:highinlimit}.
\begin{proof}
\textup{(i)} $\widetilde\Th_t=\Th_t\oplus\Th_t$ fixes the diagonal, and ellipticity as stated in
\eqref{eq:SLdef} is transversality of $L_{\mathrm{tr}}^{\downarrow}=L_{\mathrm{tr}}$ to
$\widetilde V_\pm=V_\pm[\sigma^+_p]\oplus V^u_\pm[\sigma^-_p]$. For the second summand, this is because
$\check H_-$ has principal symbol $\sigma^-_p(\xi,-\zeta)$, whose decaying data are the growing data of
$\sigma^-_p$ up to the sign conjugation $\mathrm{diag}((-1)^j)$ that the coordinates
$\widetilde\gamma\widetilde u=(\gamma^+u,\gamma^-u)$ absorb. A diagonal $(w,w)$ lies in that sum iff
$w$ lies in both summands, which is \eqref{eq:C2frozen}. When the two principal symbols agree the
summands are the stable and the unstable subspace of one companion matrix, complementary by
Lemma~\ref{lem:hermODE}\textup{(b)}.

\textup{(ii)} $\widetilde\Lm^+(\pm)=\Lm^+_{H_+}(\pm)\oplus\Lm^-_{H_-}(\pm)$ by
Theorem~\ref{thm:double}, so $(w,w)\in L_{\mathrm{tr}}\cap\widetilde\Lm^+(\pm)$ iff
$w\in\Lm^+_{H_+}(\pm)\cap\Lm^-_{H_-}(\pm)$; \eqref{eq:transdet} turns this into the vanishing of the
endpoint determinant, and the endpoints are principal-symbol data by Lemma~\ref{lem:endpoints} and
Lemma~\ref{lem:gro}, hence maximal-weight graded limits, whence \eqref{eq:highinlimit}.
Frozen data and graded limits are different. The summands of $\widetilde V_\pm$ are
transverse by \textup{(i)}, and $\Th_t$ destroys that.
\end{proof}
\noindent The other transmission conditions are covered by Proposition~\ref{prop:admtopology} at $\widetilde\kp=2\kp$: the admissible set is open, dense and non-empty, with $2\kp-\widetilde\kp_\cap+1$ chambers, the doubled defect $\widetilde\kp_\cap=\kp_\cap[H_+]+\kp_\cap[H_-;\Om_-]$ being additive because both endpoints split along $\cB\oplus\cB$; at even $p$ under \textup{[S]} and \textup{[M]} there is a single chamber.

\medskip\noindent
\textbf{[C1]--[C3] Interface hypotheses.} For the two media $H_\pm$ on $\Om_\pm$, the doubled operator $\widetilde H$ of Theorem~\ref{thm:double} and the transparent transmission condition $L_{\mathrm{tr}}$ of Proposition~\ref{prop:transparent}, we introduce:
\begin{enumerate}[label=\textup{[C\arabic*]},leftmargin=*,nosep]
\item \textup{[H1]--[H4]} for each of $H_\pm$, with the same order $p$ and a \emph{common} gap $\Delta$, and \eqref{eq:Ehyp} for the doubled operator, which holds if and only if it holds for $H_+$ and for the reflection of $H_-$.
{Admissibility is not assumed} and does not hold for the transparent transmission condition once $p\ge2$ by Proposition~\ref{prop:tradm}\textup{(ii)}.
\item \textup{(a)} $\cA_+=\cA_-$, so that the diagonal $L_{\mathrm{tr}}$ is
$\widetilde\cA$-Lagrangian and the sharp interface operator is symmetric; and \textup{(b)} the
\emph{frozen} decaying data of $\sigma_p^+$ and growing data of $\sigma_p^-$ are transverse at
$\xi=\pm1$, uniformly in $x$, which by Proposition~\ref{prop:tradm}\textup{(i)} is uniform
Shapiro--Lopatinski ellipticity of $(\widetilde H,L_{\mathrm{tr}})$.
\item with $\sigma^\pm_{p,\infty}$ the principal symbols of $H_\pm$ at $y=\pm\infty$, every convex combination of them is
elliptic,
\[
  \inf_{t\in[0,1]}\ \min_{|(\xi,\zeta)|=1}s_{\min}\bigl((1-t)\sigma_{p,\infty}^+(\xi,\zeta)+t\sigma_{p,\infty}^-(\xi,\zeta)\bigr)>0 ,
\]
and the principal symbol of $H_\pm$ at every $y$ is such a convex combination, as when it does not depend on $y$;
so that the \emph{glued} operator on $\R^2$ is uniformly elliptic along the linear interpolation and
along every interpolation between that one and a sharp jump.
\end{enumerate}

\medskip

\noindent We assume \textup{[C1]--[C3]} for the rest of the section. Proposition~\ref{prop:sharpid} and Theorem~\ref{thm:R2} are on the elliptic route and allow $x$-dependent media, the condition \eqref{eq:Ehyp} of \textup{[C1]} not being used there; \eqref{eq:APSint}, Corollary~\ref{cor:etaglue} and Remarks~\ref{rem:jumpwinding}--\ref{rem:KS} are on the spectral route, for $x$-independent media. \textup{[C2]}\textup{(a)} makes
$L_{\mathrm{tr}}$ Lagrangian, so that the transparent realization is symmetric; \textup{[C2]}\textup{(b)} is its
uniform Shapiro--Lopatinski ellipticity (Proposition~\ref{prop:tradm}\textup{(i)}), which makes it
self-adjoint and gives existence, quantization and stability of $\mathcal I(L_{\mathrm{tr}})$
(section~\ref{sec:obs}), and escape for $x$-independent media, by Remark~\ref{rem:escape} applied to $(\widetilde H,L_{\mathrm{tr}})$.
\textup{[C3]} is used for smooth walls and for $\BDI$ (Proposition~\ref{prop:sharpid}, Definition~\ref{def:BDI}).
\textup{[C2]}\textup{(b)} holds automatically when the two principal symbols agree at $y=0$, and fails, for
instance, at first order when the tangential principal symbol reverses sign
(section~\ref{sec:examples}). The correspondence is Proposition~\ref{prop:wall} applied to
$\widetilde H$, admissibility being the convenient sufficient condition where it is available.
At $L_{\mathrm{tr}}$ it is not, once $p\ge2$, and the two-medium form \eqref{eq:APSint} is then read with
the regularized asymmetries $\eta^{\mathrm{reg}}(\pm)$ of Definition~\ref{def:etareg}, which differ from
the $\eta$'s by the integer of \eqref{eq:IReq}.

\noindent The velocity flip at the end of section~\ref{sec:examples} is an example where \textup{[C2]}\textup{(b)}  fails. Neither ellipticity nor escape is then available and $\mathcal I(L_{\mathrm{tr}})$ is not defined by the present method; see also \cite{FrazierBal25}.

\paragraph{Rotation splits and bulk-difference invariant.}\label{ssec:R2}

By Theorem~\ref{thm:double}, $\widetilde\cA=\cA_+\oplus(-\cA_-)$ and
$\widetilde\Lm^+=\Lm^+_{H_+}\oplus\Lm^-_{H_-}$: a split Lagrangian in a split form. Applying
Lemma~\ref{lem:chartalg}\textup{(ii)} of Appendix~\ref{app:import} and then \textup{(i)} to the second
block,
\begin{equation}\label{eq:Rsplit}
  \widetilde\cR=\cR[H_+]-\cR[H_-;\Om_-],
\end{equation}
with $\cR[H_+]$ the rotation of the Calder\'on path of $H_+$ on $\Om_+$ and $\cR[H_-;\Om_-]$ that of $H_-$ on
its own half plane $\Om_-$, read in the chart of $\cA_-$: each term depends on one medium alone. For $x$-independent media \eqref{eq:APS} therefore reads
\begin{equation}\label{eq:APSint}
  \mathcal I=\Bigl(\cR[H_+]-\cR[H_-;\Om_-]\Bigr)
  +\eta\bigl(L,\widetilde\Lm^+(+)\bigr)-\eta\bigl(L,\widetilde\Lm^+(-)\bigr).
\end{equation}
The bulk-difference structure of the $\R^2$ invariant is automatic, and the boundary term is a
spectral asymmetry of the transmission condition against the two Calder\'on endpoints.

\begin{definition}[the bulk-difference invariant]\label{def:BDI}
Let $h(y;\xi,\zeta)$ be the Weyl symbol of an operator on $\R^2$ equal to $H_\pm$ for $\pm y\ge y_0$, with bounded lower-order terms and, at each $y$, principal symbol a convex combination of $\sigma^\pm_{p,\infty}$. Let $S^2_R\subset\R^3_{(y,\xi,\zeta)}$ be the sphere of radius $R$, oriented as the boundary of the ball. For $R$ large, set
\begin{equation}\label{eq:BDIdef}
  \BDI(H_+,H_-):=c_1\Bigl[\Pi_+(h-E_0)\big|_{S^2_R}\Bigr]
  =\frac{i}{2\pi}\int_{S^2_R}\tr\bigl(\Pi_+\,d\Pi_+\wedge d\Pi_+\bigr)\ \in\Z ,
\end{equation}
with $\Pi_+$ the positive spectral projection.
\end{definition}
\noindent This is well defined. For $R$ large, $h-E_0$ is invertible on $S^2_R$: where $|y|\ge y_0$ by the bulk gap \textup{[H3]}, and where $|y|<y_0$ because $|(\xi,\zeta)|\ge\sqrt{R^2-y_0^2}$ and the principal symbol, elliptic by \textup{[C3]}, dominates. So $\Pi_+(h-E_0)$ is a smooth vector bundle on $S^2_R$, of rank $\tfrac12(N+\sig\sigma_p)$, $\sig\sigma_p$ being constant by \textup{[C3]} and connectedness. $\BDI$ does not depend on $R$ or on the interpolation on $|y|<y_0$, both being homotopies through symbols invertible on the sphere. It depends only on the two bulk symbols.

\paragraph{Transparent transmission and Euclidean plane.}%

The comparison with $\BDI$ takes three steps. The first two in Proposition~\ref{prop:sharpid} state that the
transparent realization is the sharp interface operator, and the sharp interface has the invariant of a
smooth wall. The third is the Euclidean-plane correspondence of \cite{Bal2022,QuinnBal}.

\begin{proposition}\label{prop:sharpid}
Assume \textup{[C1]} and \textup{[C2]}, and let
$H_\sharp:=H_+\mathbf 1_{y>0}+H_-\mathbf 1_{y<0}$ on $\R^2$ with domain $H^p(\R^2)$. Then
$\widetilde H_{L_{\mathrm{tr}}}=\cW H_\sharp\cW^*$ for the doubling unitary $\cW:u\mapsto\widetilde u$
of Definition~\ref{def:double}, so that
\begin{equation}\label{eq:sharpid}
  \mathcal I(L_{\mathrm{tr}})=2\pi\,\sigma_I[H_\sharp].
\end{equation}
Assume in addition \textup{[C3]}, let $H_{\mathrm{sm}}$ have smooth coefficients on $\R^2$, equal
to $H_\pm$ for $\pm y\ge y_0$, with principal symbol a convex combination of $\sigma^\pm_{p,\infty}$ at every $y$,
and assume \eqref{eq:C2frozen} uniformly in $t\in[0,1]$ for the pair of principal symbols at $y=0^\pm$
along the interpolation of the proof. Then
\begin{equation}\label{eq:sharpsmooth}
  \sigma_I[H_\sharp]=\sigma_I[H_{\mathrm{sm}}] .
\end{equation}
\end{proposition}
\noindent In particular $i[H_\sharp,\chi]\varphi'(H_\sharp)$ is trace-class and $2\pi\sigma_I[H_\sharp]\in\Z$. The uniform transversality along the interpolation is automatic when $\sigma^+_p=\sigma^-_p$ at $y=0$, and at $p=1$ when $N=2$ (end of the proof); no further relation between the coefficients of $H_\sharp$ and $H_{\mathrm{sm}}$ on the collar is assumed.
\begin{proof}
The transparent condition $\gamma^+u=\gamma^-u$ holds exactly when $u\in H^p(\R^2)$, so $\cW$ maps $\cD(H_\sharp)=H^p(\R^2)$ onto $\cD(\widetilde H_{L_{\mathrm{tr}}})$ and intertwines the two operators. The operator $\widetilde H_{L_{\mathrm{tr}}}$ is symmetric, $L_{\mathrm{tr}}$ being $\widetilde\cA$-Lagrangian since $\cA_+=\cA_-$, and uniformly Shapiro--Lopatinski elliptic by \textup{[C2]}. Its rows $\gamma^+_j-\gamma^-_j$ are normal, so its adjoint is the realization of the adjoint condition \cite[Thms.~1.6.9, 1.6.11]{Grubb96}, namely $\{(v_1,v_2):\ip{v_1-v_2}{\cA w}=0\ \forall w\}$, which is the diagonal again as $\cA$ is invertible. Hence $\widetilde H_{L_{\mathrm{tr}}}$ is self-adjoint, and so is $H_\sharp$. As $\chi=\chi(x)$ commutes with the reflection in $y$, $\cW\chi=\chi\cW$ and $\varphi'(\widetilde H_{L_{\mathrm{tr}}})=\cW\varphi'(H_\sharp)\cW^*$, so the traces in \eqref{eq:sharpid} agree. Trace class and quantization follow from Theorems~\ref{thm:traceclass} and \ref{thm:quantisation} for $\widetilde H_{L_{\mathrm{tr}}}$, transported through $\cW$. So at a sharp interface the observable is finite and quantized, extending \cite{Bal2022,QuinnBal}.

For \eqref{eq:sharpsmooth}, interpolate the coefficients, $A_k^{(t)}:=(1-t)A_k^\sharp+tA_k^{\mathrm{sm}}$ for $t\in[0,1]$ and $0\le k\le p$, so that $H^{(0)}=H_\sharp$ and $H^{(1)}=H_{\mathrm{sm}}$. We check the hypotheses of Corollary~\ref{cor:stab}\textup{(i)} for $(\widetilde H^{(t)},L_{\mathrm{tr}})$.
(a) $A_k^\sharp=A_k^{\mathrm{sm}}$ for $|y|\ge y_0$, so the bulk operators, and with them \textup{[H3]}, do not depend on $t$.
(b) By \textup{[C3]} every $\sigma_p^{(t)}(y;\cdot)$ is a convex combination of $\sigma^\pm_{p,\infty}$, hence elliptic, and the coefficient bounds are uniform in $t$.
(c) $L_{\mathrm{tr}}$ stays Lagrangian, which Theorem~\ref{thm:double} requires since for $t<1$ the coefficients are discontinuous at $y=0$. The Green form is linear in the coefficients and their normal jets (Proposition~\ref{prop:green}), so $\cA^{(t)}_\pm=(1-t)\cA^\sharp_\pm+t\cA^{\mathrm{sm}}_\pm$, and $\cA^\sharp_+=\cA^\sharp_-$ by \textup{[C2]}\textup{(a)}, $\cA^{\mathrm{sm}}_+=\cA^{\mathrm{sm}}_-$ as $H_{\mathrm{sm}}$ has continuous coefficients.
(d) By Proposition~\ref{prop:tradm}\textup{(i)}, uniform Shapiro--Lopatinski ellipticity of $(\widetilde H^{(t)},L_{\mathrm{tr}})$ is \eqref{eq:C2frozen} along the path, which is assumed.
So $\mathcal I$ is constant in $t$; at $t=1$ the transmission condition is vacuous, and $\widetilde H^{(1)}_{L_{\mathrm{tr}}}=\cW H_{\mathrm{sm}}\cW^*$.

At $p=1$ and $N=2$, \eqref{eq:C2frozen} along the path follows from \textup{[C3]}. A common vector of $V_\pm[c_1]$ and $V^u_\pm[c_2]$, for $c_1,c_2$ the principal symbols at $y=0^\pm$ along the interpolation, would be a common eigenvector of $A^{-1}c_1$ and $A^{-1}c_2$ with eigenvalues in opposite half planes, so that some convex combination of $c_1$ and $c_2$ would have a real eigenvalue, that is, would not be elliptic. But $c_1$ and $c_2$, like the principal symbols on $\Gam$, are convex combinations of $\sigma^\pm_{p,\infty}$ by \textup{[C3]}, hence so are their convex combinations, which are therefore elliptic.
\end{proof}

\begin{theorem}[the transparent point is the bulk-difference invariant]\label{thm:R2}
Assume \textup{[C1]--[C3]}, and let $H_{\mathrm{sm}}$ be as in Proposition~\ref{prop:sharpid}, with \eqref{eq:C2frozen} holding uniformly along its interpolation. Then
\begin{equation}\label{eq:R2}
  \mathcal I(L_{\mathrm{tr}})=\BDI(H_+,H_-),
\end{equation}
and, equivalently by \eqref{eq:sharpid} and \eqref{eq:sharpsmooth}, at the sharp interface
\[
  2\pi\sigma_I[H_\sharp]=2\pi\sigma_I[H_{\mathrm{sm}}]=\BDI(H_+,H_-),
\]
both observables being finite and integral by Proposition~\ref{prop:sharpid}.
\end{theorem}
\begin{proof}
By \eqref{eq:sharpid} and \eqref{eq:sharpsmooth},
$\mathcal I(L_{\mathrm{tr}})=2\pi\sigma_I[H_\sharp]=2\pi\sigma_I[H_{\mathrm{sm}}]$.
It remains to identify $2\pi\sigma_I[H_{\mathrm{sm}}]$ with $\BDI$. The operator $H_{\mathrm{sm}}$ satisfies hypothesis (H1) of \cite{QuinnBal}: its Weyl symbol is elliptic there by \textup{[C3]}, and it equals the constant-coefficient bulk symbols $h_\pm$, gapped on $\Delta$ by \textup{[H3]}, for $\pm y\ge y_0$. So \cite[Cor.~2.4]{QuinnBal} writes $2\pi\sigma_I[H_{\mathrm{sm}}]$ as a difference of bulk integrals, which integration in the spectral parameter turns into the Chern number \eqref{eq:BDIdef} \cite{Bal2022,QuinnBal}.
No orientation convention enters beyond the definitions: the quoted identity concerns the same trace
$\sigma_I[H_{\mathrm{sm}}]=\Tr\,i[H_{\mathrm{sm}},\chi]\varphi'(H_{\mathrm{sm}})$ as \eqref{eq:sigmaI}, and
$\mathcal I=2\pi\sigma_I$ is Theorem~\ref{thm:quantisation}.
\end{proof}

\noindent The identification with $\BDI$ is the one step imported from the Euclidean plane. Read at
$L_{\mathrm{tr}}$, \eqref{eq:APSint} yields
\begin{equation}\label{eq:direct}
  \widetilde\cR+\eta^{\mathrm{reg}}(+)-\eta^{\mathrm{reg}}(-)=\BDI(H_+,H_-),
\end{equation}
an identity between quantities built from the bulk symbols alone, for which a direct proof is open.

At first order $p=1$ the left side of \eqref{eq:direct} collapses to a difference of two half-plane rotations, one per medium. The interface index at the transparent condition, and hence the bulk-difference invariant, is $\cR[H_+]-\cR[H_-]$, each rotation computed on $\Om_+$ for one medium alone and neither, in general, an integer.

\begin{corollary}\label{cor:etaglue}
Let $p=1$ and let $H_\pm$ be $x$- and $y$-independent, satisfy \textup{[H1]--[H4]} with a common gap containing $E_0$, and have the same principal symbol. Then
\[
  \mathcal I(L_{\mathrm{tr}})=\cR[H_+]-\cR[H_-],\qquad\text{hence}\qquad \BDI(H_+,H_-)=\cR[H_+]-\cR[H_-]
\]
by Theorem~\ref{thm:R2}, $\cR[H_\pm]$ being the rotations of the half-plane problems on $\Om_+$.
\end{corollary}
\begin{proof}
At $p=1$ the Green form is the coefficient of $D_y$, part of the principal symbol, so $\cA_+=\cA_-$, and \textup{[C2]} and the admissibility of $L_{\mathrm{tr}}$ are automatic (Proposition~\ref{prop:tradm}); \eqref{eq:APSint} therefore holds at $L_{\mathrm{tr}}$ for the pairs $(H_+,H_-)$ and $(H_-,H_-)$. Its $\eta$-terms are computed from $(\widetilde\cA,L_{\mathrm{tr}},\widetilde\Lm^+(\pm))$ alone, and $\widetilde\Lm^+(\pm)=\Lm^+_{H_+}(\pm)\oplus\Lm^-_{H_-}(\pm)$ depends only on the principal symbols by Lemma~\ref{lem:endpoints}, \textup{[P]} being vacuous at $p=1$. So the $\eta$-difference $\delta\eta$ is the same for both pairs. For $(H_-,H_-)$ the glued operator is $H_-$ on $\R^2$, its own bulk operator, so $\mathcal I(L_{\mathrm{tr}})=0$. By \eqref{eq:APSint} and \eqref{eq:Rsplit},
\[
  \mathcal I(L_{\mathrm{tr}})=\cR[H_+]-\cR[H_-;\Om_-]+\delta\eta,\qquad 0=\cR[H_-]-\cR[H_-;\Om_-]+\delta\eta ,
\]
and subtracting gives the claim. At $p\ge2$ the argument fails: the Green form involves $A_1|_\Gam$ (Proposition~\ref{prop:green}\textup{(i)}), and $L_{\mathrm{tr}}$ is inadmissible, so the $\eta$-terms are undefined (Proposition~\ref{prop:tradm}\textup{(ii)}).
\end{proof}

\begin{remark}[momentum-dependent transmission conditions]\label{rem:jumpwinding}
For a transmission condition differential of order $d$ along $\Gam$, Theorem~\ref{thm:main} applies to
$\widetilde H$ by Theorem~\ref{thm:double}:
\[
  \mathcal I=\widetilde\cR-\cR_L
  +\eta\bigl(L_\infty,\widetilde\Lm^+(+)\bigr)-\eta\bigl(L_\infty,\widetilde\Lm^+(-)\bigr),
  \qquad\cR_L\in\Z,\quad|\cR_L|\le2\kp d\ \text{under \textup{[S]}}.
\]
The limit $L_\infty$ is determined by the top-order tangential part of the condition alone. At $p=1$, if $L_\infty=L_{\mathrm{tr}}$, which is admissible, the $\eta$-terms are those of $\mathcal I(L_{\mathrm{tr}})$,
so that $\mathcal I=\BDI(H_+,H_-)-\cR_L$ under the hypotheses of Theorem~\ref{thm:R2}: once its top-order part agrees with $L_{\mathrm{tr}}$, the lower-order terms of the condition change the index only through the integer $\cR_L$, computed by Lemma~\ref{lem:plucker}\textup{(iii)}. Remark~\ref{rem:polyrange}\textup{(a)} applied to $\widetilde H$, with $Y=\mathsf E_j$ so that $L_\infty=L_{\mathrm{tr}}$, gives such conditions of order one with $\cR_L=-sj$ for every $j\le2\kp$. For instance, let $H_+=H_-$ be the Dirac operator of \textup{(E1)}, so that $\mathcal I(L_{\mathrm{tr}})=0$, the glued operator being $H$ on $\R^2$, and impose, with $\alpha_\pm$ of \eqref{eq:diracA} evaluated at $y=0^\pm$ and $s=\pm1$,
\[
  \alpha_-(0^+)=\alpha_-(0^-),\qquad (1+isD_x)\,\alpha_+(0^+)+(1-isD_x)\,\alpha_+(0^-)=0 .
\]
In the chart of Proposition~\ref{prop:transparent} this is $U_{L(\xi)}=\mathrm{diag}\bigl(-(1-is\xi)/(1+is\xi),1\bigr)$, the case $j=1$, whose nontrivial entry is the unitary of the termination $\gamma_0+isD_x\gamma_1$ of $-\Delta$ in \textup{(E4)}; it tends to $\mathbf 1=U_{L_{\mathrm{tr}}}$ as $\xi\to\pm\infty$ and has $\cR_L=-\sgn s$, so that $\mathcal I=\sgn s$ although the two media are identical and the top-order part of the condition is that of $L_{\mathrm{tr}}$.
\end{remark}

\begin{remark}\label{rem:KS}
The approach of this section was shaped by \cite{KS}. There, half-space and interface Hamiltonians are the
self-adjoint extensions of one symmetric operator, labelled by von Neumann unitaries $u$, and their edge classes
lie in $K_1$ of an edge algebra. Two results organize the theory: a relative bulk--edge correspondence,
comparing Hamiltonians with the same boundary condition through the boundary map of a relative bulk class;
and a comparison of boundary conditions, by which two resolvent-affiliated extensions $u,u_0$ differ by the
class of $uu_0^*$, a relative winding in $\xi$ that is constant on connected components of resolvent-affiliated
conditions \cite[Thm.~4.13, Prop.~3.18]{KS}. At an interface the class is a bulk difference plus a correction
that vanishes at the transparent matching condition \cite[Thm.~4.15]{KS}.

The two settings match as follows. Their von Neumann unitary and $U_L$ of \eqref{eq:UL} are two charts on
$\Lag(\cA)\cong\U(\kp)$. Theirs is based at the deficiency spaces $\ker(\mathring H(\xi)^*\mp i)$, that is,
at the Calder\'on data at the energies $\pm i$, which move with $\xi$, so that only relative windings are
chart-independent. Ours is based at the eigenspaces of $\cA(0)$, fixed after transport, so that the Calder\'on
path itself has a winding $\cR$ and \eqref{eq:APSpoly} is an absolute formula. For local terminations their
resolvent affiliation, $\norm{(\Hh_u(\xi)-i)^{-1}}\to0$ as $|\xi|\to\infty$, corresponds to uniform
Shapiro--Lopatinski ellipticity, which at $p=1$ is admissibility (Remark~\ref{rem:SL}). For two
terminations $L,L_0$ with escaping edge spectrum, both their relative winding and \eqref{eq:APSpoly} compute
the difference of spectral flows, so their correction is
\[
  \mathcal I(L)-\mathcal I(L_0)=-(\cR_L-\cR_{L_0})
  +\bigl[\eta\bigl(L_\infty,\Lm^\Phi(+)\bigr)-\eta\bigl(L_\infty,\Lm^\Phi(-)\bigr)\bigr]
  -\bigl[\eta\bigl(L_{0,\infty},\Lm^\Phi(+)\bigr)-\eta\bigl(L_{0,\infty},\Lm^\Phi(-)\bigr)\bigr],
\]
the termination-dependent part of \eqref{eq:APSpoly}, split into a winding and two endpoint asymmetries. Its
constancy on components becomes the chamber structure of Theorem~\ref{thm:range} and the integrality of
$\cR_L$. Their Dirac half-space formula, $(\sgn m+\sgn a)/2$ \cite[(5.6)]{KS}, splits as in \textup{(E1)}:
the bulk contributes a half-integer, as $\cR$ does, and the termination the other half. Likewise
Theorem~\ref{thm:R2} and Remark~\ref{rem:jumpwinding} are their bulk-interface correspondence in these
coordinates, the correction vanishing at $L_{\mathrm{tr}}$.

Their approach is $K$-theoretic. It is formulated for covariant families of operators, a setting which accommodates in principle covariant disorder and magnetic fields, and it determines the index relative to one model computation. Ours treats elliptic or escaping terminations; the invariant is shown to be stable under deterministic perturbations, along norm-resolvent continuous paths and on bounded sets (Corollary~\ref{cor:stab}, Theorem~\ref{thm:local}), and translation invariance is used for the formula itself. The formula is absolute: $\mathcal I$ is expressed through a bulk symbol integral
and boundary data, the range over all terminations is explicit (Theorem~\ref{thm:range}), and the size of the
correction is bounded, $|\cR_L|\le\kp(d+p-1)$.
\end{remark}

\section{Examples}\label{sec:examples}
The examples below illustrate the theory. For each we give the Green form, the endpoints
$\Lm^\Phi(\pm)$ and the defect $\kp_\cap$, the rotation $\cR$, and the values of $\mathcal I$. By
Theorems~\ref{thm:main} and \ref{thm:range} these are, for every admissible termination differential
along $\Gam$ with escaping edge spectrum,
\[
  \mathcal I(L)=\mathcal I_0-j(L_\infty)-\cR_L,\qquad j\in\{0,\dots,\kp-\kp_\cap\},
\]
so the table lists $\cR$ and the chamber values $\mathcal I_0-j$. The winding $\cR_L$ shifts them by an integer and already for terminations of tangential order $d=1$ takes every value $|\cR_L|\le\kp$ (Remark~\ref{rem:polyrange}). Transported pointwise terminations are the case $\cR_L=0$. 

Derivations are given only when they are short.  The rest are numerical evaluations of the right-hand side, with no claim of complete rigor, marked as \emph{computed}. All computations of this section were performed by AI (Claude Opus 5.5), in double precision except for \textup{(E8)}, which used extended precision ($30+10p$ significant digits); a script reproducing them, \texttt{verify\_section8.py}, is provided as supplementary material. They use the linear algebra of section~\ref{ssec:symbolic}: an ordered Schur decomposition of the companion matrix (in \textup{(E8)}, its matrix sign function) gives orthonormal bases of the decaying data, and $\arg\det U_+$ is unwrapped along a grid of the momentum line, uniform on a bounded interval and geometric beyond it, refined where the phase moves quickly. No fiber operator is diagonalized, and edge spectra are quoted only where they are explicit. Each such evaluation determines $\cR$, together with the endpoint data, on a component of parameter space by Lemma~\ref{lem:Rconst}.

\begin{center}\small
\setlength{\tabcolsep}{4pt}
\begin{tabular}{@{}lcccccll@{}}
\toprule
& $p$ & $N$ & $\kp$ & \textup{[S]} & $\kp_\cap$ & $\cR$ & chamber values $\mathcal I_0-j$\\
\midrule
(E1) Dirac & $1$ & $2$ & $1$ & yes & $0$ & $-\tfrac12$ & $\{0,-1\}$\\
(E2) spin $s\in\tfrac12+\Z_{\ge0}$, $m>0$ & $1$ & $2s{+}1$ & $s{+}\tfrac12$ & yes & $0$ & $-\tfrac12\kp^2$
 & $-\tfrac12\kp(\kp-1),\dots,-\tfrac12\kp(\kp+1)$\\
(E3) RHG, $m$ layers & $1$ & $2m$ & $m$ & yes & $0$ & \eqref{eq:RHGR} & $\cR\pm\tfrac m2$, $m{+}1$ values\\
(E4) $-\Delta$ & $2$ & $1$ & $1$ & yes & $1$ & $0$ & $\{0\}$\\
(E5) regularized Dirac, $\epsilon,m>0$ & $2$ & $2$ & $2$ & yes & $2$ & $-1$ & $\{-1\}$\\
(E6) $p$-wave sc. & $2$ & $2$ & $2$ & yes & $2$ & $-\sgn c$ & $\{-\sgn c\}$\\
(E7) $d$-wave sc. & $2$ & $2$ & $2$ & no & $2$ & $-2\sgn c$ & $\{-2\sgn c\}$\\
(E8) $p$-fold Dirac, $m>0$ & $p$ & $2$ & $p$ & no & $0$ & $-\tfrac p2$ & $\{0,\dots,-p\}$\\
\bottomrule
\end{tabular}
\setlength{\tabcolsep}{6pt}
\end{center}
\noindent{\small For \textup{(E7)} and \textup{(E8)}, where \textup{[S]} fails, the columns $\kp_\cap$, $\cR$ and the chamber values refer to the transported endpoints $\Lm^\Phi(\pm)$ and the transported Calder\'on path of Theorem~\ref{thm:graded}. In \textup{(E7)} no pointwise termination is
admissible, and in \textup{(E8)} only the middle one or two values are carried (numerically) by elliptic
terminations.
}

\medskip\noindent \textup{[P]} and \textup{[M]} hold in \textup{(E1)}--\textup{(E7)}, in \textup{(E7)} for the
transported data and numerically, and the table splits accordingly. For first-order operators, $p=1$, $\kp_\cap=0$ (Proposition~\ref{prop:admtopology}\textup{(iii)}), and in the first-order examples of the table the endpoints moreover differ by $W'=-\mathbf 1$, so the
$\kp+1$ values fill $[\cR-\frac\kp2,\cR+\frac\kp2]$ (Theorem~\ref{thm:range}). This fails for general first-order systems: for $\kp=1$, $W'=-1$ forces $\cR\in\frac12+\Z$, while the first-order systems at the end of this section have non-half-integer rotations. At even $p$ the
path closes on $\cB_{\mathrm{high}}$, $\kp_\cap=\kp$, and $\mathcal I=\cR-\cR_L$
(Corollary~\ref{cor:evenrigid}, and Corollary~\ref{cor:loop} for the transported path of \textup{(E7)}): a
single value up to the winding.
\textup{[M]} and its associated rigidity do not hold for  \textup{(E8)}.

\paragraph{(E1) The Dirac operator.}\label{ssec:E1}
$H=D_x\sigma_3-D_y\sigma_2+m\sigma_1$, $m>0$, gap $(-m,m)$. Integrating by parts,
$\ip{v}{Hu}-\ip{Hv}{u}=-iv^*\sigma_2u|_{y=0}$, so
\begin{equation}\label{eq:diracA}
  \cA=-\sigma_2,\qquad e_\pm=\tfrac1{\sqrt2}(1,\mp i),\qquad
  \alpha_\pm=\tfrac1{\sqrt2}\bigl(w_1\pm iw_2\bigr),
\end{equation}
$e_\pm$ the eigenvectors of $\cA$ for $\pm1$; the chart \eqref{eq:UL} is $U_L=\alpha_+/\alpha_-$, and
the pointwise realizations are $\lambda\psi_1+\psi_2=0$, $\lambda\in\R\cup\{\infty\}$, with
$U_L=e^{-i\vartheta_L}$, $\vartheta_L:=2\arctan\lambda$. The model problem at $m=0$ has decaying
solutions $e^{-|\xi|y}(1,\sgn\xi)$, so $\Lm^+(\pm)=\C(1,\pm1)$ and
\[
  U_+(\pm)=\pm i,\qquad U_+(0)=-1,\qquad \cR=-\tfrac12 ,
\]
the path running $-i\to-1\to+i$. So $\kp_\cap=0$, the chambers are $|\lambda|<1$ and $|\lambda|>1$,
and \eqref{eq:APS} gives $\mathcal I=0$ and $-1$ there.  {The half-integer $\cR$ is the parity
anomaly of one Dirac cone} \cite{Witten16} while the termination supplies the other half. Here, exceptionally, the edge spectrum is explicit and confirms the orientation of \eqref{eq:three}: the single branch $E(\xi)=\xi\cos\vartheta_L-m\sin\vartheta_L$ is an $L^2$ eigenvalue on the half-line $\xi\sin\vartheta_L+m\cos\vartheta_L<0$ for every $\lambda$. It lies in the gap exactly when $\cos\vartheta_L<0$, i.e., $|\lambda|>1$, and then crosses $0$ with slope $\cos\vartheta_L<0$, so $\sfl=-1$.

\paragraph{(E2) The half-integer spin family.}\label{ssec:E2}
$H=D_xS_1+D_yS_2+mS_3$ with spin-$s$ matrices $S_i$, $s\in\tfrac12+\Z_{\ge0}$, and $m>0$: $\kp=s+\tfrac12$,
$\cA=S_2$, $D_+=\mathrm{diag}(s,s-1,\dots,\tfrac12)$ and $D_-=\mathrm{diag}(\tfrac12,\dots,s-1,s)$, the spectrum of $S_2$ being symmetric; for integer $s$, $S_2$ is singular and
\textup{[H2]} fails. With $\beta:=\arg(\xi+im)$ and $\Hh_0(\rho):=\rho S_1+S_2D_y$, the identity
$\Hh(\xi)=e^{i\beta S_2}\Hh_0(\rho)e^{-i\beta S_2}$, $\rho=(\xi^2+m^2)^{1/2}$, gives
$\Lm^+(\xi)=e^{i\beta S_2}\Lm_0$, with $\Lm_0$ the decaying data of $\Hh_0(\rho)$, independent of $\rho$. With
$U_0:=U_{\Lm_0}$ this makes the path explicit,
\begin{equation}\label{eq:spinpath}
  U_+(\xi)=e^{i\beta D_+}\,U_0\,e^{i\beta D_-},\qquad
  -i\,U_+^*\partial_\xi U_+=-\frac{m}{\xi^2+m^2}\bigl(U_+^*D_+U_++D_-\bigr)<0 ,
\end{equation}
so the path is monotone and, taking traces,
$
  \cR=-\tr D_+=-\tfrac12\tr|\cA|=-\tfrac12\kp^2 ,
$
independently of $m>0$. For $m<0$ the sign of \eqref{eq:spinpath} reverses and $\cR=+\tfrac12\kp^2$. The endpoints differ by $-\mathbf 1$.

Indeed, $\Hh_0(\rho)$ gives $D_y\psi=\rho K\psi$ with $K=-S_2^{-1}S_1$. Since
$\det(S_1+\zeta S_2)=(1+\zeta^2)^\kp\det S_1$ (rotate by $e^{-i\theta S_3}$, $\tan\theta=\zeta$), $\Lm^+(+)$ is
the generalized eigenspace of $K$ at $i$. The parity $\mathrm{diag}((-1)^{s-\mu})$ anticommutes with
$S_1,S_2$, hence commutes with $K$, and $\ker(K-i)=\ker S_+=\C|s\rangle$ is even. The generalized eigenspace $\Lm^+(+)$ of $K$ at $i$ is therefore parity-invariant, and its odd part, $K$-invariant, would contain an eigenvector of $K$ at $i$ if it were nonzero; so $\Lm^+(+)$ meets the odd sector trivially and lies in the even one. Both have dimension $\kp$, the even sector being spanned by the $s+\tfrac12$ vectors $|\mu\rangle$ with $s-\mu$ even, so
$\Lm^+(+)=\mathrm{span}\{|\mu\rangle:s-\mu\text{ even}\}$. This space is preserved by $S_2^2=\cA^2$, so its
graph matrix intertwines $|\cA|$ on the two blocks: $D_+U_0=U_0D_-$. By \eqref{eq:spinpath} at
$\beta=\pi$, $U_+(-)=e^{i\pi D_+}U_0e^{i\pi D_-}=U_0e^{2\pi iD_-}=-U_0$, the entries of $D_-$ being
half-integers. So $\kp_\cap=0$ and Theorem~\ref{thm:range}
gives $\tfrac12\kp(\kp-1)\le-\mathcal I\le\tfrac12\kp(\kp+1)$, all $\kp+1$ values attained: $\{0,-1\}$ at
$s=\tfrac12$, which is \textup{(E1)}, $\{-1,-2,-3\}$ at $s=\tfrac32$. For $\kp\ge2$, no pointwise termination can gap the edge, and the bound \eqref{eq:horm} is attained.

\paragraph{(E3) Gated rhombohedral graphene.}\label{ssec:E3}
Following \cite[(1)--(3)]{FrazierBal}, for $m\ge2$ layers
\[
  H=\bigl(\delta_{jl}(u_j+D_x\sigma_1+D_y\sigma_2)+\delta_{j,l-1}B+\delta_{j,l+1}B^*\bigr)_{j,l=1}^m,
  \qquad B=\gamma\begin{pmatrix}0&0\\1&0\end{pmatrix},\quad u_j=\tfrac{u}{m-1}\bigl(j-\tfrac{m+1}2\bigr),
\]
in the valley $\tau=1$, the other reversing every rotation by Lemma~\ref{lem:chartalg}\textup{(i)}.
Since $B$ and $u_j$ have order zero,
\begin{equation}\label{eq:RHGA}
  \sigma_p=\mathbf 1_m\otimes(\xi\sigma_1+\zeta\sigma_2),\qquad \cA=\mathbf 1_m\otimes\sigma_2 ,
\end{equation}
corresponding to $m$ copies of the Dirac symbol $\xi\sigma_1+\zeta\sigma_2$, whose endpoints $\C(1,0)$ and $\C(0,1)$ have
$U_+(\pm)=\pm1$, so $U_+(-)^{-1}U_+(+)=-\mathbf 1_m$,
$\kp_\cap=0$, and a pointwise termination imposes one scalar condition per layer. With
$\delta(\gamma,u):=(\tfrac14+\gamma^2(m-1)^2u^{-2})^{1/2}$ and
$\delta_k:=k+\lceil\tfrac m2\rceil-\tfrac m2$, the phases of \cite[(12)--(14)]{FrazierBal} are
$\sgn(u)j$, $1\le j\le\lfloor\tfrac m2\rfloor$, with $\delta_{j-1}<\delta(\gamma,u)<\delta_j$, the last
one unbounded above in $\delta$; they are separated by the closings of the gap at $\xi=\zeta=0$ at the
critical values of \cite[(13)]{FrazierBal}, which exist only for $m\ge4$. On the phase $\sgn(u)j$,
\begin{equation}\label{eq:RHGR}
  \cR=\sgn(u)\Bigl(\tfrac{m^2}4-\delta_j(\delta_j-1)\Bigr)%
\end{equation}
computed to five decimals for $2\le m\le6$ and in every phase, the ordered Schur decomposition of the companion matrix supplying
the decaying data. So $\mathcal I$ ranges over the $m+1$ integers of $[\cR-\tfrac m2,\cR+\tfrac m2]$,
each attained by a termination that decouples, $\lambda_k\psi^{(k)}_1+\psi^{(k)}_2=0$ on layer $k$:
\emph{$\mathcal I_0-\mathcal I$ is the number of layers with $\lambda_k>0$}, the Dirac chambers of this
symbol being $\lambda\gtrless0$ (computed), and the edge current of a terminated sample is free within a
window of width $m$. The value $0$ lies in the window only on the phase $\delta_j=\tfrac m2$, that is
$j=\lfloor\tfrac m2\rfloor$, where $|\cR|=\tfrac m2$; on every other phase $|\cR|>\tfrac m2$.

\begin{proposition}[the invariant of \cite{FrazierBal} is a difference of rotations]\label{prop:RHGBDI}
For an interface between a north phase $\sgn(u^N)j$ and a south phase $\sgn(u^S)k$, the invariant of
\cite[Thm.~2.3]{FrazierBal} is $\mathcal I(L_{\mathrm{tr}})=\cR_N-\cR_S$.
\end{proposition}
\begin{proof}
The two media differ only in the order-zero $u_j$. They therefore share \eqref{eq:RHGA} and are gapped at
$E=0$, so \textup{[C1]--[C3]} hold, with \textup{[C2](b)}, \textup{[C3]} and \eqref{eq:C2frozen} automatic
for equal principal symbols. Theorem~\ref{thm:R2} identifies the invariant of
\cite[Thm.~2.3]{FrazierBal} with $\mathcal I(L_{\mathrm{tr}})$. Both media are $y$-independent, so Corollary~\ref{cor:etaglue} gives $\mathcal I(L_{\mathrm{tr}})=\cR_N-\cR_S$. The value $u=0$ cannot serve as reference, the gap
closing there.
\end{proof}
\noindent With the computed values \eqref{eq:RHGR}, write $\rho_j:=\tfrac{m^2}4-\delta_j(\delta_j-1)$, so
that $\cR=\sgn(u)\rho_j$. The difference $\cR_N-\cR_S$ is then $\varsigma(\rho_j-\rho_k)$ for gate
potentials of common sign $\varsigma=\sgn(u^N)$, and $\varsigma(\rho_j+\rho_k)$ for opposite signs:
\[
  \cR_N-\cR_S=\begin{cases}
    \sgn(u^N)\bigl(\delta_k(\delta_k-1)-\delta_j(\delta_j-1)\bigr), & \sgn(u^N)=\sgn(u^S),\\[4pt]
    \sgn(u^N)\Bigl(\dfrac{m^2}{2}-\delta_j(\delta_j-1)-\delta_k(\delta_k-1)\Bigr),
      & \sgn(u^N)=-\sgn(u^S),
  \end{cases}
\]
which is \cite[(15)]{FrazierBal}. Conversely, \cite[(15)]{FrazierBal} determines every difference
$\cR_N-\cR_S$, consistently with \eqref{eq:RHGR}.

\paragraph{Second order examples.} At $p=2$, with $A_2$ independent of $y$ near $\Gam$, \eqref{eq:blocks} gives
\begin{equation}\label{eq:green2}
  \cA(x,\xi)=\begin{pmatrix}a_1(x,0,\xi)&A_2\\ A_2&0\end{pmatrix}
\end{equation}
in the blocks of $(\gamma_0u,\gamma_1u)$. So the Dirichlet condition $\{\gamma_0=0\}$ is Lagrangian, and when $a_1=0$, as in \textup{(E4)}, a
Robin boundary condition $\partial_yu=\mu u$, i.e., $\gamma_1u=-i\mu\gamma_0u$, has Lagrangian
$\C(1,-i\mu)$. In \textup{(E4)}, \textup{(E5)}, \textup{(E6)} the principal symbol is a scalar times a
fixed invertible matrix, so Lemma~\ref{lem:scalarp2} gives $\Lm^+(\pm)=\{\gamma_0=0\}$ and
$\mathcal I=\cR-\cR_L$.

\paragraph{(E4) The Laplacian.}\label{ssec:E4}
$H=-\Delta$: $\kp=1$, $\cA=\begin{psmallmatrix}0&1\\1&0\end{psmallmatrix}$, and the {\em pointwise} realizations
are the Robin lines $L_\mu=\C(1,-i\mu)$, $\mu=\infty$ being Dirichlet. For $E<0$,
$\Lm^+(\xi,E)=\C(1,ik)$ with $k=k(\xi,E):=\sqrt{\xi^2-E}$, so the endpoint is the graded limit
\[
  \Lm^+(\pm)=\lim_{t\to\infty}\Th_t\,\C(1,i)=\C(0,1),
\]
the Dirichlet Lagrangian, and not the frozen datum $\C(1,i)$. The loop is contractible since $k$ is even in $\xi$. Thus $\cR=0$, so $\mathcal I=0$ for every Robin parameter, and also for Dirichlet, by Proposition~\ref{prop:wall} with $n^\flat(\pm)=1$.

{{\em Momentum dependence} changes this drastically.} For
$\mu(\xi)=s\xi$, i.e., termination $\partial_y u = sD_x u$, $U_{L(\xi)}=(1-is\xi)/(1+is\xi)$ gives $\cR_L=-\sgn s$, and, the endpoints $L(\pm)$ of the termination still being Dirichlet, Proposition~\ref{prop:wall} gives $\mathcal I=\sgn s$ for $|s|>1$, where the branch $E(\xi)=(1-s^2)\xi^2$ on $s\xi<0$ enters the gap, and $\mathcal I=0$ for $|s|<1$, where
$n^\flat(-)-n^\flat(+)=-\sgn s$ cancels $-\cR_L$. This is the only example in the paper with a nonzero wall correction $n^\flat(-)-n^\flat(+)$; for $s<0$ it is the case $\kp=1$, $W(\pm)=1$ discussed at the end of the proof of Proposition~\ref{prop:wall}. The
same values arise from admissible terminations: $T=\gamma_0+isD_x\gamma_1$, $s=\pm1$, has endpoint the
Neumann space, $\cR_L=-s$ and $W(\pm)=-1$, so \eqref{eq:APSpoly} gives $\mathcal I=s$. So $\mathcal I$
takes the values $\{-1,0,1\}$ although it vanishes on every pointwise termination. For $s=1$ we recall that the bulk commutator alone would give a non-integer (Remark~\ref{rem:E2current}), showing that terminations do carry current.

\paragraph{(E5) The regularized Dirac operator.}\label{ssec:E5}
Following \cite{JudTauber}, with $\epsilon\ne0$,
\begin{equation}\label{eq:JTop}
  H=-D_x\sigma_1-D_y\sigma_2+\bigl(m-\epsilon(D_x^2+D_y^2)\bigr)\sigma_3 ,
\end{equation}
$\kp=2$, $\cA=\begin{psmallmatrix}-\sigma_2&-\epsilon\sigma_3\\ -\epsilon\sigma_3&0\end{psmallmatrix}$,
and \textup{[H3]} holds at $E_0=0$ exactly when $m\ne0$, so the parameter region has four components.
The decaying data are explicit, the discriminant $1-4\epsilon m$ of the fiber equation being independent of $\xi$, and for $\epsilon,m>0$,
\[
  \mathcal I=\cR-\cR_L=-1-\cR_L,
\]
with $\cR=-1$. The path closes, so $\cR\in\Z$ is constant on each component (Lemma~\ref{lem:Rconst}). More generally $\cR=-\frac12(\sgn{m}+\sgn{\epsilon})$ \cite{Bal2019,JudTauber}.
{Here, the regularization makes the rotation an integer}, against $-\tfrac12$ for \textup{(E1)}, and all the
variation of the edge index reported in \cite{JudTauber} comes from the winding $\cR_L$ of their
momentum-dependent terminations. Dirichlet, on the wall, has $\mathcal I=-1$ by \eqref{eq:IReq}, with
one edge state at $\xi=0$; $\sigma_3$ being indefinite, other terminations on the wall need not be
elliptic (Lemma~\ref{lem:scalarp2}).

\paragraph{(E6) and (E7) Superconductors.}\label{ssec:E67}
The $p$- and $d$-wave Hamiltonians of \cite{QuinnBal} (see \cite{Volovik03}) are
\[
\begin{aligned}
  H_p&=\Bigl(\frac{D_x^2+D_y^2}{2m}-\mu\Bigr)\sigma_1+\tfrac12\bigl(cD_y+D_yc\bigr)\sigma_2
       +c_0D_x\sigma_3,\\
  H_d&=\Bigl(\frac{D_x^2+D_y^2}{2m}-\mu\Bigr)\sigma_1+c_0\bigl(D_y^2-D_x^2\bigr)\sigma_2
       +\tfrac12D_x\bigl(cD_y+D_yc\bigr)\sigma_3,
\end{aligned}
\]
$m,\mu,c_0>0$, $c=c(y)$, with gap at $E_0=0$ and $\kp=2$. By \eqref{eq:green2},
$a_1=c\,\sigma_2$ for $H_p$ but $a_1=c\,\xi\sigma_3$ for $H_d$: {\textup{[S]} holds for $H_p$ and
fails for $H_d$.}

\emph{\textup{(E6)}.} $\mathcal I=-\sgn c-\cR_L$ (computed), with plenty of admissible pointwise terminations,
every $L$ transverse to $\cB_{\mathrm{high}}$. For two $p$-wave media, with $c$ equal to $c_\pm$ for $\pm y\ge y_0$ and continuous at $y=0$, so that $\cA_+=\cA_-$ ($\cA_{00}=c\,\sigma_2$ involves $c|_\Gam$), both paths close,
$\widetilde\Lm^+(+)=\widetilde\Lm^+(-)$, and Corollary~\ref{cor:loop} gives
$\mathcal I=\widetilde\cR=\sgn(c_-)-\sgn(c_+)$ for every admissible transmission condition. Here the rotations are those of the $y$-independent bulks by Corollary~\ref{cor:bulkR} and its reflection to $\Om_-$, the paths being closed, and $\cR[H;\Om_-]=\cR[H]$ for these bulks, as evaluating the rotations shows. {This is
$2\pi\sigma_I(H_p)$ of \cite{QuinnBal}}, and at the transparent transmission condition, inadmissible by
Proposition~\ref{prop:tradm}\textup{(ii)} but elliptic by \textup{(i)} there, the two bulks sharing
their principal symbol, Theorem~\ref{thm:R2} gives the same value, so the correction
$n^\flat(-)-n^\flat(+)$ of \eqref{eq:IReq} vanishes at $L_{\mathrm{tr}}$.

\emph{\textup{(E7)}.} $\cA$ is affine in $\xi$, so the transport truncates to $\Phi(\xi)=\mathbf 1+\xi\mathsf N$
with $\mathsf N=\begin{psmallmatrix}0&0\\ n&0\end{psmallmatrix}$, $n=-\tfrac c2A_2^{-1}\sigma_3$ invertible.
It fixes $\cB_{\mathrm{high}}$, so the transported path still closes, \eqref{eq:Ehyp} holding for the
transported datum at every parameter value tested. In this example, no pointwise termination is admissible. A fixed $\kp$-plane $L$ is a termination only if it is $\cA(\xi)$-isotropic for every $\xi$; by \eqref{eq:green2}, with $a_1=c\,\xi\sigma_3$, this forces $\pi_0(L)$ to be $\sigma_3$-isotropic, hence of dimension $\le1$, so that $L\cap\cB_{\mathrm{high}}\ne0$. Since $\mathsf N$ vanishes on $\cB_{\mathrm{high}}$, the transported $L^\Phi(\xi)=(\mathbf 1-\xi\mathsf N)L$ contains $L\cap\cB_{\mathrm{high}}$ for every $\xi$, so its limits meet $\cB_{\mathrm{high}}$, the common endpoint. The
transported pointwise terminations $\Phi L_0$ with $L_0\cap\cB_{\mathrm{high}}=0$ are admissible and give
\[
  \mathcal I=\cR=-2\sgn c\qquad\text{(computed)},
\]
 \eqref{eq:APSfinite} returning $-2$ independently of the
cut-off. A termination differential of order $d$ along $\Gam$ adds $-\cR_{L^\Phi}$, where by
Lemma~\ref{lem:plucker} $\cR_{L^\Phi}=n_+(\det C_-)-n_+(\det C_+)$, with $n_+$ the number of roots in
$\Im\xi>0$ and $C_\pm$ the $\alpha_\pm$-blocks of $C(\xi)=B(\xi)\Phi(\xi)\Upsilon^*$, so that $|\cR_{L^\Phi}|\le2(d+1)$; every integer occurs
at some order. For instance, the second-order termination
$\bigl(A_2+iD_x(K_0+D_xK_1)n\bigr)\gamma_0u=i(K_0+D_xK_1)\gamma_1u$, $K_0,K_1$ Hermitian and $K_1$ zero or
invertible, is admissible with $\cR_{L^\Phi}=\operatorname{sig}K_1$, so that
$\mathcal I=-2\sgn c-\operatorname{sig}K_1$ (computed: $-4,-2,0$ for $c>0$ and
$K_1=\mathbf 1,\mathrm{diag}(1,-1),-\mathbf 1$). 

\paragraph{(E8) The $p$-fold Dirac cone: \textup{[P]} without \textup{[M]}.}\label{ssec:E8}
This example separates the hypotheses \textup{[P]} and \textup{[M]}, and admissibility from
ellipticity. Most statements below are obtained numerically. Let
\[
  H_p=\Re\bigl(D_x+iD_y\bigr)^{p}\sigma_1+\Im\bigl(D_x+iD_y\bigr)^{p}\sigma_2+m\sigma_3,\qquad m\ne0,
\]
so $N=2$, $\kp=p$, the gap is $(-|m|,|m|)$, and $p=1$ is \textup{(E1)}. The coefficients are
Hermitian, $a_k=\binom pk\xi^{p-k}\tau_k$ with $\tau_k=\cos\frac{k\pi}2\sigma_1+\sin\frac{k\pi}2\sigma_2$,
so $\cA_{jl}=\binom p{j+l+1}\xi^{p-1-j-l}\tau_{j+l+1}$ for $j+l<p$: \textup{[S]} fails for $p\ge2$ and
everything is read in the transported frame. At $\xi=\pm1$ the root $\zeta=i$ has multiplicity $p$
and a one-dimensional kernel, and its Jordan chain fills every trace slot,
\[
  V_+=\C^p\otimes e_1,\qquad V_-=\C^p\otimes e_2 .
\]
The blocks of $\cA(0)^{-1}(\cA(\xi)-\cA(0))$ are multiples of products of two $\tau_k$, hence diagonal, so $\Phi$ preserves $\C^p\otimes e_1$ and $\C^p\otimes e_2$, and $V^\Phi_\pm=V_\pm$. These are $\Th_t$-invariant, so $V_\pm^\uparrow=V_\pm$, and numerically $\Lm^\Phi(\pm)=V_\pm$ at the
rate $O(|\xi|^{-1})$, which is \textup{[P]} in the transported frame; but $\overline w(V_\pm)=\frac12p(p-1)<w^*$, so \textup{[M]} fails for $p\ge2$. Since
$V_+\pitchfork V_-$, $\kp_\cap=0$ at every order, even $p$ included, and there are $p+1$ chambers.

\emph{Range.} $\cR=-\frac p2\sgn m$ (computed for $p\le5$). Over the transported
pointwise terminations $L=\Phi L_0$ the right-hand side of \eqref{eq:APSpoly} is integral to within $10^{-7}$ and
equals $-j(L)$, so
\begin{equation}\label{eq:E8range}
  \mathcal I\in\{0,-1,\dots,-p\},\qquad\text{all values attained.}
\end{equation}
At $p=2,4$ this shows, numerically, that without \textup{[S]} Corollary~\ref{cor:evenrigid} needs \textup{[M]} for the transported data. Even-order rigidity
comes from \textup{[M]}, not from the parity of $p$. The data being $x$-independent, Theorem~\ref{thm:fiberexist} makes every value in \eqref{eq:E8range} a bona fide
invariant once the edge spectrum escapes. Without \textup{[M]} Lemma~\ref{lem:endpoints} does not give the uniformity in energy on which Proposition~\ref{prop:exact}\textup{(i)} rests, but here it holds numerically on the whole gap: for $p=2,\dots,5$, $m=1$ and $E\in\{0,\pm\tfrac12,\pm0.99\}$, the largest principal angle between $\Lm^\Phi(\xi,E)$ and $V_\pm$ decays like $C_p\,|m\mp E|/|\xi|$ for $|\xi|$ large, with $C_p\le\tfrac12$. So every admissible termination escapes by Proposition~\ref{prop:exact}\textup{(i)}, and every value in \eqref{eq:E8range} is an invariant.

\emph{Ellipticity.} Numerically, admissibility does not imply it here: $(H_p,L)$ is elliptic exactly for the
transported terminations whose lowest-weight part has one direction in each trace slot, transverse to $e_1$
and to $e_2$. The generic admissible $L_0$ is not elliptic. On those elliptic terminations
$\eta(+)-\eta(-)$ is $0$ for even $p$ and $\pm\frac12$ for odd $p$, so that
\[
  \mathcal I=\cR\ \ (p\ \text{even}),\qquad
  \mathcal I\in\bigl\{\cR\pm\tfrac12\bigr\}\ \ (p\ \text{odd}),
\]
that is $\{-1\}$, $\{-2,-1\}$, $\{-2\}$, $\{-3,-2\}$ at $p=2,3,4,5$ for $m>0$, the middle chamber or the two
middle ones. All $p+1$ values are invariants, and ellipticity singles out the middle ones as protected
against deformations within the elliptic terminations, a set that is not open, $L\mapsto L^\downarrow$ being discontinuous. These statements, and the value of $\cR$, are purely based on numerical simulations.

\paragraph{Interfaces: the velocity flip.}\label{ssec:vflip}
At $p=1$ let $\cA_+=\cA_-=\cA$ and $a_{10}^-=\lambda a_{10}^+$ with $\lambda\in\R\setminus\{0\}$, $a_{10}^\pm$ the
coefficients of $D_x$. The
decaying data at $\xi=\pm\infty$ are spectral subspaces of $\cA^{-1}a_{10}^\pm$, which has no real eigenvalue by \textup{[H2]}, for one half plane, the
growing data for the other, and $\cA^{-1}a_{10}^-=\lambda\cA^{-1}a_{10}^+$; so \eqref{eq:C2frozen} and
\eqref{eq:C2det} hold iff $\lambda>0$, and for $\lambda<0$ the failure is maximal,
$\Lm^+_{H_+}(\pm)=\Lm^-_{H_-}(\pm)$. For $H_\pm=\pm D_x\sigma_1+D_y\sigma_2+m\sigma_3$, the case
$\lambda=-1$ with $\cA=\sigma_2$, flipping the tangential velocity thus exchanges decaying and growing
data, and \textup{[C2]}\textup{(b)} fails at $L_{\mathrm{tr}}$, to which we assign no value.  This is consistent with \cite{FrazierBal25}. 
The generic transmission conditions are however unaffected: $\widetilde\Lm^+(\pm)$ are transverse, and $\mathcal I$ takes three consecutive values on the three chambers.

\section*{Acknowledgments.} This work was supported in part by NSF grant DMS-2606421 and ONR grant
N00014-26-1-2017.
\section*{Use of generative AI.}
Generative AI (Claude Opus 5.5) was used to perform the numerical evaluations of section~\ref{sec:examples} and to write the script reproducing them. It was also used to audit the manuscript for internal consistency and gaps in its arguments and proofs, and to help with bibliographical data. All such output was checked by the author, who is alone accountable for the correctness of statements and proofs.
\appendix
\numberwithin{equation}{section}

\section{Background material}\label{app:import}

We record the grading of the Cauchy data and the graded limits used in section~\ref{sec:green}, facts on
Lagrangian Grassmannians and the Maslov index, and the results of the Boutet de Monvel calculus used in
sections~\ref{sec:green} and \ref{sec:obs}, with a right inverse for the non-normal rows of the examples.

\paragraph{The grading of $\cB$, the Grassmannian, and graded limits.}\label{app:grass}
We follow \cite[Ch.~1 \S5]{GriffithsHarris} for the Grassmannian and its Pl\"ucker embedding. We use the grading and the filtrations \eqref{eq:filtration} of $\cB$ throughout, together
with $\Th_t:=\mathrm{diag}(1,t,\dots,t^{p-1})\otimes\mathbf 1_N$, $t>0$. Fix a basis
$e_1,\dots,e_{2\kp}$ of $\cB$ adapted to the grading and ordered by weight: $\Th_te_i=t^{\lambda_i}e_i$
with $\lambda_1\le\dots\le\lambda_{2\kp}$, the weight $\lambda_i=j$ occurring for the $N$ indices
spanning $\cB_j$, $j=0,\dots,p-1$. So $\cB_{\mathrm{high}}$ is spanned by the $e_i$ with
$\lambda_i\ge\lceil p/2\rceil$, the last $N\lfloor p/2\rfloor$ of them. The Pl\"ucker embedding
$\Gr(\kp,\cB)\hookrightarrow\mathbb P(\bigwedge^\kp\cB)$ sends a $\kp$-plane $P$ to the line spanned
by $v_1\wedge\dots\wedge v_\kp$ for any frame; it is injective with closed image, and in the basis
$\{e_I\}$ of $\bigwedge^\kp\cB$ indexed by the $\kp$-element subsets $I\subseteq\{1,\dots,2\kp\}$ the
coordinates $p_I(P)$ are the $\kp\times\kp$ minors of a frame. The
decomposable vectors, those of the form $v_1\wedge\dots\wedge v_\kp$, are defined by the Pl\"ucker
relations and form a closed cone. The scaling group acts diagonally,
\begin{equation}\label{eq:weights}
  \Th_te_I=t^{w_I}e_I,\qquad w_I:=\sum_{i\in I}\lambda_i\in\{0,\dots,\kp(p-1)\},
\end{equation}
and we put
\begin{equation}\label{eq:wstar}
  w^*:=\max_{|I|=\kp}w_I,\qquad
  \underline w(P):=\min\{w_I:p_I(P)\ne0\},\qquad \overline w(P):=\max\{w_I:p_I(P)\ne0\} ,
\end{equation}
so that $\underline w(P)\le \overline w(P)\le w^*$, and we write $w_{\mathrm{tot}}:=\sum_{i=1}^{2\kp}\lambda_i=w_I+w_{I^c}$ for the total weight. Since the coordinates of $\Th_t^{\pm1}P$ are the monomials
$t^{\pm w_I}p_I(P)$, both limits $P^{\uparrow}:=\lim_{t\to\infty}\Th_tP$ and
$P^{\downarrow}:=\lim_{t\to\infty}\Th_t^{-1}P$ exist for every $P$, with no expansion required and
with normalized representatives
\begin{equation}\label{eq:torusreps}
  \varpi_{P^{\uparrow}}=\!\!\sum_{w_I=\overline w(P)}\!\!p_I(P)\,e_I,\qquad
  \varpi_{P^{\downarrow}}=\!\!\sum_{w_I=\underline w(P)}\!\!p_I(P)\,e_I ,
\end{equation}
both decomposable by the above. Each slot of the trace grading carries $N$ basis vectors of equal
weight and $\kp=\frac12pN$, so a $\kp$-set of maximal weight contains the indices of $\cB_{\mathrm{high}}$, and
for even $p$ consists of them. By \eqref{eq:torusreps} the Pl\"ucker coordinates of $P^\uparrow$ then vanish
off the sets containing any such index $i$, and a $\kp$-plane with this property contains $e_i$; therefore
\begin{equation}\label{eq:highinlimit}
  \cB_{\mathrm{high}}\subseteq P^{\uparrow}\qquad\text{whenever }\overline w(P)=w^*.
\end{equation}
More generally, with $\pi_m:\cB\to\cB_m$ the projection on the slot $\cB_m$, every $P$ has
\[
  P^{\uparrow}=\bigoplus_{m<p}\pi_m\bigl(P\cap G^m\bigr),\qquad
  P^{\downarrow}=\bigoplus_{m<p}\pi_m\bigl(P\cap F_m\bigr):
\]
in a basis of $P$ adapted to the filtration $P\cap G^m$, a vector $v\in(P\cap G^m)\setminus G^{m-1}$ has
$t^{-m}\Th_tv\to\pi_mv$, and these limits are independent, $\pi_m$ being injective on
$(P\cap G^m)/(P\cap G^{m-1})$; likewise for $\Th_t^{-1}$ and $F_m$. For a field $L(\xi)=\ker b(\xi)$ polynomial
in $\xi$, the Pl\"ucker coordinates of $\Th_t^{-1}L(\pm t)=\ker\bigl(b(\pm t)\Th_t\bigr)$ are polynomial in $t$, so
$L^\downarrow(x,\pm)$ of section~\ref{sec:setting} exists, and equals $\ker\sigma_T(x,\pm)$ when $\sigma_T$ has rank $\kp$.
Lemma~\ref{lem:endpoints} needs a limit along a curve, not covered by \eqref{eq:torusreps}, as follows.
\begin{lemma}[graded limits of curves]\label{lem:gradedcurve}
Let $\epsilon\mapsto P(\epsilon)\in\Gr(\kp,\cB)$ be $C^\infty$ on $[0,\epsilon_0)$. Then
$P^\dagger:=\lim_{t\to\infty}\Th_tP(t^{-1})$ exists in $\Gr(\kp,\cB)$, is approached at rate $O(t^{-1})$
with $t$-derivative $O(t^{-2})$, and $P^\dagger=P(0)^{\uparrow}$ if $\overline w(P(0))=w^*$, the constants then
depending only on the $C^1$ norm of a frame of $P$ and on a lower bound for the maximal-weight coordinates of
$P(0)$ in an orthonormal frame.
\end{lemma}
\begin{proof}
Let $\varpi(\epsilon)\in\bigwedge^\kp\cB\setminus\{0\}$ be the wedge of a $C^\infty$ frame of $P(\epsilon)$,
on whose Pl\"ucker coordinates $\Th_t$ acts by \eqref{eq:weights}. If $\overline w(P(0))=w^*$, the terms of
$t^{-w^*}\Th_t\varpi(t^{-1})=\sum_It^{w_I-w^*}p_I(\varpi(t^{-1}))\,e_I$ with $w_I<w^*$, or with
$w_I=w^*$ and $p_I(\varpi(0))=0$, are $O(t^{-1})$, and the others are $p_I(\varpi(0))+O(t^{-1})$; so the
limit is the maximal-weight truncation of $\varpi(0)$, which spans $P(0)^{\uparrow}$ by
\eqref{eq:torusreps}, and the $t$-derivative is $O(t^{-2})$.

In general, let $o_I\in\N\cup\{\infty\}$ be the order of vanishing of $p_I(\varpi(\epsilon))$ at
$\epsilon=0$, so that $p_I(\varpi(\epsilon))=\epsilon^{o_I}(a_I+O(\epsilon))$ with $a_I\ne0$ when
$o_I<\infty$, and $p_I(\varpi(\epsilon))=O(\epsilon^n)$ for every $n$ when $o_I=\infty$. Since $\varpi(0)\ne0$, some $p_I(\varpi(0))$ is nonzero, that is $o_I=0$ for some $I$, so $m:=\max_I(w_I-o_I)\ge\min_Iw_I$ is finite. With $t=\epsilon^{-1}$,
\[
  t^{-m}\,\Th_t\,\varpi(t^{-1})=\sum_{o_I<\infty}t^{\,(w_I-o_I)-m}\bigl(a_I+O(t^{-1})\bigr)e_I+O(t^{-n})
  \ \longrightarrow\ \varpi^\dagger:=\sum_{w_I-o_I=m}a_I\,e_I\ \ne0 ,
\]
every exponent being a non-positive integer; the convergence is at rate $O(t^{-1})$ and differentiating in
$t$ gives $O(t^{-2})$. The decomposable vectors form a closed cone, so $\varpi^\dagger$ is decomposable and
spans a $\kp$-plane $P^\dagger$; the two rates transfer to $\Gr(\kp,\cB)$ because the Pl\"ucker embedding
is a homeomorphism onto its image with smooth inverse near $\varpi^\dagger\ne0$.
\end{proof}

\paragraph{Elements of symplectic geometry.}
We follow \cite{BBF98,Gontier23} for the Maslov index in complex symplectic spaces and its relation to spectral flow, Furutani's survey \cite{Furutani04} for the Fredholm--Lagrangian--Grassmannian, and Booss-Bavnbek--Wojciechowski \cite{BBW} for the Calder\'on
projector and the Grassmannian of terminations.

\paragraph{Lagrangian Grassmannian.} Let $q(w)=w^*\cA w$ be a nondegenerate Hermitian form of signature zero on $\cB=\C^{pN}\cong\C^{2\kp}$, $\kp=\tfrac12pN$, and $\omega(v,w):=i\ip{v}{\cA w}$. The form $\omega$ is skew-Hermitian, $\omega(w,v)=-\overline{\omega(v,w)}$, a symplectic form in the complex sense of \cite{BBF98,Gontier23}. So $\dim_\C\cB=2\kp$, a $q$-Lagrangian has $\dim_\C=\kp$, $\U(\kp)$ has $\dim_\R=\kp^2$, and $\Gr(j,\C^\kp)$ has $\dim_\R=2j(\kp-j)$. A subspace is \emph{$q$-isotropic} if $q$ vanishes on it identically and \emph{$q$-Lagrangian} if it is maximal $q$-isotropic, equivalently $L=L^{\perp_q}$, equivalently $q$-isotropic of dimension $\kp$. Then $\Lag(\cA)$ denotes the set of $q$-Lagrangians.

Three facts are used. First, the form is Hermitian on a complex space, so $\Lag(\cA)\cong\U(\kp)$ rather than
its real counterpart; it has $\pi_1\cong\Z$, detected by the determinant that underlies the Maslov theory.
Second, $\Lag(\cA)$ carries a group structure pulled back from $\U(\kp)$ through the chart \eqref{eq:UL},
dependent on the base point and not canonical, used in Proposition~\ref{prop:transparent}. Third, the
codimension of the corank strata
\[
  \Sigma^{(j)}_{U_0}=\bigl\{U:\dim\ker(U-U_0)=j\bigr\}
\]
of the wall \eqref{eq:wall}. Translating so that
$U_0=\mathbf 1$, a unitary with $\dim\ker(U-\mathbf 1)=j$ is determined by that kernel, a point of
$\Gr(j,\C^\kp)$ of real dimension $2j(\kp-j)$, together with a unitary on the orthogonal complement
having no eigenvalue $1$, an open subset of $\U(\kp-j)$ of real dimension $(\kp-j)^2$. So
\[
  \dim_\R\Sigma^{(j)}_{U_0}=2j(\kp-j)+(\kp-j)^2=\kp^2-j^2,
  \qquad \operatorname{codim}_\R\Sigma^{(j)}_{U_0}=j^2 ,
\]
against $\dim_\R\U(\kp)=\kp^2$. Therefore, $\Sigma^{(1)}_{U_0}$ is a hypersurface, open and dense in
$\Sigma_{U_0}$; $\Sigma^{(2)}_{U_0}$ has codimension $4$, which is what the wall analysis of
section~\ref{sec:range} needs; and
$\Sigma^{(\kp)}_{U_0}$ is the single point $U_0$.
Finally, if $U_1\ne U_2$ then $\Sigma_{U_1}\cap\Sigma_{U_2}$ has codimension $\ge2$. A point $U$ of it has unit
vectors $v,w$ with $Uv=U_1v$ and $Uw=U_2w$. If $v$ and $w$ are proportional, then $v\in\mathsf K:=\ker(U_1-U_2)\ne\C^\kp$,
and, $U$ being fixed on $v$ up to $\U(\kp-1)$, these $U$ form a set of dimension $\le2(\dim\mathsf K-1)+(\kp-1)^2\le\kp^2-3$.
Otherwise $U$ is fixed on $\mathrm{span}\{v,w\}$ up to $\U(\kp-2)$, and the lines of $v,w$ satisfy
$\ip{U_1v}{U_2w}=\ip vw$, an equation holomorphic in $w$ and nontrivial for $v\notin\mathsf K$, so these $U$ form a
set of dimension $\le4(\kp-1)-2+(\kp-2)^2=\kp^2-2$.

\noindent Two further properties of the chart \eqref{eq:UL} are used in section~\ref{sec:interface}.

\begin{lemma}[two properties of the chart]\label{lem:chartalg}
Let $\cA$ be nondegenerate of signature zero on $\cB$, and let $U_L^{(\cA)}$ denote the image of
$L\in\Lag(\cA)$ under \eqref{eq:UL} formed with $\cA$.
\begin{enumerate}[label=\textup{(\roman*)},nosep]
\item \emph{\textup{(}inversion\textup{)}} $\Lag(-\cA)=\Lag(\cA)$ and
$\det U_L^{(-\cA)}=\varepsilon_\cA\bigl(\det U_L^{(\cA)}\bigr)^{-1}$ for every $L$, with
$\varepsilon_\cA=\pm1$ independent of $L$. Hence the rotation of any path,
computed in the chart of $-\cA$, is minus its rotation computed in the chart of $\cA$.
\item \emph{\textup{(}block factorization\textup{)}} If $\cB=\cB_1\oplus\cB_2$ with
$\cA=\cA_1\oplus\cA_2$, each $\cA_i$ nondegenerate of signature zero, and if $L=L_1\oplus L_2$ with
$L_i\in\Lag(\cA_i)$, then $\det U_L^{(\cA)}=\varepsilon\,\det U_{L_1}^{(\cA_1)}\det U_{L_2}^{(\cA_2)}$
with $\varepsilon=\pm1$ independent of $L_1,L_2$. Hence rotations of split paths add.
\end{enumerate}
\end{lemma}
\begin{proof}
The signs come from the ordering of the eigenvectors, immaterial by the proof of Proposition~\ref{prop:U}. (i) In the frame of Proposition~\ref{prop:U}, $\cA=\mathrm{diag}(D_+,-D_-)$, so
$-\cA=\mathrm{diag}(-D_+,D_-)$: passing to $-\cA$ exchanges the two blocks, the
$q$-positive block of $-\cA$ being the old negative one, with moduli $D_-$, and conversely.
A Lagrangian is unchanged by $\cA\mapsto-\cA$, and the graph $\{(B\alpha_-,\alpha_-)\}$ over the old negative block is the
graph of $B^{-1}$ over the new one. Hence by \eqref{eq:detUL}
\[
  \det U_L^{(-\cA)}=\pm\,\det\bigl(D_-^{1/2}\bigr)\,\det B^{-1}\,\det\bigl(D_+^{-1/2}\bigr)
  =\pm\,\varkappa_\cA^{-1}\bigl(\det B\bigr)^{-1}
  =\pm\bigl(\det U_L^{(\cA)}\bigr)^{-1}.
\]

(ii) In the direct sum of the two frames the graph map of $L$ is $B_1\oplus B_2$ and the diagonals are
$D_{+,1}\oplus D_{+,2}$ and $D_{-,1}\oplus D_{-,2}$, so the determinant factorizes.
\end{proof}

\noindent The identification $\Lag(\cA)\cong\U(\kp)$ of \eqref{eq:UL} is simultaneously the Cayley
parametrization of Lagrangians \cite{Furutani04,Gontier23}, von Neumann's parametrization of the self-adjoint
extensions of a symmetric operator with deficiency indices $(\kp,\kp)$ \cite{BHS20}, and the graph coordinates on the Grassmannian
of terminations; Remark~\ref{rem:KS} compares the von Neumann chart of \cite{KS} with \eqref{eq:UL}.

\paragraph{Maslov cycle, crossing forms, and Maslov index.} The wall, or \emph{Maslov cycle}, of $L_0\in\Lag(\cA)$ is
$\Sigma_{U_{L_0}}$ of \eqref{eq:wall}, stratified by the corank as above. The \emph{crossing
form} of a smooth path $\lambda\mapsto\Lm(\lambda)$ relative to a fixed $L$, at a crossing $\lambda_0$, is the
quadratic form on $\Lm(\lambda_0)\cap L$
\begin{equation}\label{eq:crossdef}
  v\longmapsto\frac{d}{d\lambda}\Big|_{\lambda_0}\omega\bigl(v,\tilde v(\lambda)\bigr),
\end{equation}
$\tilde v$ being any smooth curve with $\tilde v(\lambda)\in\Lm(\lambda)$ and $\tilde v(\lambda_0)=v$. Although $\omega$ is skew-Hermitian, \eqref{eq:crossdef} is real, the diagonal of a Hermitian form on $\Lm(\lambda_0)\cap L$ \cite[Lemma~22]{Gontier23}. For a moving
$L(\lambda)$, the crossing form is the difference of those of $\Lm$ and of $L$. A crossing
is regular when this form is nondegenerate, and the Maslov index $\muRS$ is the sum of the signatures of
the crossing forms at regular crossings, with halves at the endpoints, extended to all paths by homotopy invariance \cite{BBF98};
the orientation relative to $\sfl$ is discussed after Theorem~\ref{thm:three}. It is an integer when the
endpoints are transverse to $L$, and lies in $\frac12\Z$ otherwise. It is
invariant under homotopies fixing the endpoints and additive under concatenation. Along a
\emph{loop} $L(t)$ the count \eqref{eq:unsplit} of counterclockwise passages for $W=U_{L_0}^{-1}U_{L(t)}$ is
$\wind\det U_{L(t)}$, independent of $L_0$: the fractional parts at the two ends agree, so it equals
$\frac1{2\pi}[\arg\det W]$, and $\det U_{L_0}$ is constant. By the computation in the proof of
Lemma~\ref{lem:crossing}, a positive crossing form is a clockwise passage, so that
$\muRS(L(\cdot);L_0)=-\wind\det U_{L(t)}$. This is the mechanism of
Corollary~\ref{cor:loop}.
Finally \emph{spectral flow}, Definition~\ref{def:sf}, equals a Maslov index under the correspondence
of \cite{BBF98,Gontier23}, up to the orientation convention; Theorem~\ref{thm:three} is an instance, with
$\sfl=-\muRS$.

\paragraph{Boutet de Monvel theory.} The main references are \cite{BdM71} as well as  Grubb \cite{Grubb96} and Schrohe \cite{Schrohe01}.

\paragraph{Operator classes.} The resolvent of an elliptic realization is an operator of Boutet de Monvel's class: the
restriction $Q_+=r^+Qe^+$ of an interior $\psid$, plus a \emph{singular Green} operator, an operator-valued
$\psid$ along $\Gam$ whose kernels decay in $y+y'$. Trace operators of the calculus are $\sum_{j<p}B_j\gamma_j$, with $B_j$ $\psid$s on $\Gam$, plus an integral term $T'$ of class zero \cite[(1.2.17)--(1.2.18)]{Grubb96}; boundary conditions are those with $T'=0$. The boundary part $\cG(z)$ of \eqref{eq:res} differs from the singular Green
part by $\psid$ differences supported on the collar and on the $x$-dependent region. A scalar trace operator has order $m$ when each $B_j$ has order $\le m-j$; the
\emph{row orders} $m_i$ of $T$ are the orders of its $\kp$ scalar rows, so that
$T:H^p(\Om)\to\prod_iH^{p-m_i-1/2}(\Gam)$. The order $d$ of a differential boundary condition in
Definition~\ref{def:local} is instead the tangential order $\max_j\operatorname{ord}B_j$, the degree in $\xi$ of
the symbol matrix of $T$.

\paragraph{Symbols, ellipticity, resolvent.} A boundary problem $(H,T)$ has an interior principal symbol on
$S^*\Om$ and an operator-valued \emph{boundary principal symbol} on $S^*\Gam$ \cite{Grubb96,Schrohe01}, the model problem on the inward normal
obtained by freezing tangential data. Ellipticity is invertibility of both. Invertibility of the
second is the Shapiro--Lopatinski condition, which is \eqref{eq:SLdef} in transversality form.

From the calculus the paper uses: the a priori estimate \cite[Prop.~1.4.2]{Grubb96}, with the order reduction of rows of \cite[\S1.4]{Grubb96}, made uniform on $\Gam=\R$ in Lemma~\ref{lem:ellreg}; the adjoint of a normal realization \cite[Thms.~1.6.9, 1.6.11]{Grubb96}, in Propositions~\ref{prop:SA} and \ref{prop:sharpid}; the bounded right inverse $R_T$, a hypothesis of section~\ref{sec:obs}, which normal trace operators have \cite[Prop.~1.6.5, Rem.~1.6.5$'$]{Grubb96} and which is constructed below for the rows $\gamma_0+iB\gamma_1$, used in Propositions~\ref{prop:SA} and \ref{prop:normres}, Lemma~\ref{lem:CT} and Theorem~\ref{thm:local}; and, for $x$-independent problems, uniform parameter-ellipticity in $\xi$ \cite[Def.~3.1.4, Thm.~3.2.13]{Grubb96}, used in Remark~\ref{rem:escape}. From these, the decay of the kernel of $f(H_L)$ is proved in Theorem~\ref{thm:sG}, with the Helffer--Sj\"ostrand formula and
Combes--Thomas estimates. Outside the calculus, the endpoint limits of Lemmas~\ref{lem:endpoints} and
\ref{lem:gro} rest on exponential dichotomies of slowly varying systems \cite[Lecture~6, Prop.~1]{Coppel} and the associated right inverses \cite[Lecture~3]{Coppel}, with the smoothness in the parameter proved in Lemma~\ref{lem:endpoints}; this is vacuous when the coefficients are $y$-independent.

\paragraph{A right inverse for $\gamma_0+iB\gamma_1$.} Let $p=2$, $T=\gamma_0+iB\gamma_1$ with $B$ first order along $\Gam$, uniformly elliptic for $|\xi|\ge1$, as for $B=sD_x$, $s\ne0$, and $B=\tfrac12(bD_x+D_xb)$, $|b|\ge b_0>0$; then $\partial\cH=H^{-1/2}(\Gam)$. Such rows are not normal in the strict sense of \cite{Grubb96}, the coefficient of $\gamma_1$ being of order one and not invertible near $\xi=0$. Let $K_0,K_1$ be bounded lifts into $H^2(\Om)$ with $(\gamma_0,\gamma_1)K_0g=(g,0)$ and $(\gamma_0,\gamma_1)K_1h=(0,h)$, and $Q$ a parametrix of $iB$ of order $-1$ with $iBQ=\mathbf 1-S_0$, $S_0$ smoothing \cite[\S18.1]{Hormander3}. Then $R_T:=K_1Q+K_0S_0$ satisfies $TR_T=iBQ+S_0=\mathbf 1$ and is bounded $H^{-1/2}(\Gam)\to H^2(\Om)$, with a bound depending only on the ellipticity constant and finitely many $C_b^\infty$ seminorms of the symbol of $B$.

\end{document}